\documentclass[10pt]{article}
\usepackage{graphicx} 

\title{Concise Clique Covers for Parsimonious Graph Algorithms}

\author{
Ahammed Ullah\textsuperscript{\dag} \quad
Alex Pothen\textsuperscript{\ddag}
}
\date{November 2026}

\usepackage[margin=0.9in]{geometry}

\usepackage{float}
 
\usepackage{booktabs}
\usepackage{multirow}
\usepackage{multicol}
\usepackage{subcaption}
\usepackage{enumitem}

\usepackage{tabularx}

\usepackage[dvipsnames]{xcolor}
\usepackage[many]{tcolorbox}

\usepackage[framemethod=TikZ]{mdframed}

\usepackage{bm}
\usepackage{amsmath}
\usepackage{amsthm}
\usepackage{amsfonts}
\usepackage{amssymb}

\usepackage{newpxtext}

\usepackage[T1]{fontenc}
\usepackage[utf8]{inputenc}

\usepackage{thmtools,thm-restate}

\newtheorem*{theorem*}{Theorem}

\newtheorem{definition}{Definition}[section]

\newtheorem{lemma}[definition]{Lemma}
\newtheorem{theorem}[definition]{Theorem}

\newtheorem{remark}[definition]{Remark}

\usepackage{hyperref}

\hypersetup{
    colorlinks=true,
    linkcolor=Blue,
    filecolor=magenta,      
    urlcolor=Sepia,
    citecolor=Sepia,
    pdftitle={Concise Clique Covers for Parsimonious Graph Algorithms},
    pdfauthor={Ahammed Ullah, Alex Pothen},
    pdfpagemode=FullScreen,
    }

\usepackage{tikz}
\usetikzlibrary{arrows.meta, positioning}
\usetikzlibrary{shapes.multipart, fit}
\usetikzlibrary{calc}
\usetikzlibrary{patterns}
\usetikzlibrary{decorations.pathmorphing}

\tikzset{
  implies/.style={
    thin,
    double,
    double distance=2.0pt,
    -{Implies[angle'=60,length=8pt]}
  }
}

\usepackage{centernot}
\usepackage{pifont} 
\usepackage{longtable}

\usepackage{pgfplots}
\usepackage{pgfplotstable}
\usepgfplotslibrary{groupplots}

\pgfplotsset{compat=1.11,
    /pgfplots/ybar legend/.style={
    /pgfplots/legend image code/.code={%
       \draw[##1,/tikz/.cd,yshift=-0.25em]
        (0cm,0cm) rectangle (3pt,0.8em);},
   },
}

\newcommand{\hyref}[2]{\hyperref[#2]{#1~\ref*{#2}}}

\newcommand{\aref}[2]{\hyperref[#2]{#1}}

\newcommand{\fref}[1]{\hyref{Figure}{#1}}

\newcommand{\lemref}[1]{\hyref{Lemma}{#1}}

\newcommand{\A}{\mathcal{A}}

\newcommand{\B}{\mathcal{B}}

\newcommand{\C}{\mathcal{C}}

\newcommand{\D}{\mathcal{D}}

\newcommand{\F}{\mathcal{F}}

\newcommand{\Lcal}{\mathcal{L}}

\newcommand{\ts}[1]{\operatorname{\sigma\br{#1}}}

\newcommand{\M}{\mathcal{M}}

\newcommand{\N}{\mathcal{N}}

\newcommand{\T}{\mathcal{T}}

\newcommand{\inline}[1]{\ensuremath{#1}}

\newcommand{\brsq}[1]{\left[#1\right]}
\newcommand{\brcu}[1]{\left\{#1\right\}}
\newcommand{\br}[1]{\left(#1\right)}

\newcommand{\etaf}[2]{\eta\br{#1,#2}}

\newcommand{\ccov}[2]{\operatorname{CrossCover}\br{#1,#2}}

\DeclareMathOperator*{\argmax}{arg\,max}
\DeclareMathOperator*{\argmin}{arg\,min}

\newcommand{\parent}[1]{\pi\br{#1}}

\newcommand{\card}[1]{\left|#1\right|}

\newcommand{\bigO}[1]{\mathcal{O}\br{#1}}
\newcommand{\bigOmega}[1]{\Omega\br{#1}}
\newcommand{\bigTheta}[1]{\Theta\br{#1}}

\newcommand{\smallO}[1]{o\br{#1}}

\newcommand{\bigOtilde}[1]{\widetilde{\mathcal{O}}\br{#1}}

\let\emptyset\varnothing

\renewcommand{\epsilon}{\varepsilon}
\begin{document}

\maketitle
\renewcommand{\thefootnote}{\fnsymbol{footnote}}
\footnotetext[2]{Texas A\&M University, College Station, TX, USA; most of this work was conducted at Purdue University.}
\footnotetext[3]{Purdue University, West Lafayette, IN, USA}

Large cliques or overlapping collections of cliques can make adjacency lists highly redundant, even in sparse graphs. We exploit this structure through clique-cover representations that support graph computation directly in compressed form.
We develop \emph{cover-parsimonious} algorithms that, given an edge clique cover \inline{\C} of total size \inline{\ts{\C}:=\sum_{C_\ell \in \C} \card{C_\ell}}, use \inline{\bigOtilde{\ts{\C}}} time and \inline{\bigO{\ts{\C}}} space. Our results include breadth-first and depth-first search forests, a \inline{3/2}-approximation of maximum matching, and a \inline{2}-approximation of minimum weighted vertex cover, all in \inline{\bigO{\ts{\C}}} time. We also give a \inline{\br{3/2+\epsilon}}-approximation of maximum vertex-weighted matching in \inline{\bigOtilde{\ts{\C}}} time for any fixed \inline{\epsilon >0}.

The computational benefits of these algorithms depend on obtaining small covers. We prove that, unless \inline{\mathsf{P}=\mathsf{NP}}, there is no polynomial-time \inline{n^c}-approximation for minimizing total cover size on \inline{n}-vertex graphs, for some constant \inline{c >0}.
We therefore introduce \emph{concise clique covers}, in which no two cliques can be merged into a single clique and the total size is at most linear in the number of graph edges. For an \inline{n}-vertex, \inline{m}-edge graph of degeneracy \inline{d}, we construct a concise cover of total size \inline{\bigO{\min\brcu{m, n\sqrt{\mu}}}} in \inline{\bigOtilde{d^2 \cdot \min\brcu{m, \mu}}} time and \inline{\bigO{m}} space, where \inline{\mu := \binom{n}{2}-m+1}. The size bound is sublinear in \inline{m} whenever \inline{\mu=\smallO{n^2}}.

We evaluate our methods on 48 real and synthetic graphs with up to 24 billion edges. Our two concise-cover constructions achieve geometric-mean compression ratios of 731 and 562 relative to adjacency lists, with the largest ratio exceeding four orders of magnitude. Despite the approximation hardness, the resulting covers have near-optimal total sizes.
For every one of the 240 graph--algorithm pairs involving five algorithms, our cover-parsimonious implementations run faster than their counterparts operating on compressed adjacency lists produced by WebGraph and Zuckerli.
The median speedups over WebGraph and Zuckerli are 4,700 and 4,900, respectively, with the largest exceeding six orders of magnitude. The corresponding median peak-memory reduction factors are 4.9 and 4.5, with the largest exceeding two orders of magnitude.
\section{Introduction}

Many real-world graphs contain strong local interaction patterns that induce large cliques or overlapping collections of cliques.
Such structure arises naturally in affiliation data, for example, among users participating in the same online discussion or reviewing the same product. 
Affiliation-derived graphs are common in graph repositories; examples include coAuthorsDBLP and coPapersDBLP in the SuiteSparse Matrix Collection~\cite{davis2011university}. 
Clique structure can also arise when underlying group interactions are observed only through pairwise interactions, particularly when group interactions are difficult or costly to measure directly~\cite{wang2024graphs}.

A clique of order \inline{k} induces \inline{\binom{k}{2}} edges, while standard graph representations such as adjacency lists store these edges individually and therefore cannot exploit the redundancy created by large or overlapping cliques. Graph compression may reduce this redundancy in storage, but compression alone does not imply a corresponding reduction in algorithmic cost. Computation may still require decompression, repeated decoding, or work proportional to the number of graph edges rather than to the compressed representation size. These limitations motivate graph representations that encode clique structure compactly while supporting algorithms whose time and space scale with the representation size.

A \emph{clique cover} (or \emph{edge clique cover}) of a graph \inline{G=\br{V,E}} is a collection \inline{\C} of cliques such that every edge of \inline{G} belongs to at least one clique. Its \emph{total size} is \inline{\ts{\C}:=\sum_{C_\ell \in \C} \card{C_\ell}}. 
Clique-cover representations can be substantially smaller than standard representations even for globally sparse graphs. For example, consider an \inline{n}-vertex graph consisting of \inline{n/k} disjoint cliques of order \inline{k}. An adjacency-list representation has size \inline{\bigTheta{nk}}, whereas representing each clique once gives a clique cover of total size \inline{\bigTheta{n}}, a reduction by a factor \inline{\bigTheta{k}}, yet the graph has edge density only \inline{\br{k-1}/\br{n-1}= \smallO{1}} whenever \inline{k=\smallO{n}}.

A straightforward enumeration of vertex pairs within each cover clique uses \inline{\bigO{\ts{\C}}} space, but takes \inline{\bigO{\sum_{C_\ell \in \C}\card{C_\ell}^2} = \bigO{n\ts{\C}}} time. We call a graph algorithm \emph{cover-parsimonious} if, given \inline{\C}, it uses \inline{\bigOtilde{\ts{\C}}} time and \inline{\bigO{\ts{\C}}} space, where \inline{\bigOtilde{\cdot}} hides polylogarithmic factors. \hyref{Table}{tab:cp_algos} summarizes our cover-parsimonious algorithmic results.

\begin{table}[h]
    \centering    
    \setlength{\tabcolsep}{10pt}
    \renewcommand{\arraystretch}{1.1}
    \caption{Cover-parsimonious algorithms. Given a cover
    \inline{\C}, all algorithms use \inline{\bigO{\ts{\C}}} space.}
    \begin{tabularx}{\linewidth}{@{}lX@{}}
        \toprule
        \textbf{Running time} & \textbf{Problems} \\
        \midrule

        \inline{\bigO{\ts{\C}}}
        & \inline{3/2}-approximation of maximum matching \\

        & Breadth-first and depth-first search forests; maximal independent set \\

        & \inline{2}-approximation of minimum weighted vertex cover \\
        
        \midrule

        \inline{\bigOtilde{\ts{\C}}}
        & \inline{\br{3/2+\epsilon}}-approximation of maximum
        vertex-weighted matching, for any constant \inline{\epsilon>0} \\

        & \inline{\br{\Delta+1}}-vertex coloring
          (expected time) \\
        
        \bottomrule
    \end{tabularx}
    \label{tab:cp_algos}
\end{table}

To obtain small clique-cover representations of graphs, one may consider minimizing either the total size \inline{\ts{\C}} or the number of cliques \inline{\card{\C}}. Both objectives are NP-hard \cite{kou1978covering, orlin1977contentment, ullah2022computing}, and strong approximation hardness is known for the cardinality objective \cite{lund1994hardness}. We establish analogous hardness for the total-size objective.

\begin{theorem}
\label{thm:aopt_apx}    
Unless \inline{\mathsf{P}=\mathsf{NP}}, there exists a constant \inline{\epsilon >0} such that  no polynomial-time \inline{n^\epsilon}-approximation exists for minimizing \inline{\ts{\C}} over all clique covers \inline{\C} of an \inline{n}-vertex graph.
\end{theorem}

We therefore identify a class of clique covers that we call \emph{concise clique covers}. 
These covers satisfy two requirements: every pair of distinct cliques has a non-edge between their vertices, and their total size is \inline{\ts{\C} = \bigO{m}}, where \inline{m} is the number of edges. We then explore progressively stronger notions of conciseness based on the number of non-edges, culminating in the following result.

\begin{theorem}
\label{thm:cpa}
Let \inline{G=(V,E)} be a graph with \inline{n} vertices, \inline{m} edges, and define \inline{\mu := \binom{n}{2}-m+1}. Let \inline{d} denote the degeneracy of \inline{G}, that is, the maximum minimum degree over all subgraphs of \inline{G}. Then \inline{G} admits a concise clique cover \inline{\C} with
    \[\ts{\C}=\bigO{\min\brcu{m, n\sqrt{\mu}}},\]
which can be computed in \inline{\bigOtilde{d^2 \cdot \min\brcu{m, \mu}}} time using \inline{\bigO{m}} space.        
\end{theorem}

This gives clique-cover representations whose total size is provably sublinear in \inline{m} whenever \inline{\mu= \smallO{n^2}}. While this bound guarantees strong compression for sufficiently dense graphs, our concise constructions also achieve substantial compression on real-world sparse graphs. The theorem's running time bound is nearly linear in \inline{m} for constant degeneracy. The construction uses an admissibility relation together with a batching mechanism that maintains \inline{\bigO{m}} space. 
In our evaluation, the admissibility-based constructions exhibit scaling consistent with \inline{\bigOtilde{dm}}.

We evaluate our constructions and algorithms on 48 real and synthetic graphs with 16 million to 24 billion edges.
We measure compression by the ratio \inline{2m/\sigma},
where \inline{\sigma} is the total number of vertex entries
in a clique cover or biclique partition.
Our concise constructions FPA and CPA achieve geometric-mean compression ratios of 731 and 562, respectively, compared with 3.1 and 4.8 for their non-concise baselines. FPA achieves a compression ratio as high as 21,762. Despite the strong approximation hardness established in \hyref{Theorem}{thm:aopt_apx}, FPA produces covers within a factor of 2 of optimal total size on 47 of the 48 graphs. In contrast, the Feder--Motwani (FM) biclique method \cite{FederMotwani1995} achieves a geometric-mean compression ratio of only 1.21 and expands the representation on 32 of the 48 graphs.

Given the poor compression achieved by FM, we compare five representative cover-parsimonious algorithms with corresponding algorithms operating on compressed adjacency lists produced by WebGraph~\cite{boldi2004webgraph} and Zuckerli~\cite{versari2020zuckerli}, two state-of-the-art graph compression schemes supporting access to individual adjacency lists. For all 240 graph--algorithm pairs, the cover-parsimonious implementations are faster than both their WebGraph- and Zuckerli-based counterparts, with speedups ranging from 9.2 to approximately 1.8 million. The median speedups are approximately 4,700 and 4,900 against WebGraph and Zuckerli, respectively. The median peak-memory reduction factors are 4.9 and 4.5, respectively, with maxima of approximately 200 and 176.

The clique-cover representations also require less storage than WebGraph on all 48 graphs and Zuckerli on 47. The geometric-mean size reduction factors are 41.5 and 36.7, respectively, reaching maxima of 965 and 851.
Clique covers are faster to construct on every graph: WebGraph requires 2.0--21.5 times as long, and Zuckerli 5.8--58.4 times as long.

\paragraph{Organization.} \hyref{Section}{sec:prelim} introduces the necessary background. \hyref{Section}{sec:concise} discusses concise covers and their constructions. In \hyref{Section}{sec:cp_algo}, we present the design and analyses of representative cover-parsimonious algorithms. \hyref{Section}{sec:evals} reports our evaluation results. We conclude in \hyref{Section}{sec:conc} with a discussion of future research directions.
\section{Preliminaries}
\label{sec:prelim}
 We assume that all graphs are simple, undirected, and have no isolated vertices. For a graph \inline{G=(V,E)}, let \inline{n:=|V|} denote the number of vertices and \inline{m:=|E|} denote the number of edges. For two distinct vertices \inline{u, v \in V}, we write \inline{\brcu{u,v}} for their unordered pair; when \inline{\brcu{u,v} \in E}, we call it an edge, and when \inline{\brcu{u,v} \notin E}, we call it a non-edge. Let \inline{\mu -1} denote the number of non-edges \inline{\binom{n}{2}-m}.
 
For a vertex \inline{v\in V}, let \inline{N(v)} denote the (open) neighborhood of \inline{v} in \inline{G} and let \inline{N[v] := N(v) \cup \brcu{v}} denote its closed neighborhood. Let \inline{\Delta := \max_{v \in V} \card{N(v)}} and \inline{\delta := \min_{v \in V} \card{N(v)}} denote the maximum and minimum degree of \inline{G}, respectively.

We call \inline{H=(V_H,E_H)} a \emph{subgraph} of \inline{G} if \inline{V_H \subseteq V} and \inline{E_H \subseteq E}, and every edge of \inline{E_H} has both endpoints in \inline{V_H}. We call \inline{H} an \emph{induced subgraph} of \inline{G} if \inline{E_H} consists of all edges in \inline{E} with both endpoints in \inline{V_H}. In this case, we say that \inline{V_H} induces the subgraph \inline{H} in \inline{G}, and we denote it by \inline{G\brsq{V_H}}.

Let $\delta_H$ denote the minimum degree of a nonempty subgraph $H \subseteq G$. Then the \emph{degeneracy} of \inline{G} is \inline{d:=\max_{H \subseteq G} \delta_H}.
A \emph{degeneracy ordering} of \inline{G} is an ordering of its vertices such that each vertex \inline{v} has at most \inline{d} neighbors that appear before \inline{v}. A \emph{degeneracy coloring} is a first-fit greedy coloring along a degeneracy ordering, which uses at most $d+1$ colors. We have \inline{d \leq \Delta \leq n-1} and \inline{m \leq d\br{n-1}}. We also have \inline{d^2=\bigO{m}}, since \inline{G} has a subgraph \inline{H=\br{V_H, E_H}} with \inline{\delta_H = d}, so \inline{\card{V_H} \geq d+1} and \inline{\card{E_H} \geq \binom{d+1}{2}}. Furthermore, no clique in \inline{G} can contain more than \inline{d+1} vertices, since otherwise $G$ would have a subgraph $H$ with \inline{\delta_H > d}.

A $\rho$-approximation algorithm computes a solution whose value is within a factor $\rho$ of the \emph{optimal}. The factor $\rho$ is called the \emph{approximation ratio}. We assume $\rho \geq 1$ for both maximization and minimization problems. Thus, for maximization, a $\rho$-approximation guarantees a solution whose value is at least $\frac{1}{\rho}$ times the optimal.
\section{Concise Clique Covers}
\label{sec:concise}

This section introduces concise clique covers and describes efficient algorithms for constructing them. \hyref{Section}{subsec:landscape} formalizes a hierarchy of clique-cover minimality notions and uses this hierarchy to define concise covers and their subclasses. \hyref{Section}{subsec:concise_lzf} and \hyref{Section}{subsec:concise_bcc} explore progressively stronger notions of conciseness and their constructions. Proofs and supporting details deferred from the exposition appear in \hyref{Appendix}{sec:concise_deferred}.

\subsection{Minimality Landscape and Concise Covers}
\label{subsec:landscape}

Recall that a \emph{clique cover} of a graph is a collection of cliques such that every edge is contained in at least one of them. For representing graphs compactly, we may consider the following optimization objectives.

\begin{definition}[Cardinality-Optimal]
\label{def:cc_copt}
A clique cover \inline{\C} is \emph{cardinality-optimal} if it has the minimum possible number of cliques among all clique covers of the graph; that is, if \inline{\C^\prime} is any clique cover of the graph, then \inline{\card{\C} \leq \card{\C^\prime}}.
\end{definition}

\begin{definition}[Assignment-Optimal]
\label{def:cc_aopt}
A clique cover \inline{\C} is \emph{assignment-optimal} if the total number of vertex assignments across all cliques is minimum among all clique covers of the graph; equivalently, \inline{\ts{\C}} is minimum over all clique covers.
\end{definition}

Both objectives are NP-hard \cite{kou1978covering, orlin1977contentment, ullah2022computing}, and the cardinality-optimal objective is inapproximable within a factor of \inline{n^\epsilon} for some \inline{\epsilon > 0} unless \inline{\mathsf{P}=\mathsf{NP}} \cite{lund1994hardness}. In \hyref{Appendix}{subsec:apx_aopt}, we prove \hyref{Theorem}{thm:aopt_apx}, which establishes the same approximation hardness for the assignment-optimal objective. These inapproximability results motivate us to explore local notions of minimality.

\begin{figure}[H]
    \centering
    \begin{tikzpicture}[
  node distance=1.0cm and 1.0cm,  
  np/.style={
    draw=red!60!black, thin, fill=red!10,
    rounded corners, minimum width=2.8cm, align=center,
    label={[red!80!black]left:{\textcolor{Sepia}{\Large{\ding{74}}}}}
  },
  poly/.style={
    draw=blue!60!black, thin, fill=blue!10,
    rounded corners, minimum width=2.8cm, align=center,
    label={[blue!80!black]left:{\textcolor{Blue}{\Large{\ding{51}}}}}
  }
]

\node[np] (assignopt) {Assignment-optimal};
\node[poly] (assign) [below=of assignopt] {Assignment-minimal};
\node[np] (card) [right=2cm of assignopt] {Cardinality-optimal};
\node[poly] (support) [below=of assign] {Support-minimal};
\node[poly] (composition) [below=of card] {Composition-minimal};
\node[poly] (inclusion) [below=2.0cm of composition] {Inclusion-minimal};

\draw[implies] (assignopt) -- (assign);
\draw[implies] (assign) -- (support);
\draw[implies] (card) -- (support);
\draw[implies] (card) -- (composition);
\draw[implies] (support) -- (inclusion);
\draw[implies] (composition) -- (inclusion);

\end{tikzpicture}
    \caption{Partial order of minimality properties of clique covers of a graph under logical implication. Properties marked with \ding{51} are achievable in polynomial time, and those marked with \ding{74} are NP-hard to optimize.}
    \label{fig:partial_order}
\end{figure}

\begin{definition}[Inclusion-Minimal]
\label{def:cc_inclusion}
A clique cover $\C$ is \emph{inclusion-minimal} if no clique in the cover is contained in another. That is, for each $C_i \in \C$ there does not exist $C_{j} \in \C$ with $j \ne i$ such that $C_i \subseteq C_j$.
\end{definition}

\begin{definition}[Support-Minimal]
\label{def:cc_support}
A clique cover $\C$ is \emph{support-minimal} if every clique covers at least one edge that no other clique covers. That is, for each $C_i \in \C$, there exists an edge $\brcu{u,v} \subseteq C_i$ such that for all $C_{j} \in \C$ with $j \ne i$, the edge $\brcu{u,v}$ is not contained in $C_j$.
\end{definition}

\begin{definition}[Composition-Minimal]
\label{def:cc_compose}
A clique cover $\C$ is \emph{composition-minimal} if it does not contain any subset of two or more cliques whose union induces a clique in the graph. Equivalently, for each pair of distinct cliques \inline{C_i, C_j \in \C}, there exists a non-edge \inline{\brcu{u,v}} with \inline{u \in C_i} and \inline{v \in C_j}.
\end{definition}

\begin{definition}[Assignment-Minimal]
\label{def:cc_amin}
A clique cover \inline{\C} is \emph{assignment-minimal} if no vertex can be removed from any of its cliques without leaving some edge uncovered. That is, for all $C_i \in \C$ and all $v\in C_i$, replacing $C_i$ with $C_i\setminus \{v\}$ in $\C$ yields a family that does not cover all edges.
\end{definition}

We organize these minimality notions into a partial order under logical implication, summarized in \fref{fig:partial_order} (the implications shown are the only ones that hold, as detailed in \hyref{Appendix}{subsec:landscape_dd}). Among the four local minimality notions, all except composition-minimality can be achieved by taking every edge as a distinct clique in the cover. Therefore, assignment-minimality (or any of its superclasses) alone is not sufficient to guarantee nontrivial compression in total size. On the other hand, composition-minimal covers impose strong structural requirements, but are too broad, as shown below.

\begin{restatable}{lemma}{LemmaCmExp}
\label{lemma:cm_exp}
There exist \inline{n}-vertex graphs admitting composition-minimal clique covers with \inline{2^{\bigTheta{n}}} cliques. 
\end{restatable}

We therefore restrict attention to composition-minimal covers whose total size is at most linear in the number of edges.

\begin{definition}[Concise Clique Cover]
\label{def:concise}
For a graph $G=(V,E)$ with \inline{m} edges, a clique cover \inline{\C} of \inline{G} is called a \emph{concise clique cover} if (1) \inline{\C} is composition-minimal, and (2) \inline{\ts{\C} = \bigO{m}}.
\end{definition}

The requirement \inline{\ts{\C} = \bigO{m}} is strictly stronger than requiring only \inline{\card{\C} = \bigO{m}}. The following lemma separates these two requirements.

\begin{restatable}{lemma}{LemmaScSize}
\label{lemma:sc_size}
There exists a family of graphs with \inline{n} vertices and \inline{m} edges that admits composition-minimal clique covers \inline{\C} with \inline{\card{\C} = \bigTheta{m}}, but \inline{\ts{\C} = \bigTheta{nm}}.
\end{restatable}

We now introduce a stronger notion of conciseness as follows.

\begin{definition}[\inline{\phi}-Concise Clique Cover]
\label{def:concise_phi}
Let \inline{\phi} be a nonnegative graph function such that \inline{\phi\br{G} =\smallO{m}} for some infinite family of graphs. A concise clique cover \inline{\C} of \inline{G} is \inline{\phi}\emph{-concise} if \inline{\ts{\C} = \bigO{\phi\br{G}}}.
\end{definition}

We focus on the instantiations \inline{\phi\br{G}=n\mu} and \inline{\phi\br{G}=n\sqrt{\mu}}, giving \inline{\br{n\mu}}-concise and \inline{\br{n\sqrt{\mu}}}-concise covers, respectively. The following separation shows that concise, \inline{\br{n\mu}}-concise, and \inline{\br{n\sqrt{\mu}}}-concise covers are progressively stronger notions.

\begin{restatable}{lemma}{LemmaConciseSep}
\label{lemma:concise_sep}
There exists a family of graphs admitting two concise covers \inline{\C_1} and \inline{\C_2} such that \inline{\C_1} is not \inline{\br{n\mu}}-concise, while \inline{\C_2} is \inline{\br{n\mu}}-concise but not \inline{\br{n\sqrt{\mu}}}-concise.
\end{restatable}

\subsection{\texorpdfstring{\inline{\br{n\mu}}}{n mu}-Concise Covers}
\label{subsec:concise_lzf}

Lov{\'a}sz gave a tight upper bound on the number of cliques in a cardinality-optimal clique cover in terms of the number of non-edges of the graph~\cite{lovasz1968covering} (in our notation, this quantity is \inline{\mu - 1}). We first restate the construction underlying this bound, and then modify it to obtain concise covers of total size \inline{\bigO{\min\brcu{m, n\mu}}}.

\begin{definition}[Lov{\'a}sz's Construction]
\label{def:lov_con}
Given a graph \inline{G=(V,E)}, construct a family \inline{\F=\brcu{F_1, \dots, F_p}} of pairwise disjoint cliques by repeatedly peeling a maximum clique. Let \inline{F_1} be a maximum clique in \inline{G}. Having chosen \inline{F_1,\dots, F_i}, let \inline{F_{i+1}} be a maximum clique of \inline{G[V\setminus\br{F_1 \cup  \dots  \cup F_i}]}, and continue until \inline{\F} partitions \inline{V}. 

For \inline{v \in F_k}, define \inline{T_v:=\brcu{\ell <k \mid F_\ell \cap N(v) \ne \emptyset}}.
\begin{enumerate}
    \item Initialize \inline{\C :=\brcu{F_\ell \in \F \mid \card{F_\ell} \geq 2}}.
    \item For each \inline{v\in V} and \inline{\ell \in T_v}, add the clique \inline{C_{v,\ell}:= \br{F_\ell \cap N(v)} \cup \brcu{v}} to \inline{\C}.
\end{enumerate}
\end{definition}

Lov{\'a}sz's construction does not enforce composition-minimality, and its total size can be polynomially larger than that of a composition-minimal cover. For example, consider its output on the family \inline{\brcu{G_n}_{n\geq2}}, where \inline{G_n} consists of three disjoint cliques \inline{A}, \inline{B}, and \inline{U} of order \inline{n}, with every vertex of \inline{U} adjacent to all of \inline{A \cup B}, and with no edges between \inline{A} and \inline{B}. If \inline{F_1=A \cup U}, then \inline{F_2=B}, and the construction also creates \inline{C_{b_i,1} := U \cup \brcu{b_i}} for each \inline{b_i \in F_2}, so \inline{F_2 \cup C_{b_i,1} = B \cup U} is a clique for each \inline{i}. Thus the output is not composition-minimal. The same argument applies symmetrically if \inline{F_1 = B \cup U}. Moreover, the construction produces covers of total size \inline{\bigTheta{n^2}}, while the graph admits a composition-minimal cover of total size \inline{\bigTheta{n}}, namely the two cliques \inline{A \cup U} and \inline{B \cup U}.

We observe that in Lov{\'a}sz's construction, the peeling process can be replaced by a first-fit coloring of the complement \inline{\overline{G}}. We next introduce a local construction for concise covers, which we will use to make this first-fit variant of Lov{\'a}sz's construction concise.

\begin{definition}[Local-Augmentation (LA) Construction]
\label{def:la_con}
Let $\C$ be a clique cover of a subgraph of $G$. An edge \inline{\brcu{u,v}} is \emph{uncovered} with respect to \inline{\C} if no clique in \inline{\C} contains both \inline{u} and \inline{v}. 

For an uncovered edge $\brcu{u,v}$, update \inline{\C} as follows:
\begin{enumerate}
    \item If there exists $C_{\ell} \in \C$ such that $C_{\ell} \cup \brcu{u,v}$ is a clique, choose exactly one such \inline{C_\ell} and replace it by $C_{\ell} \cup \brcu{u,v}$.
    \item Otherwise, introduce a new label $\ell$ and add the clique $C_{\ell} := \brcu{u,v}$.
\end{enumerate}

A clique cover is \emph{LA-concise}\ if it can be obtained from the empty cover by repeatedly applying this rule to uncovered edges, always applying Step~1 whenever possible.
\end{definition}

The following lemma verifies the conciseness of this construction.

\begin{restatable}{lemma}{LemmaLac}
\label{lemma:lac}
Let \inline{\C} be a family of cliques obtained by the local-augmentation construction, starting with an empty cover and processing all edges of \inline{G}. Then \inline{\C} is a concise cover of \inline{G}.
\end{restatable}

A naive implementation of the LA construction can be prohibitively expensive: for each uncovered edge, it may scan all current cliques to find one that can be augmented. This requires \inline{\bigTheta{m^2}} time even on the star \inline{K_{1,m}}, since each edge creates a distinct clique and the \inline{i}-th edge may scan the preceding \inline{\br{i-1}} cliques.
To avoid such repeated scans, we maintain the relation between vertices and cliques that can still be augmented. A vertex \(v\) is \emph{admissible} for a clique \inline{C_\ell} (or equivalently, \inline{C_\ell} can admit $v$) if \inline{C_\ell \subseteq N[v]}. Thus an uncovered edge \inline{\brcu{u,v}} can reuse a label \inline{\ell} precisely when both \(u\) and \(v\) are admissible for \inline{C_\ell}. This relation is formalized as follows.

\begin{definition}[Admissibility Sets and Duals]
\label{def:admissible}
For a clique cover $\C=\{C_{\ell}\}_{\ell \in I}$, define \[A_v := \{\ell \in I \mid C_{\ell} \subseteq N[v]\}\]
for each \inline{v \in V}, and
\[D_{\ell} := \{v \in V \mid C_{\ell} \subseteq N[v]\}\]
for each \inline{\ell \in I}. We call \inline{\A:=\brcu{A_v}_{v \in V}} the \emph{admissibility sets} and \inline{\D:=\brcu{D_\ell}_{\ell \in I}} the \emph{admissibility duals}.
\end{definition}

With these structures, Step~1 of the LA construction on an edge \inline{\brcu{u,v}} reduces to finding a label in \inline{A_u \cap A_v}, while \inline{D_\ell} identifies the vertices whose admissibility may change when \inline{C_\ell} grows. Maintaining the complete admissibility relation therefore gives an efficient implementation of the LA updates, but its space usage can be superlinear in the worst case (see \lemref{lemma:admissible_space_lb} in \hyref{Appendix}{subsec:deferred_pfs_cons}).

\phantomsection\label{para:bpa} To maintain linear space, we use the admissibility structures in phases. A phase is initially open and maintains admissibility structures only for the cliques created in that phase. While the phase is open, both steps of the LA construction are allowed. Once \inline{\ts{\D} > \tau m}, for a constant \inline{\tau >0}, the phase is closed: no new clique is introduced, while Step~1 continues to be applied to the remaining uncovered edges whenever possible. At the end of the phase, every remaining uncovered edge is inadmissible to every clique created in that phase. Hence these cliques need not be considered for augmentation in subsequent phases, and their admissibility structures can be discarded.

\begin{figure}[h]
    \centering
    \begin{parbox}{4.8in}{    
        \begin{mdframed}[linewidth=0.5pt, roundcorner=7pt, backgroundcolor=gray!5, frametitle={\underline{\textsc{FirstFit-PA}$\br{G=\br{V,E}, \tau}$}}]          
        \begin{enumerate}            
            \item Initialize \inline{\mathrm{open} \gets \mathrm{true}} and \inline{M_v \gets N\br{v}}, for all \inline{v \in V}.
            \item Let \inline{\F := \brcu{F_{\ell}}_{\ell \in [p]}} be a first-fit greedy coloring of \inline{\overline{G}}, where \inline{F_\ell} denotes the \inline{\ell}th color class in greedy order.
            \item Initialize \inline{\C \gets \brcu{F_\ell \in \F \mid \card{F_\ell} \geq 2}}.
            \item For each \inline{C_{\ell} \in \C} and for each \inline{v \in C_{\ell}}, set \inline{M_v \gets M_v \setminus C_{\ell}}.
            \item Create the admissibility structures \inline{\A} and \inline{\D} for \inline{\C}.
            \item For each \inline{k \geq 2} and each \inline{v \in F_k}, set \inline{T_v \gets \brcu{\, \ell < k \mid F_\ell \cap N(v) \ne \emptyset\,}}.
            \item While some \inline{T_v} is nonempty do    
            \begin{enumerate}                
                \item For each current pair \inline{\br{v,\ell}} with \inline{\ell \in T_v} do
                \begin{enumerate}
                    \item Set \inline{E^\prime \gets \brcu{\brcu{u,v} \mid u \in F_\ell \cap M_v}}.
                    \item \aref{\textsc{AugmentCover}\inline{\br{E^\prime}}}{fig:sub_fpa}.
                    \item If \inline{F_\ell \cap M_v = \emptyset}, then set \inline{T_v \gets T_v \setminus \brcu{\ell}}.
                \end{enumerate}
                \item Set \inline{\mathrm{open} \gets \mathrm{true}, \D \gets \emptyset}, and \inline{A_v\gets \emptyset}, for all \inline{v \in V}.
            \end{enumerate}
            \item Return \inline{\C}.
        \end{enumerate}
        \end{mdframed}        
    }
    \end{parbox}    
    \caption{Concise covers via phased admissibility.}
    \label{fig:algo_fpa}
\end{figure}

\begin{figure}[h]
    \centering
    \begin{parbox}{5.0in}{    
        \begin{mdframed}[linewidth=0.5pt, roundcorner=7pt, backgroundcolor=gray!5, frametitle={\underline{\textsc{AugmentCover}$\br{E^\prime}$}}]          
        \textcolor{teal}{/* Uses and updates the state of \textsc{FirstFit-PA}.*/}
        \vspace{0.2cm}
        \\For each edge \inline{\brcu{u,v} \in E^\prime} with \inline{u \in M_v} do            
        \begin{enumerate}
            \item If \inline{A_u \cap A_v \ne \emptyset} then,
            \begin{enumerate}
                \item Select a label \inline{\ell \in A_u \cap A_v} and set \inline{C_\ell \gets C_\ell \cup \brcu{u,v}}.
                \item Set \inline{M_x \gets M_x \setminus C_\ell}, for \inline{x\in \brcu{u,v}}, and \inline{M_y \gets M_y \setminus \brcu{u,v}}, for \inline{y \in C_\ell}.
                \item Let \inline{D_\ell^\prime := D_\ell \cap N[u] \cap N[v]} and \inline{X \gets D_\ell \setminus D_\ell^\prime}.
                \item Set \inline{D_\ell \gets D_\ell^\prime}, and \inline{A_w \gets A_w \setminus \brcu{\ell}}, for \inline{w \in X}.
            \end{enumerate}
            \item Else if \inline{\mathrm{open} = \mathrm{true}} then,
            \begin{enumerate}
                \item Introduce a new label \inline{\ell} and define \inline{C_\ell := \brcu{u,v}}.
                \item Set \inline{\C\gets \C \cup \brcu{C_\ell}}, \inline{M_u \gets M_u \setminus \brcu{v}}, and \inline{M_v \gets M_v \setminus \brcu{u}}.
                \item Define \inline{D_\ell := N[u] \cap N[v]}.
                \item Set \inline{A_w \gets A_w \cup \brcu{\ell}}, for all \inline{w \in D_\ell}.
                \item If \inline{\ts{\D} > \tau m}, then set \inline{\mathrm{open} \gets \mathrm{false}}.
            \end{enumerate}
        \end{enumerate}        
        \end{mdframed}        
    }
    \end{parbox}    
    \caption{A subroutine of \aref{\textsc{FirstFit-PA}}{fig:algo_fpa}.}
    \label{fig:sub_fpa}
\end{figure}

By processing the edges associated with each clique together, the phased-admissibility scheme can make a given cover construction concise without increasing the number of cliques it produces. \fref{fig:algo_fpa} shows this implemented for the first-fit variant of Lov{\'a}sz's construction. It initializes the cover with the nontrivial first-fit classes, and processes the pairs \inline{\br{v,\ell}}, for \inline{\ell \in T_v}, in phases. For each such pair, the currently uncovered edges between \inline{v} and \inline{F_\ell} are passed together to the subroutine \aref{\textsc{AugmentCover}}{fig:sub_fpa}. If all of them become covered, the pair is removed; otherwise, the pair is reconsidered in the next phase. After each complete pass over the remaining pairs, the admissibility structures of the completed phase are discarded and a new phase begins. The construction terminates when no pair remains. The subroutine \aref{\textsc{AugmentCover}}{fig:sub_fpa} implements the local-augmentation and admissibility updates of the phased-admissibility scheme described above.

Using \lemref{lemma:lac} together with a simpler counting argument for the first-fit variant of Lov{\'a}sz's construction, we obtain the following bounds.

\begin{restatable}{lemma}{LemmaFpaSz}
\label{lemma:fpa_sz}
Algorithm~\aref{\textsc{FirstFit-PA}}{fig:algo_fpa} returns a concise cover \inline{\C} satisfying \[\card{\C} = \bigO{\min\brcu{m, \mu}} \text{ and } \ts{\C}=\bigO{\min\brcu{m,n\mu}}.\]
\end{restatable}

We now analyze time and space use of \aref{\textsc{FirstFit-PA}}{fig:algo_fpa}. We use the following generalization of Lemma~2 of \cite{chiba1985arboricity}.

\begin{restatable}{lemma}{LemmaEdgeSum}
\label{lemma:edgesum}
Let $G=(V,E)$ be a graph of degeneracy $d$ and $f:V \to \mathbb{R}_{\geq 0}$. Then \[\sum_{\brcu{u,v} \in E} \min\brcu{f(u), f(v)} \leq d \cdot \sum_{u \in V} f(u).\]
\end{restatable}

\begin{lemma}
\label{lemma:fpa_cx}
Algorithm~\aref{\textsc{FirstFit-PA}}{fig:algo_fpa} uses \inline{\bigOtilde{d^2 \cdot \min\brcu{m,\mu}}} time and \inline{\bigO{m}} space.
\end{lemma}
\begin{proof}
We first account for the initialization and space usage. Since the nontrivial \inline{F_\ell}'s are vertex-disjoint cliques, we have \inline{\sum_{\ell:\card{F_\ell} \geq 2} \card{F_\ell}^2 = \bigO{m}}. When constructing \inline{T_v}, we can store with each pair \inline{\brcu{v,\ell}} the vertices of \inline{F_\ell \cap N\br{v}}. These sets are disjoint for a fixed \inline{v}, and hence \[\sum_{v} \sum_{\ell \in T_v} \card{F_\ell \cap N\br{v}} \leq \sum_{v \in V}\card{N\br{v}} = \bigO{m}.\]
For the admissibility sets initialized in Step~5, fix \inline{v \in V} and for every initial clique \inline{F_\ell} with \inline{\ell \in A_v}, choose one neighbor of \inline{v} in \inline{F_\ell}. Different admissible classes give different neighbors because the classes are disjoint. Thus \inline{\card{A_v} \leq \card{N\br{v}}} and \inline{\ts{\A} = \ts{\D} = \bigO{m}}. These initial sets \inline{\A} and \inline{\D} can be computed in \inline{\bigO{m}} time by scanning the sets \inline{N\br{v}} and grouping by the first-fit color classes. Therefore, Steps~1--6 can be implemented in \inline{\bigOtilde{m}} time and \inline{\bigO{m}} space. The \inline{\bigO{m}} space bound holds for the remaining steps, since a phase closes if adding a dual of size \inline{\bigO{n}} satisfies \inline{\ts{\D} > \tau m}.

It remains to analyze time spent in Step~7. Suppose Step~7 executes for \inline{r} phases and let \inline{I_t} be the labels created in phase \inline{t}. The initial first-fit labels belong to \inline{I_1}. Thus the sets \inline{I_t} are disjoint. For each \inline{\ell \in I_t}, let \inline{D_\ell^{\br{t}}} be the dual of \inline{C_\ell} when its label enters phase \inline{t}, and define \inline{\D^{\br{t}}:=\brcu{D_\ell^{\br{t}}}_{\ell \in I_t}}. Let \inline{\A^{\br{t}} := \brcu{A_v^{\br{t}}}_{v \in V}}, where \inline{A_v^{\br{t}}=\brcu{\ell \in I_t \mid v \in D_\ell^{\br{t}}}}. Let \inline{\D^*= \bigcup_{t=1}^r \D^{\br{t}}}, and define \inline{A_v^* := \bigcup_{t}A_v^{\br{t}}} and \inline{\A^* := \brcu{A_v^*}_{v\in V}}.

For each label \inline{\ell}, choose a pivot edge \inline{\brcu{u_\ell, v_\ell}}. For an initial clique \inline{F_\ell}, choose any edge of \inline{F_\ell} and for a clique created by \aref{\textsc{AugmentCover}}{fig:sub_fpa}, choose the edge that creates it. These pivot edges are pairwise distinct since the initial ones lie in distinct first-fit classes and every later pivot is uncovered when it is selected. For every label \inline{\ell}, \inline{D_\ell^{\br{t}} \subseteq N[u_\ell] \cap N[v_\ell]}, and therefore \inline{ \card{D_\ell^{\br{t}}} \leq 2 + \min\brcu{\card{N\br{u_\ell}}, \card{N\br{v_\ell}}}}. Using \inline{\card{D_\ell^{\br{t}}} \leq n} and applying \lemref{lemma:edgesum} to the distinct pivot edges with \inline{f\br{v}=\card{N\br{v}}} give
\[\ts{\D^*}=\ts{\A^*} = \bigO{\min\brcu{n\card{\C}, dm}}.\]

During phase \inline{t}, \inline{A_v \subseteq A_v^{\br{t}}}, since in \aref{\textsc{AugmentCover}}{fig:sub_fpa} admissibility of an existing label can only be removed after the label is introduced. Each graph edge is considered at most once in a phase. Thus the total cost of the intersections \inline{A_u \cap A_v} in a phase \inline{t} is at most \inline{\sum_{\brcu{u,v} \in E} \min\brcu{\card{A_u^{\br{t}}, \card{A_v^{\br{t}}}}} \leq d \ts{\A^{\br{t}}}}, by \lemref{lemma:edgesum} with \inline{f\br{v} = \card{A_v^{\br{t}}}}. Summing over all phases gives \inline{\bigOtilde{d\ts{\A^*}}= \bigOtilde{d\ts{\D^*}}} time for all admissibility tests. 

Whenever a clique \inline{C_\ell} is augmented in \aref{\textsc{AugmentCover}}{fig:sub_fpa}, its current dual \inline{D_\ell \subseteq D_\ell^{\br{t}}} is scanned once. Since \inline{C_\ell} can have at most \inline{\bigO{d}} augmentations, all such scans also take \inline{\bigOtilde{d\ts{\D^*}}} time. After Step~5, creating new admissibility sets and duals takes \inline{\bigOtilde{\min\brcu{n\card{\C}, dm}}} time, by the same pivot-edge argument used above. Each augmentation of \inline{C_\ell} uses at most \inline{\bigOtilde{\card{C_\ell}} \subseteq \bigOtilde{d}} time for updating the relevant \inline{M_v}-sets. Therefore, for the entire cover \inline{\C}, these updates use at most \inline{\bigOtilde{d\ts{\C}} \subseteq \bigOtilde{d\ts{\D^*}}} time.

A closed phase \inline{t} satisfies \inline{\ts{\D^{\br{t}}} > \tau m}, since \inline{\ts{\D} > \tau m} and \inline{D_\ell \subseteq D_\ell^{\br{t}}} for each \inline{\ell \in I_t}. Since the sets \inline{I_t} are disjoint, we have \inline{\br{r-1}\tau m <\ts{\D*}}, hence \inline{r \leq 1 + \ts{\D^*}/\br{\tau m}}. Therefore, except the calls to \aref{\textsc{AugmentCover}}{fig:sub_fpa}, all passes through Step~7 take \inline{\bigOtilde{m+\ts{\D^*}}} time. Combining the preceding accounting with \lemref{lemma:fpa_sz} , the total time is \[\bigOtilde{m+d\cdot\min\brcu{n\card{\C}, dm}} \subseteq \bigOtilde{\min\brcu{dn\mu, d^2m}}.\]
The claimed bound follows by noting that \inline{m=\bigO{nd}}, so if \inline{m > \mu} then \inline{m=\bigTheta{n^2}} and \inline{n=\bigO{d}}.
\end{proof}

\subsection{\texorpdfstring{\inline{\br{n\sqrt{\mu}}}}{n sqrt(mu)}-Concise Covers}
\label{subsec:concise_bcc}

The bound \inline{\card{\C}=\bigO{\mu}} in \lemref{lemma:fpa_sz} is tight, as follows from Lov{\'a}sz's tight example \cite{lovasz1968covering}. The following result shows that the \inline{\bigO{n\mu}} total-size bound is also tight for both Lov{\'a}sz's construction and its concise counterpart.

\begin{restatable}{lemma}{LemmaFpaTight}
\label{lemma:fpa_tight}
There exists a family of graphs for which both Lov{\'a}sz's construction and Algorithm~\aref{\textsc{FirstFit-PA}}{fig:algo_fpa} produce clique covers of total size \inline{\bigTheta{n\mu}}.
\end{restatable}

We next describe a construction for covers of total size \inline{\bigO{\min\brcu{m, n\sqrt{\mu}}}}, and subsequently show how to make these covers concise while preserving this bound.

In Lov{\'a}sz's construction, each vertex in a later clique \inline{F_k} is considered separately with respect to an earlier clique \inline{F_\ell}. Thus, if several vertices in \inline{F_k} have the same neighbors in \inline{F_\ell}, the same subset of \inline{F_\ell} is included repeatedly in different cliques. These vertices can instead be grouped and covered together with their common neighbors in \inline{F_\ell}. We formalize this idea below for covering the edges between two cliques.

\begin{definition}[\inline{\ccov{X}{Y}}]
\label{def:crosscover}
Let \inline{X} and \inline{Y} be disjoint cliques of \inline{G}. Define an equivalence relation on \inline{X} by 
\[u \sim_Y v \iff N\br{u} \cap Y = N\br{v} \cap Y.\]
For each equivalence class \inline{R}, define \inline{Z_R := N\br{u} \cap Y}, for any \inline{u \in R}. Then define
\[\ccov{X}{Y} := \brcu{R \cup Z_R \mid R \text{ is an equivalence class of } \sim_Y, Z_R \ne \emptyset}.\]
\end{definition}

To apply cross-cover globally, we first partition \inline{V} into clique parts. Since \inline{\overline{G}} has \inline{\mu-1} edges, a degeneracy coloring of \inline{\overline{G}} uses \inline{\bigO{\sqrt{\mu}}} colors, whose color classes are cliques of \inline{G}. These color classes can be large, so applying cross-cover directly between them can still have \inline{\bigTheta{n\mu}} total size. We therefore split each color class into parts of size at most \inline{b}. There is a tradeoff in choosing \inline{b}. Smaller \inline{b} creates more parts and hence more pairs of parts, contributing \inline{\bigO{n^2/b}} to the total size. Larger \inline{b} allows each non-edge between two parts to contribute up to \inline{\bigO{b}} assignments through cross-cover, giving a total contribution of \inline{\bigO{b\mu}}. Balancing these terms gives \inline{b = \bigTheta{n/\sqrt{\mu}}}. The following definition formalizes this construction.

\begin{definition}[Balanced Cross-Cover (BCC) Construction]
\label{def:bcc}
~
\begin{enumerate}
    \item Compute a degeneracy coloring of \inline{\overline{G}}, and let \inline{F_1, \dots, F_t} be its color classes.
    \item Set \inline{b:= \left\lceil \frac{n}{\sqrt{\mu}} \right\rceil}, and split each color class \inline{F_{\ell}} into \inline{\lceil \card{F_\ell}/b \rceil} parts of size at most \inline{b}. Denote all resulting parts by \inline{P_1, \dots, P_r}.
    \item Define \[\C := \brcu{T \in \brcu{P_1, \dots, P_r} \mid \card{T} \geq 2} \cup \bigcup_{1\leq i <j \leq r} \ccov{P_i}{P_j}.\]
\end{enumerate}
\end{definition}

\begin{restatable}{lemma}{LemmaBccSz}
\label{lemma:bcc_size}
The BCC construction gives a clique cover \inline{\C} of \inline{G} satisfying
\[\card{\C} = \bigO{\min\brcu{m,\mu}} \text{ and } \ts{\C}= \bigO{\min\brcu{m, n \sqrt{\mu}}}.\]    
\end{restatable}

\begin{remark}
The bound in \lemref{lemma:bcc_size} is asymptotically tight. For example, let \inline{G=K_{q,q}}. Then \inline{n=2q}, \inline{m = \bigTheta{n^2}}, and \inline{\mu=\bigTheta{n^2}}. Since every clique cover of \inline{G} has total size \inline{\bigOmega{m}}, \inline{\ts{\C}=\bigOmega{m} = \bigOmega{n^2} = \bigOmega{n \sqrt{\mu}}}.
\end{remark}

To make the BCC covers concise, we use the admissibility structures and phased-admissibility scheme introduced in \hyref{Section}{subsec:concise_lzf}. Processing the edges of a BCC clique individually, as in \aref{\textsc{FirstFit-PA}}{fig:algo_fpa}, does not preserve the \inline{\bigO{n\sqrt{\mu}}} total-size bound (the graph family in \lemref{lemma:concise_sep} can be used to exhibit this, under a suitable processing order). We therefore process each BCC clique atomically, either augmenting one current clique by the entire BCC clique or introducing it as a new clique.

\begin{figure}[h]
    \centering
    \begin{parbox}{5.6in}{    
        \begin{mdframed}[linewidth=0.5pt, roundcorner=7pt, backgroundcolor=gray!5, frametitle={\underline{\textsc{CrossCover-PA}$\br{G=\br{V,E}, \tau}$}}]          
        \begin{enumerate}                        
            \item Construct the BCC cover \inline{\B := \brcu{B_i}_{i \in [p]}} and the spanning trees \inline{\T:=\brcu{T_i}_{i \in [p]}}.
            \item Initialize \inline{\C \gets \emptyset}.
            \item While \inline{\B \ne \emptyset} do    
            \begin{enumerate}         
                \item Set \inline{\mathrm{open} \gets \mathrm{true}, \D \gets \emptyset}, and \inline{A_v\gets \emptyset}, for all \inline{v \in V}.
                \item For each current \inline{B_i \in \B} do
                \begin{enumerate}
                    \item Set \[L \gets \bigcap_{\brcu{u,v} \in E\br{T_i}} \br{A_u \cap A_v}.\]
                    \item If \inline{L \ne \emptyset}, select \inline{\ell \in L}, set \inline{C_\ell \gets C_\ell \cup B_i}, update \inline{D_\ell} and the relevant \inline{A_w}'s.
                    \item Else if \inline{\mathrm{open} = \mathrm{true}}, introduce a label \inline{\ell}, set \inline{C_\ell \gets B_i} and \inline{\C \gets \C \cup \brcu{C_\ell}}, create \inline{D_\ell}, update the relevant \inline{A_w}'s, and set \inline{\mathrm{open} \gets \mathrm{false}} if \inline{\ts{\D} > \tau m}.
                    \item If all edges of \inline{B_i} are covered by \inline{\C}, remove \inline{B_i} and \inline{T_i} from \inline{\B} and \inline{\T}, respectively.
                \end{enumerate}
            \end{enumerate}
            \item Return \inline{\C}.
        \end{enumerate}
        \end{mdframed}        
    }
    \end{parbox}    
    \caption{Concise BCC covers via phased admissibility.}
    \label{fig:algo_cpa}
\end{figure}

Algorithm~\aref{\textsc{CrossCover-PA}}{fig:algo_cpa} applies the phased-admissibility scheme to the BCC cover \inline{\B:=\brcu{B_i}_{i \in [p]}}. In each phase, the admissibility structures contain the labels introduced in that phase. For each remaining \inline{BCC} clique \inline{B_i}, the algorithm checks whether some current clique \inline{C_\ell} can be augmented by all of \inline{B_i}. If so, \inline{B_i} is added atomically to one such \inline{C_\ell}. Otherwise, \inline{B_i} is introduced as a new clique when the phase is open, and deferred to the next phase when the phase is closed unless all its edges are already covered.

To retain the runtime bound of \aref{\textsc{FirstFit-PA}}{fig:algo_fpa}, we associate each \inline{B_i} with a spanning tree \inline{T_i} and use only its edges for admissibility tests. For an initial part \inline{P_j}, we take a star centered at any vertex of \inline{P_j}. For a cross-cover clique \inline{B_i = R \cup Z_R}, we choose \inline{x \in R} and \inline{y \in Z_R}, and take the spanning tree with edge set
\[E\br{T_i} = \brcu{\brcu{x,u} \mid u \in Z_R} \cup \brcu{\brcu{v,y} \mid v \in R}.\]
Since the corresponding edge sets are pairwise edge-disjoint across the BCC cliques, the trees \inline{T_i} are also pairwise edge-disjoint. Since \inline{T_i} spans \inline{B_i},
\[\bigcap_{\brcu{u,v} \in E\br{T_i}} \br{A_u \cap A_v} = \bigcap_{w \in B_i} A_w.\]

The admissibility creation and updates follow the scheme in \aref{\textsc{FirstFit-PA}}{fig:algo_fpa}, extended from an edge to the entire clique \inline{B_i}. When \inline{B_i} is added to an existing clique \inline{C_\ell}, no update is needed if \inline{B_i \subseteq C_\ell}. Otherwise, define \inline{D_\ell^\prime := D_\ell \cap \bigcap_{\brcu{u,v} \in E\br{T_i}} \br{N[u]\cap N[v]}, X:=D_\ell \setminus D_\ell^\prime}. The algorithm sets \inline{D_\ell \gets D_\ell^\prime} and removes \inline{\ell} from \inline{A_w} for every \inline{w \in X}. Since \inline{T_i} spans \inline{B_i}, 
\[D_\ell^\prime = D_\ell \cap \bigcap_{\brcu{u,v} \in E\br{T_i}}\br{N[v]\cap N[u]} = D_\ell \cap \bigcap_{w \in B_i} N[w].\]
Similarly, when \inline{B_i} introduces a new label \inline{\ell}, it sets \inline{D_\ell := \bigcap_{\brcu{u,v} \in E\br{T_i}}\br{N[v]\cap N[u]}} and adds \inline{\ell} to \inline{A_w} for every \inline{w \in D_\ell}. This atomic processing of the BCC cliques, together with \lemref{lemma:bcc_size}, completes the proof of \hyref{Theorem}{thm:cpa} as follows.

\begin{restatable}{lemma}{LemmaCpa}
\label{lemma:cpa}
Algorithm~\aref{\textsc{CrossCover-PA}}{fig:algo_cpa} uses \inline{\bigOtilde{d^2 \cdot \min\brcu{m,\mu}}} time and \inline{\bigO{m}} space, and returns a concise cover \inline{\C} with \inline{\ts{\C} = \bigO{\min\brcu{m, n\sqrt{\mu}}}}.    
\end{restatable}
\section{Cover-Parsimonious Algorithms}
\label{sec:cp_algo}

Recall that a  graph algorithm is cover-parsimonious if, given a cover \inline{\C}, it uses \inline{\bigOtilde{\ts{\C}}} time and \inline{\bigO{\ts{\C}}} space. This section describes cover-parsimonious algorithms for matching and graph traversal problems. Additional cover-parsimonious algorithms for weighted vertex cover, maximal independent set, and vertex coloring, along with deferred proofs, appear in \hyref{Appendix}{sec:cp_algo_deferred}.

\subsection{Maximum Matching}
\label{subsec:mcm}

\begin{definition}[Maximal and Maximum Matching]
A \emph{matching} in a graph \inline{G=\br{V,E}} is a set \inline{\M \subseteq E} of pairwise vertex-disjoint edges. A matching \inline{\M} is \emph{maximal} if no edge from \inline{E\setminus \M} can be added to \inline{\M} without violating the matching property. Equivalently, every edge of \inline{G} has at least one endpoint matched in \inline{\M}. A \emph{maximum matching} is a matching of maximum cardinality among all matchings in \inline{G}.  
\end{definition}

\begin{definition}[\inline{\M}-Augmenting Path]
\label{def:augmenting_path}
An \emph{\inline{\M}-alternating path} is a path whose edges alternate between edges in \inline{\M} and edges in \inline{E\setminus \M}. An \emph{\inline{\M}-augmenting path} is an \inline{\M}-alternating path whose endpoints are unmatched in \inline{\M}. The length of a path is its number of edges.   
\end{definition}

We describe a cover-parsimonious \inline{3/2}-approximation algorithm for maximum matching using the following standard fact.

\begin{lemma}[\cite{hopcroft1973n}]
\label{lemma:augmenting_path}
Let \inline{\M} be a matching, let \inline{\M^*} be a maximum matching, and let \inline{k \geq 1} be an integer. If \inline{\M} has no augmenting path of length at most \inline{2k-1}, then \inline{\card{\M^*} \leq \frac{k+1}{k} \card{\M}}.
\end{lemma}

It follows immediately that every maximal matching is a \inline{2}-approximation for maximum matching, since a maximal matching has no augmenting path of length \inline{1}. We first describe a cover-parsimonious algorithm for computing a maximal matching and then augment the resulting matching so that no augmenting path of length at most \inline{3} remains.

\fref{fig:algo_matching} describes a maximal-matching algorithm on a clique cover \inline{\C}. The algorithm maintains an indicator array \inline{I_v \in \brcu{0,1}} that records whether \inline{v} is matched. For each clique \inline{C_\ell}, it forms the set \inline{U \subseteq C_\ell} of currently unmatched vertices, pairs them arbitrarily, and adds the resulting edges to the matching.

\begin{figure}[h]
    \centering
    \begin{parbox}{3.8in}{    
        \begin{mdframed}[linewidth=0.5pt, roundcorner=7pt, backgroundcolor=gray!5,
                         frametitle={\underline{\textsc{Maximal-Matching}$\br{V,\, \C}$}}]  
        \begin{enumerate}                        
            \item Initialize \inline{\M \gets \emptyset} and \inline{I_v \gets 0}, for all \inline{v \in V}.
            \item For each clique \inline{C_\ell \in \C} do
            \begin{enumerate}
                \item Set \inline{U \gets \brcu{\, v\in C_\ell \mid I_v =0\,}}.
                \item If \inline{\card{U} \leq 1}, then continue to the next clique.
                \item Let \inline{\langle u_1,\dots,u_k\rangle} be an arbitrary ordering of \inline{U}.
                \item For \inline{i=1} to \inline{k-1} in steps of 2 do
                \begin{enumerate}                    
                    \item Set \inline{\M \gets \M \cup \brcu{\brcu{u_i,u_{i+1}}}}.
                    \item Set \inline{I_{u_i} \gets 1} and \inline{I_{u_{i+1}} \gets 1}.
                \end{enumerate}
            \end{enumerate}
            \item Return \inline{\M}.
        \end{enumerate}
        \end{mdframed}    
    
    }
    \end{parbox}    
    \caption{A cover-parsimonious algorithm for computing a maximal matching.}
    \label{fig:algo_matching}
\end{figure}

\begin{lemma}
\label{lemma:maximal_matching}
Let \inline{\C} be a clique cover of \inline{G}. Algorithm~\aref{\textsc{Maximal-Matching}}{fig:algo_matching} computes a 2-approximate maximum matching in \inline{G} in \inline{\bigO{\ts{\C}}} time and space.
\end{lemma}
\begin{proof}
The algorithm adds an edge to \inline{\M} only when both of its endpoints are currently unmatched and immediately marks both endpoints as matched. Hence, no two edges in \inline{\M} share an endpoint, so \inline{\M} is a matching.

To prove that \inline{\M} is maximal, suppose that some edge \inline{\brcu{u,v}} has both endpoints unmatched at termination. When the algorithm processed a clique containing \inline{\brcu{u,v}}, both \inline{u} and \inline{v} belonged to \inline{U}. However, the pairing step leaves at most one vertex of \inline{U} unmatched, a contradiction.

Processing a clique \inline{C_\ell} takes \inline{\bigO{\card{C_\ell}}} time. Since the algorithm processes every clique once, its running time is \inline{\bigO{\ts{\C}}}. The algorithm stores the indicator array \inline{\brcu{I_v}_{v \in V}} and the temporary set \inline{U}, using \inline{\bigO{n}\subseteq \bigO{\ts{\C}}} space.

Since every maximal matching is a 2-approximate maximum matching, the lemma follows.
\end{proof}

\begin{figure}[h]
    \centering
    \begin{parbox}{4.4in}{    
        \begin{mdframed}[linewidth=0.5pt, roundcorner=7pt, backgroundcolor=gray!5,
                         frametitle={\underline{\textsc{Three-Halves-Matching}$\br{V,\, \C}$}}]  
        \begin{enumerate}
            \item Initialize an empty list \inline{\Gamma_v} and set \inline{I_v \gets 0}, for all \inline{v \in V}.
            \item Set \inline{\M^\prime \gets} \aref{\textsc{Maximal-Matching}$\br{V,\, \C}$}{fig:algo_matching}.
            \item Set \inline{\M \gets \M^\prime} and \inline{I_v, \, I_u \gets 1}, for all \inline{\brcu{u, v} \in \M^\prime}.
            \item For each clique \inline{C_\ell \in \C} do
            \begin{enumerate}
                \item If \inline{I_v = 1} for all \inline{v \in C_\ell}, then continue to the next clique.
                \item Let \inline{u \in C_\ell} be the unique vertex with \inline{I_u=0}.
                \item Append \inline{u} to \inline{\Gamma_v}, for all \inline{v \in C_\ell \setminus \brcu{u}}.      
            \end{enumerate}
            \item For each \inline{\brcu{u,v} \in \M^\prime} do
            \begin{enumerate}
                \item If \inline{\Gamma_u} contains a vertex \inline{x} and \inline{\Gamma_v} contains a vertex \inline{y \ne x} with \inline{I_x=I_y=0}, then
                \begin{enumerate}
                    \item Set \inline{\M \gets \br{\M \setminus \brcu{\brcu{u,v}}} \cup \brcu{\brcu{x,u},\brcu{v,y}}}.
                    \item Set \inline{I_x, \, I_y \gets 1}.
                \end{enumerate}
            \end{enumerate}
            \item Return \inline{\M}.
        \end{enumerate}
        \end{mdframed}        
    }
    \end{parbox}    
    \caption{A cover-parsimonious algorithm for \inline{3/2}-approximate maximum matching.}
    \label{fig:algo_matching2}
\end{figure}

\fref{fig:algo_matching2} describes an algorithm that augments the maximal matching \inline{\M^\prime} into a \inline{3/2}-approximation for maximum matching. For each vertex \inline{v}, the algorithm maintains a list \inline{\Gamma_v} of neighbors that are unmatched in the initial maximal matching \inline{\M^\prime}. Duplicate entries are allowed: the same vertex may be appended to \inline{\Gamma_v} through several cliques. Because \inline{\M^\prime} is maximal, each clique contains at most one vertex unmatched in \inline{\M^\prime}. Thus, when a clique \inline{C_\ell} has an unmatched vertex \inline{u}, the algorithm appends \inline{u} to \inline{\Gamma_v} for every other vertex \inline{v \in C_\ell}.

The algorithm then processes every edge \inline{\brcu{u,v}\in \M^\prime}. If \inline{\Gamma_u} contains a currently unmatched vertex \inline{x} and \inline{\Gamma_v} contains a distinct currently unmatched vertex \inline{y}, then \inline{\br{x,u,v,y}} is a length-3 \inline{\M^\prime}-augmenting path. The algorithm performs this augmentation by replacing \inline{\brcu{u,v}} with \inline{\brcu{x,u}} and \inline{\brcu{v,y}}. Hence, it constructs a maximal collection of vertex-disjoint length-3 \inline{\M^\prime}-augmenting paths.

\begin{lemma}
\label{lemma:approx_mcm}
Let \inline{\C} be a clique cover of \inline{G}. Algorithm~\aref{\textsc{Three-Halves-Matching}}{fig:algo_matching2} computes a \inline{3/2}-approximate maximum matching in \inline{G} in \inline{\bigO{\ts{\C}}} time and space.
\end{lemma}
\begin{proof}
Let \inline{U^\prime} be the set of vertices unmatched in \inline{\M^\prime}. Ignoring duplicate entries, \inline{\Gamma_v} contains exactly the neighbors of \inline{v} that belong to \inline{U^\prime}. Every vertex appended to \inline{\Gamma_v} is unmatched in \inline{\M^\prime} and shares a clique with \inline{v}, so it is a neighbor of \inline{v} in \inline{U^\prime}. Conversely, if \inline{u \in U^\prime} is adjacent to \inline{v}, then some clique contains both \inline{u} and \inline{v}. Since \inline{\M^\prime} is maximal, this clique contains no other vertex of \inline{U^\prime}. Thus Step~4 identifies \inline{u} as its unique unmatched vertex and appends \inline{u} to \inline{\Gamma_v}.

Throughout Step~5, \inline{\M} is obtained from \inline{\M^\prime} by augmenting along a collection of vertex-disjoint length-3 paths. Since the edges of \inline{\M^\prime} are vertex disjoint and \inline{x} and \inline{y} are currently unmatched, each augmentation preserves this property; hence, \inline{\M} remains a matching.

Each augmentation matches two vertices of \inline{U^\prime} without making any previously matched vertex unmatched. Therefore, every vertex unmatched in \inline{\M} belongs to \inline{U^\prime}.
Since \inline{\M^\prime} is maximal, no two vertices of \inline{U^\prime} are adjacent. Hence, \inline{\M} has no augmenting path of length \inline{1}.

Suppose for contradiction that \inline{\br{a,p,q,b}} is a length-3 \inline{\M}-augmenting path, where \inline{\brcu{p,q} \in \M} and \inline{a,b} are unmatched in \inline{\M}. 

First, suppose \inline{\brcu{p,q} \in \M^\prime}. Since \inline{a,b} are unmatched at termination, they were also unmatched when \inline{\brcu{p,q}} was processed. Since \inline{a \in \Gamma_p}, \inline{b \in \Gamma_q}, and \inline{a \ne b}, the algorithm would have augmented \inline{\brcu{p,q}}, contradicting \inline{\brcu{p,q} \in \M}.

Now suppose that \inline{\brcu{p,q} \not\in \M^\prime}. Then \inline{\brcu{p,q}} was introduced by an augmentation, so exactly one of \inline{p} and \inline{q} belongs to \inline{U^\prime}. If \inline{p \in U^\prime}, then \inline{a \in U^\prime} and \inline{\brcu{a,p} \in E}, contradicting the maximality of \inline{\M^\prime}. The case \inline{q \in U^\prime} is symmetric, using the edge \inline{\brcu{q,b}}.

Therefore, \inline{\M} has no augmenting path of length at most \inline{3}. Applying \lemref{lemma:augmenting_path} with \inline{k=2} shows that \inline{\M} is a \inline{3/2}-approximation for maximum matching.

By \lemref{lemma:maximal_matching}, computing \inline{\M^\prime} takes \inline{\bigO{\ts{\C}}} time and space. Step~4 processes each clique \inline{C_\ell} in \inline{\bigO{\card{C_\ell}}} time, so both the time required to construct the lists and their total length are \inline{\bigO{\ts{\C}}}. In Step~5, each list \inline{\Gamma_v} is scanned at most once; while scanning, the algorithm retains at most two distinct unmatched entries to determine whether suitable \inline{x} and \inline{y} exist. Thus, Step~5 takes \inline{\bigO{\ts{\C}}} time. The matching \inline{\M^\prime}, the indicator array, and the lists use \inline{\bigO{\ts{\C}}} space.
\end{proof}

\subsection{Vertex-Weighted Matching}

\begin{definition}[Maximum Vertex-Weighted Matching (MVM)]
\label{def:mvm}
Let \inline{G=\br{V,E}} be a graph with a nonnegative vertex-weight function \inline{w:V\rightarrow \mathbb{R}_{\geq 0}}. For a matching \inline{\M}, its vertex weight is \inline{w\br{\M}:= \sum_{v \text{ matched in }\M} w\br{v} }. A maximum vertex-weighted matching in \inline{G} is a matching of maximum vertex weight among all matchings in \inline{G}.
\end{definition}

We describe a cover-parsimonious \inline{\br{3/2+\epsilon}}-approximation algorithm for MVM. In addition to the notion of an \inline{\M}-augmenting path (\hyref{Definition}{def:augmenting_path}), we use the following weighted local improvement.

\begin{definition}[\inline{\M}-Increasing Path]
An even-length \inline{\M}-alternating path \inline{P=\br{x,v,\dots, z}} is an \emph{\inline{\M}-increasing path} if \inline{x} is unmatched in \inline{\M}, \inline{z} is matched in \inline{\M}, and \inline{w\br{x} > w\br{z}}. Its gain is \inline{g\br{P}:=w\br{x} - w\br{z}}.
\end{definition}

Flipping an \inline{\M}-increasing path \inline{P=\br{x,v,\dots,z}} matches \inline{x}, leaves every internal vertex matched, and makes \inline{z} unmatched. Hence the matching weight increases by \inline{g\br{P}}. Let \inline{E\br{P}} denote the edge set of \inline{P}. For two edge sets \inline{A} and \inline{B}, define \inline{A \oplus B :=\br{A \setminus B} \cup \br{B \setminus A}}.

For MVM, Theorem 2.1 of Pettie and Sanders~\cite{pettie2004simpler} simplifies to the following stronger statement.

\begin{restatable}{lemma}{LemmaMVMPathPacking}
\label{lemma:mvm_path_packing}
Let \inline{k \geq 1} be an integer, let \inline{\M^*} be an MVM, and let \inline{\M} be a matching with no \inline{\M}-augmenting path of length at most \inline{2k-1}. Then there exists a vertex-disjoint family \inline{\A} of \inline{\M}-increasing paths of length at most \inline{2k} such that 
\[g\br{\A} :=\sum_{P \in \A} g\br{P} \geq \frac{k}{k+1}w\br{\M^*} - w\br{\M}.\]
\end{restatable}

Given a matching \inline{\M} with no \inline{\M}-augmenting path of length at most \inline{3}, we construct a vertex-disjoint family \inline{\B} of \inline{\M}-increasing paths of length at most 4 satisfying \inline{g\br{\B} = \bigOmega{g\br{\A}}}, where \inline{\A} is a family guaranteed by \lemref{lemma:mvm_path_packing} for \inline{k=2}. The matching \inline{\M} remains fixed throughout the construction. We use the following definitions and data structures.

An unmatched vertex and an edge of \inline{\M} are called \emph{active} until they occur in a selected path. Since \inline{\M} has no augmenting path of length 1, every clique \inline{C_\ell} contains at most one unmatched vertex in \inline{\M}. Let \inline{x_\ell} denote this vertex while it is active. For each vertex \inline{v}, let \inline{L_v := \brcu{\ell \mid v \in C_\ell, C_\ell \in \C}} denote its clique-incidence list, and let \inline{\Lcal:=\brcu{L_v}_{v\in V}}. 

For every orientation \inline{\br{y,z}} of an edge \inline{\brcu{y,z} \in \M}, store a \emph{terminal record} \inline{\br{y,z}}, with key \inline{w\br{z}}, in a list \inline{\beta_\ell} for every \inline{\ell \in L_y}. Each list \inline{\beta_\ell} is sorted by its keys and contains only records of active matching edges. 

For an orientation \inline{\br{u,v}} of an active edge \inline{\brcu{u,v} \in \M}, define its \emph{source vertex} \[S\br{u} \in \argmax_{x_\ell: \ell \in L_u, x_\ell \text{ is active}} w\br{x_\ell},\]
and its \emph{terminal edge} \[T\br{u,v} \in \argmin_{ \br{y,z} \in \brcu{\br{u,v}} \cup \bigcup_{\ell \in L_v} \brcu{\br{y,z} \in \beta_\ell \mid \brcu{y,z} \ne \brcu{u,v}}} w\br{z}.\]

Whenever \inline{S\br{u}} exists, let \[P(u,v):=\begin{cases}
\br{S\br{u}, u, v}, T\br{u,v} = \br{u,v},\\
\br{S\br{u}, u, v, y, z}, T\br{u,v} = \br{y,z} \ne \br{u,v}.\\
\end{cases}\]

\fref{fig:sub_path_packing} outlines a procedure for packing a vertex-disjoint family of these paths.

\begin{figure}[h]
    \centering
    \begin{parbox}{5.0in}{    
        \begin{mdframed}[linewidth=0.5pt, roundcorner=7pt, backgroundcolor=gray!5,
                         frametitle={\underline{\textsc{Increasing-Path-Packing}$\br{V, \, \C, \, \Lcal, \, \M, \, w:V\rightarrow \mathbb{R}_{\geq 0}}$}}]  
        \begin{enumerate}            
            \item Initialize \inline{\B \gets \emptyset}. Mark every vertex unmatched in \inline{\M} and every edge of \inline{\M} as active, and construct the values \inline{x_\ell} and sorted lists \inline{\beta_\ell}.         
            \item Initialize \inline{g_{\max} \gets 0}. Scan both orientations \inline{\br{u,v}} of every edge \inline{\brcu{u,v} \in \M}, and whenever \inline{S\br{u}} exists, compute \inline{P\br{u,v}}, and set \inline{g_{\max} \gets \max \brcu{g_{\max}, g\br{P\br{u,v}}}}.
            \item If \inline{g_{\max} \leq 0}, return \inline{\emptyset}.            
            \item For \inline{i=0,1,\dots, \lceil \log_{2}n\rceil},  set \inline{\tau_i \gets \frac{g_{\max}}{2^i}}, and process both orientations \inline{\br{u, v}} of every active edge \inline{\brcu{u,v} \in \M}:
            
            \begin{enumerate}
                \item Scan \inline{L_u} to determine \inline{S\br{u}}. If no active source exists, continue to the next orientation.
                \item Scan \inline{L_v} and compute \inline{T\br{u,v}} and \inline{P\br{u,v}}.
                \item If \inline{g\br{P\br{u,v}} < \tau_i}, continue to the next orientation.
                \item Add \inline{P\br{u,v}} to \inline{\B}, deactivate its source vertex and the matching edges on \inline{P\br{u,v}}, and delete the terminal records of these matching edges from the corresponding lists \inline{\beta_\ell}.
            \end{enumerate}            
            \item Return \inline{\B}.
        \end{enumerate}
                
        \end{mdframed}        
    }
    \end{parbox}    
    \caption{A cover-parsimonious \inline{\M}-increasing path-packing subroutine.}
    \label{fig:sub_path_packing}
\end{figure}

\begin{lemma}
\label{lemma:sub_path_packing}
Let \inline{\M^*} be an MVM, and let \inline{\M} be a matching with no \inline{\M}-augmenting path of length at most 3. The subroutine \aref{\textsc{Increasing-Path-Packing}}{fig:sub_path_packing} uses \inline{\bigO{\ts{\C} \log n}} time and \inline{\bigO{\ts{\C}}} space and returns a vertex-disjoint family \inline{\B} satisfying \[g\br{\B} \geq \frac{3}{19} \br{\frac{2}{3}w\br{\M^*} - w\br{\M}}.\]
\end{lemma}
\begin{proof}
The family \inline{\B} is vertex-disjoint because whenever the subroutine selects a path it immediately deactivates the source vertex and matching edges along the path. 

Let \inline{\A} be a family guaranteed by \lemref{lemma:mvm_path_packing} for \inline{k=2}. Since the first threshold is \inline{g_{\max}}, a path of gain \inline{g_{\max}} is either selected when its orientation is processed or has already been intersected by another path selected at the same threshold. Hence \inline{g\br{\B} \geq g_{\max}}. 

Let \inline{\A_{L}} be the set of paths \inline{Q \in \A} such that \inline{g\br{Q} < \frac{g_{\max}}{2^{\lceil \log_2 n \rceil}} \leq \frac{g_{\max}}{n}}. The paths in \inline{\A} are vertex-disjoint and each contains at least three vertices, so \inline{\card{\A} \leq \frac{n}{3}}. Therefore, \inline{g\br{\A_L} \leq \frac{n}{3} \cdot \frac{g_{\max}}{n} \leq \frac{g\br{\B}}{3}}. 

Consider \inline{Q \in \A_H:=\A \setminus \A_L} and let \inline{i} be the first iteration such that \inline{\tau_i \leq g\br{Q} < 2\tau_i}.
If \inline{ Q \in \B}, charge \inline{Q} to itself. Otherwise, \inline{Q} intersects a path \inline{P \in \B} selected in iteration \inline{i} or earlier. In either case, \inline{g\br{Q} < 2\tau_i \leq 2 g\br{P}}. 

Any increasing path that meets a matched vertex contains its incident matching edge. Since the paths in \inline{A_H} are vertex-disjoint, at most one can contain the source of a selected path \inline{P}, and at most one can intersect each matching edge of \inline{P}. Thus \inline{P \in \B} receives at most 3 charges, and \inline{g\br{\A_H} \leq 6 g\br{\B}}. Consequently,
\inline{g\br{\A} = g\br{\A_L}+ g\br{\A_H} \leq \frac{19}{3}g\br{\B}}. The claimed gain bound follows from \lemref{lemma:mvm_path_packing}.

It remains to analyze time and space. Since each clique contains at most one unmatched vertex, constructing the values \inline{x_\ell} takes \inline{\bigO{\ts{\C}}} time and space. The lists \inline{\beta_\ell} store \inline{\bigO{\ts{\C}}} terminal records in total and can be sorted in \inline{\bigO{\ts{\C} \log n}} time.

For an orientation \inline{\br{u,v}}, computing \inline{S\br{u}} takes \inline{\card{L_u}} time. To compute \inline{T\br{u,v}}, it suffices to examine the first three active records of each \inline{\beta_\ell}, for \inline{\ell \in L_v}, because the center edge contributes at most two records to one list. Thus \inline{T\br{u,v}} takes \inline{\bigO{\card{L_v}}} time. 

Therefore, Step~2 and each iteration of Step~4 take at most \inline{\sum_{\brcu{u,v} \in \M} \bigO{\card{L_u}+ \card{L_v}} =\bigO{\ts{\C}}} time. Using direct pointers, every terminal record is deleted at most once, so all deletions take \inline{\bigO{\ts{\C}}} total time. Since Step~4 has \inline{\bigO{\log n}} iterations, the claimed bounds follow.
\end{proof}

\begin{figure}[h]
    \centering
    \begin{parbox}{5.0in}{    
        \begin{mdframed}[linewidth=0.5pt, roundcorner=7pt, backgroundcolor=gray!5,
                         frametitle={\underline{\textsc{Approximate-MVM}$\br{V,\, \C, \epsilon, w:V \rightarrow \mathbb{R}_{\geq 0}}$}}]  
        \begin{enumerate}                        
            \item Initialize \inline{\M \gets \emptyset} and construct \inline{\Lcal}.
            \item Repeat \inline{\left\lceil\log_{19/16}\frac{3+2\epsilon}{2\epsilon} \right\rceil} times
            \begin{enumerate}
                \item Augment \inline{\M} until no augmenting path of length at most 3 remains.
                \item Let \inline{\B \gets} \aref{\textsc{Increasing-Path-Packing}$\br{V,\, \C, \Lcal, \M, w}$}{fig:sub_path_packing}.
                \item If \inline{\B =\emptyset}, then return \inline{\M}.
                \item Set \inline{\M \gets \M \oplus \bigcup_{P \in \B} E\br{P}}.
            \end{enumerate}
            \item Return \inline{\M}.
        \end{enumerate}
        \end{mdframed}    
    
    }
    \end{parbox}    
    \caption{A cover-parsimonious algorithm for computing a \inline{\br{3/2+\epsilon}}-approximate MVM.}
    \label{fig:algo_mvm}
\end{figure}

Using the procedures for cardinality matchings outlined in \hyref{Section}{subsec:mcm}, we obtain a matching \inline{\M} with no \inline{\M}-augmenting path of length at most 3. Repeatedly invoking the subroutine \aref{\textsc{Increasing-Path-Packing}}{fig:sub_path_packing} with such matchings \inline{\M} gives a \inline{\br{3/2+\epsilon}}-approximate MVM as outlined in \fref{fig:algo_mvm}. The following holds for this algorithm.

\begin{lemma}
Let \inline{\C} be a clique cover of an \inline{n}-vertex graph. For every constant \inline{\epsilon >0}, Algorithm~\aref{\textsc{Approximate-MVM}}{fig:algo_mvm} computes a \inline{\br{3/2+\epsilon}}-approximation of MVM in \inline{\bigO{\ts{\C} \log n}} time and \inline{\bigO{\ts{\C}}} space.
\end{lemma}
\begin{proof}
Let \inline{\M_t} be the matching \inline{\M} after \inline{t} iterations of \aref{\textsc{Approximate-MVM}}{fig:algo_mvm}, where \inline{\M_0 = \emptyset}, and let \[\lambda_t := \frac{2}{3}w\br{\M^*} - w\br{\M_t}.\] 
If \inline{\lambda_t \leq 0}, then \inline{\M_t} is already a \inline{3/2}-approximate MVM. Otherwise, the monotonicity of \inline{w\br{\M}} and \lemref{lemma:sub_path_packing} give \[w\br{\M_{t+1}} - w\br{\M_t} \geq \frac{3}{19}\lambda_t,\]
from which it follows that \inline{\lambda_{t+1} \leq \frac{16}{19}\lambda_t}. Since \inline{\lambda_0 = \frac{2}{3} w\br{\M^*}}, \inline{\lambda_t \leq \br{\frac{16}{19}}^{t}\frac{2}{3}w\br{\M^*}}. Setting \inline{t=\left\lceil\log_{19/16}\frac{3+2\epsilon}{2\epsilon} \right\rceil}, we obtain \inline{\lambda_t=\frac{2}{3}w\br{\M^*}-w\br{\M_t} \leq \br{\frac{2\epsilon}{3+2\epsilon}}\cdot \frac{2}{3}w\br{\M^*}}, so \inline{\br{3/2+\epsilon}w\br{\M_t} \geq w\br{\M^*}}.

By arguments identical to the proofs of \lemref{lemma:maximal_matching} and \lemref{lemma:approx_mcm}, Step~2(a) requires \inline{\bigO{\ts{\C}}} time and \inline{\bigO{\ts{\C}}} space. By \lemref{lemma:sub_path_packing}, Step~2(b) requires \inline{\bigO{\ts{\C} \log n}} time and \inline{\bigO{\ts{\C}}} space. Since these steps are executed for \inline{\bigO{\log\frac{1}{\epsilon}}} iterations, the claim follows.
\end{proof}

\begin{remark}
For MVM, the deterministic path-packing phase of Pettie and Sanders~\cite{pettie2004simpler} simplifies relative to the general edge-weighted setting and can also be implemented in \inline{\bigO{\ts{\C} \log n}} time and \inline{\bigO{\ts{\C}}} space, but the gain-thresholding approach used in \aref{\textsc{Increasing-Path-Packing}}{fig:sub_path_packing} is simpler. Al-Herz and Pothen~\cite{al20222} designed a \inline{3/2}-approximation for MVM, but their algorithm does not appear to admit a cover-parsimonious implementation.
\end{remark}

\subsection{Graph Traversals}
\label{subsec:traversals}

We describe algorithms to compute breadth-first search (BFS) and depth-first search (DFS) forests of an undirected graph. 

A standard BFS maintains a discovery indicator \inline{I_v \in \brcu{0,1}} and a parent pointer \inline{\parent{v}} for each vertex \inline{v}, and processes discovered vertices in FIFO order. When an undiscovered neighbor \inline{u} of the current vertex \inline{v} is encountered, it sets \inline{\parent{u} \gets v} and appends \inline{u} to the queue.

\fref{fig:bfs} describes an implementation of this traversal on a clique cover for a connected graph, starting from a source vertex \inline{r}. For each clique \inline{C_\ell \in \C}, the algorithm maintains an indicator \inline{J_\ell} of whether the clique has been scanned. When \inline{v} is removed from the queue, each unscanned clique incident to \inline{v} is scanned once. Every undiscovered vertex in such a clique is assigned \inline{v} as its parent and appended to the queue.

\begin{figure}[h]
    \centering
    \begin{parbox}{4.9in}{    
    \begin{mdframed}[
        linewidth=0.5pt,
        roundcorner=7pt,
        backgroundcolor=gray!5,
        frametitle={\underline{\textsc{BFS}$\br{V, \C, r}$}}
    ]
    \begin{enumerate}
        \item Set \inline{L_v \gets \brcu{\ell \mid v \in C_\ell, C_\ell \in \C}}, \inline{\parent{v} \gets \bot} and \inline{I_v \gets 0},
        for all \inline{v \in V}.
        
        \item Set \inline{J_\ell \gets 0}, for all \inline{C_\ell \in \C}.
        
        \item Set \inline{I_r \gets 1} and initialize a queue \inline{Q} with
        \inline{r}.
        
        \item While \inline{Q} is nonempty do
        \begin{enumerate}
            \item Remove the first vertex \inline{v} from \inline{Q}.
            
            \item For each \inline{\ell \in L_v} with \inline{J_\ell=0},
            set \inline{J_\ell \gets 1} and scan \inline{C_\ell}. For every
            \inline{u \in C_\ell} with \inline{I_u=0}, set
            \inline{I_u \gets 1} and \inline{\parent{u} \gets v}, and append
            \inline{u} to \inline{Q}.
        \end{enumerate}
        
        \item Return \inline{\brcu{\parent{v}}_{v \in V}}.
    \end{enumerate}
        
    \end{mdframed}
    }
    \end{parbox}   
    \caption{Breadth-first search tree using a clique cover.}
    \label{fig:bfs}
    \label{fig:algo_bfs}
\end{figure}

For a disconnected graph, the same traversal gives a BFS forest by performing Steps~1--2 once and repeating Steps~3--4 from each vertex that remains undiscovered.

\begin{lemma}
\label{lemma:bfs}
Let \inline{\C} be a clique cover of \inline{G}. Algorithm~\aref{\textsc{BFS}}{fig:algo_bfs},
extended as above to disconnected graphs, computes parent pointers
\inline{\brcu{\parent{v}}_{v \in V}} that define a breadth-first search forest of
\inline{G} in \inline{\bigO{\ts{\C}}} time and space.
\end{lemma}
\begin{proof}
Each vertex is discovered at most once and receives a parent only upon its first discovery. If \inline{\parent{u}=v}, then \inline{u} and \inline{v} lie in a common clique, so \inline{\brcu{u,v} \in E}. Since parents always point from a newly discovered vertex to an already discovered vertex, no cycle can form. Thus the parent pointers form a spanning forest.

Consider one root \inline{r}. The FIFO queue processes vertices in nondecreasing forest depth. Consider a shortest path \inline{r=v_0,\dots,v_k=u}. We show by induction that each \inline{v_i} is discovered at depth at most \inline{i}. The claim is immediate for \inline{v_0=r}. Suppose it holds for \inline{v_{i-1}}, and let \inline{C_\ell} cover \inline{\brcu{v_{i-1},v_i}}. 
When \inline{v_{i-1}} is processed, if \inline{C_\ell} has not yet been scanned, it is scanned then, ensuring that \inline{v_i} has been discovered at depth at most \inline{i}. 
If \inline{C_\ell} was scanned earlier, then it was scanned while processing a vertex of depth no greater than that of \inline{v_{i-1}}, and hence at most \inline{i-1}, so \inline{v_i} was already discovered at depth at most \inline{i}. 
Hence every vertex has forest depth at most its shortest-path distance from \inline{r}. The reverse inequality follows because the parent pointers themselves define a path from \inline{r} to each discovered vertex. Therefore forest depth equals shortest-path distance, giving the breadth-first property.

Each incidence list \inline{L_v} is scanned once, giving \inline{\sum_{v \in V} \card{L_v}=\ts{\C}} iterations. Each clique \inline{C_\ell} is scanned at most once, so the total number of clique positions inspected is at most \inline{\ts{\C}}. The incidence lists, clique indicators, queue, and  discovery array use \inline{\bigO{\ts{\C}}} space. Hence the algorithm uses \inline{\bigO{\ts{\C}}} time and space.
\end{proof}

A standard DFS maintains a discovery indicator \inline{I_v \in \brcu{0,1}} and a parent pointer \inline{\parent{v}} for each vertex \inline{v}. Starting from each undiscovered vertex, it recursively scans incident edges. When an undiscovered neighbor \inline{u} of the current vertex \inline{v} is encountered, it sets \inline{\parent{u} \gets v}, immediately visits \inline{u}, and resumes the scan of \inline{v} only after the recursive visit to \inline{u} completes. The resulting parent pointers define a DFS forest.

\fref{fig:dfs} describes an implementation of this traversal on a clique cover. For each \inline{C_\ell \in \C}, the algorithm fixes an ordering \inline{C_\ell[1],\dots, C_\ell[\card{C_\ell}]} and maintains a cursor \inline{a_\ell} indicating the next position of \inline{C_\ell} to inspect. When visiting \inline{v}, it scans the incident clique labels \inline{\ell \in L_v} in a fixed order and, for each such clique, resumes scanning \inline{C_\ell} from its current cursor \inline{a_\ell}. Whenever an undiscovered vertex \inline{u \in C_\ell} is encountered, it sets \inline{\parent{u} \gets v} and immediately visits \inline{u}.

\begin{figure}[h]
    \centering
    \captionsetup[subfigure]{font=small, labelfont=bf}

    \begin{subfigure}[t]{0.48\textwidth}
        \centering
        \begin{mdframed}[
            linewidth=0.5pt,
            roundcorner=7pt,
            backgroundcolor=gray!5,
            frametitle={\underline{\textsc{DFS}$\br{V, \C}$}}
        ]
        \begin{enumerate}
            \item Set \inline{L_v \gets \brcu{\ell \mid v \in C_\ell, C_\ell \in \C}}, for all \inline{v \in V}.
            \item Set \inline{\parent{v} \gets \bot} and \inline{I_v \gets 0}, for all \inline{v \in V}.
            \item For each \inline{C_\ell \in \C}, set \inline{a_\ell \gets 1}, fix an ordering of \inline{C_\ell}, and let \inline{C_\ell[i]} denote its $i$-th vertex.
            \item For each \inline{v \in V} with \inline{I_v=0} do
            \begin{enumerate}
                \item \aref{\textsc{DFS-Visit}}{fig:sub_dfs}\inline{\br{v}}.
            \end{enumerate}
            \item Return \inline{\brcu{\parent{v}}_{v \in V}}.
        \end{enumerate}
        \end{mdframed}
        \caption{Top-level routine.}
        \label{fig:algo_dfs}
    \end{subfigure}
    \hfill
    \begin{subfigure}[t]{0.48\textwidth}
        \centering
        \begin{mdframed}[
            linewidth=0.5pt,
            roundcorner=7pt,
            backgroundcolor=gray!5,
            frametitle={\underline{\textsc{DFS-Visit}$\br{v}$}}
        ]
        \textcolor{teal}{/*Uses \inline{\brcu{L_v}, \brcu{I_v}}, \inline{\brcu{\parent{v}}}, \inline{\brcu{a_\ell}}, \inline{\C} from \textsc{DFS}.*/}

        \begin{enumerate}
            \item Set \inline{I_v \gets 1}.
            \item For each \inline{\ell \in L_v} do
            \begin{enumerate}
                \item While \inline{a_\ell \leq \card{C_\ell}} do
                \begin{enumerate}
                    \item Set \inline{u \gets C_\ell[a_\ell]} and \inline{a_\ell \gets a_\ell + 1}.
                    \item If \inline{I_u = 1}, continue.
                    \item Set \inline{I_u \gets 1} and \inline{\parent{u} \gets v}.
                    \item \aref{\textsc{DFS-Visit}}{fig:sub_dfs}\inline{\br{u}}.
                \end{enumerate}
            \end{enumerate}            
        \end{enumerate}
        \end{mdframed}
        \caption{Recursive subroutine.}
        \label{fig:sub_dfs}
    \end{subfigure}

    \caption{Depth-first search forest using a clique cover.}
    \label{fig:dfs}
\end{figure}

\begin{lemma}
\label{lemma:dfs}
Let \inline{\C} be a clique cover of \inline{G}. Algorithm~\aref{\textsc{DFS}}{fig:algo_dfs} computes parent pointers \inline{\brcu{\parent{v}}_{v \in V}} that define a depth-first search forest of \inline{G} in \inline{\bigO{\ts{\C}}} time and space.
\end{lemma}
\begin{proof}
Each vertex is visited at most once and receives a parent only upon its first discovery. If \inline{\parent{u} = v}, then \inline{u} and \inline{v} lie in a common clique, so \inline{\brcu{u,v} \in E}. Since parents always point from a newly discovered vertex to an already discovered vertex, no cycle can form. The top-level routine eventually visits every vertex. Hence the parent pointers form a spanning forest. 

For each clique \inline{C_\ell}, every position preceding \inline{a_\ell} has already been inspected, and its vertex discovered. 
Thus, if \inline{v} is active and \inline{u} is an undiscovered neighbor of \inline{v}, then \inline{u} occurs at or beyond the current cursor of some clique containing both, and that clique is fully processed before the visit to \inline{v} returns. Therefore \inline{u} is discovered before backtracking from \inline{v}, and the recursive execution makes \inline{u} a descendant of \inline{v}. Thus every edge of \inline{G} joins an ancestor and a descendant in the parent forest, giving the depth-first property.

Each incidence list \inline{L_v} is scanned once, giving \inline{\sum_{v \in V} \card{L_v} = \ts{\C}} iterations. Each cursor \inline{a_\ell} only advances, so the total number of clique positions inspected is also \inline{\ts{\C}}. The incidence lists, cursors, discovery array, and recursion stack use \inline{\bigO{\ts{\C}}} space. Hence the algorithm uses \inline{\bigO{\ts{\C}}} time and space.
\end{proof}
\section{Empirical Evaluation}
\label{sec:evals}
This section presents evaluation results for clique cover construction algorithms and cover-parsimonious algorithms. Our code is available at \url{https://github.com/ahammed-ullah/ccc4pa}. Additional details appear in \hyref{Appendix}{sec:evals_deferred}.

\subsection{Graph Instances and Experimental Setup}
\label{sec:graphs}

\begin{table}[!hbtp]
    \centering
    \caption{Summary of the graph instances. Each collection contains six graphs.}
    \begin{tabular}{@{}llcc@{}}
        \toprule
        Type & Collection & \(\lvert E\rvert\) (billions) & Edge density (\%) \\
        \midrule

        \multirow{5}{*}{Real}
        
        & SuiteSparse
        & \(0.02\)--\(0.16\)
        & \(0.007\)--\(0.05\) \\
        
        & Affiliation-Sparse
        & \(0.02\)--\(15\)
        & \(0.01\)--\(0.64\) \\

        & AffiliationSX-LD
        & \(0.23\)--\(24\)
        & \(4.1\)--\(6.3\) \\

        & AffiliationSX-MD
        & \(0.10\)--\(4.5\)
        & \(6.3\)--\(11\) \\

        & AffiliationSX-HD
        & \(0.03\)--\(0.64\)
        & \(11\)--\(73\) \\

        \midrule

        \multirow{3}{*}{Synthetic}

        & RI-Power-Law
        & \(0.25\)--\(7.3\)
        & \(0.02\)--\(0.44\) \\

        & RI-Uniform
        & \(0.25\)--\(8\)
        & \(0.02\)--\(0.48\) \\

        & RI-Lognormal
        & \(0.26\)--\(8.1\)
        & \(0.02\)--\(0.48\) \\

        \bottomrule
    \end{tabular}
    \label{tab:graphs}
\end{table}

Table~\ref{tab:graphs} summarizes the 48 graph instances, organized by data source, generation model, and edge density. Each collection contains six graphs, with edge counts ranging from sixteen million to twenty-four billion. The 30 real instances consist of six graphs from the SuiteSparse Matrix Collection~\cite{davis2011university} and 24 affiliation graphs obtained from product-review and Stack Exchange data \cite{Benson-2018-simplicial, Fu-Yu-2019-sx-tags, ni2019justifying, Veldt-2020-local}. The affiliation graphs are projections of product--reviewer, question--user, discussion-thread--participant, and question--tag relations, where two objects are adjacent when they share an affiliation. The 18 synthetic instances are random intersection graphs obtained by sampling possibly overlapping vertex sets and making each set a clique, with the set sizes per instance either fixed, lognormally distributed, or power-law distributed. Additional details appear in \hyref{Appendix}{subsec:graphs_deferred}.

We ran all experiments on a compute cluster \cite{McCartney2014}, using AMD Milan processors running at 2.2 GHz with 256--1024~GB of memory. All runs were single-threaded. We implemented the algorithms in \texttt{C++} and compiled the code using \textit{gcc} 12.2.0 with the \texttt{-O3} optimization flag.

\subsection{Construction Quality and Time}
\label{subsec:cc_con_evals}

We summarize the evaluation results of the construction algorithms from \hyref{Section}{sec:concise}. 
In \hyref{Table}{tab:summary_cons},
\textsc{Lov{\'a}sz-FirstFit} denotes the variant of \aref{Lov{\'a}sz's construction}{def:lov_con} in which the maximum-clique peeling step is replaced by a first-fit greedy coloring of \inline{\overline{G}}. The remaining constructions are as described in \hyref{Definition}{def:bcc}, \fref{fig:algo_fpa}, and \fref{fig:algo_cpa}. We compare the concise constructions FPA and CPA with their respective baselines LZF and BCC, as well as with the FM biclique partitioning method \cite{FederMotwani1995}. For each graph–algorithm pair, we report the median construction time over five runs.

\begin{table}[h]
    \centering
    \caption{Conciseness guarantees and running times of clique-cover construction algorithms.}
    \begin{tabular}{ccl}
        \toprule
        Construction & Concise? & Time \\
        \midrule

        \aref{\textsc{Lov{\'a}sz-FirstFit (LZF)}}{def:lov_con}
        & \multirow{2}{*}{NO}
        & \multirow{2}{*}{\(\bigOtilde{m}\)} \\

        \aref{\textsc{Balanced Cross-Cover (BCC)}}{def:bcc}
        & & \\

        \cmidrule(lr){1-3}

        \aref{\textsc{FirstFit-PA (FPA)}}{fig:algo_fpa}
        & \multirow{2}{*}{YES}
        & \multirow{2}{*}{\(\bigOtilde{d^2\min\{m,\mu\}}\)} \\

        \aref{\textsc{CrossCover-PA (CPA)}}{fig:algo_cpa}
        & & \\

        \bottomrule
    \end{tabular}
    \label{tab:summary_cons}
\end{table}

We define \emph{compression ratio} of a clique cover \inline{\C} as \inline{2m/\ts{\C}}. For a biclique partition \inline{\B} of the edges of \inline{G}, the total size of \inline{\B} is \inline{\ts{\B}:=\sum_{B_\ell \in \B} \card{B_\ell}}, where \inline{\card{B_\ell}=\card{V_\ell^L}+\card{V_\ell^R}} denotes the order of biclique \inline{B_\ell}. \emph{Compression ratio} of a biclique partition \inline{\B} is \inline{2m/\ts{\B}}. Feder and Motwani describe a method of partitioning the edges of a graph across a set of bicliques \inline{\B} with total size \inline{\ts{\B}:=\bigO{m\br{\log\br{n^2/m}}/\log n}} \cite{FederMotwani1995} (for a description of the FM method see \hyref{Appendix}{subsec:related_gc}).

\input{tikz/evals_con/data}

\begin{figure}[h]
\centering
\begin{tikzpicture}

\begin{groupplot}[
    group style={
        group size=2 by 1,
        horizontal sep=1.0cm,
    },
    width=0.49\linewidth,
    height=0.34\linewidth,
    ymode=linear,
    ymin=0,
    ymax=100,
    ytick={0,20,40,60,80,100},
    tick align=outside,
    xmajorgrids=true,
    ymajorgrids=true,
    grid style={draw=gray!25},    
]

\nextgroupplot[
    xmode=log,
    xmin=0.8,
    xmax=30000,
    xlabel={Compression ratio (log scale)},
    legend to name=ccdflegend,
    legend columns=5,
    legend style={
        draw=none,
        /tikz/every even column/.append style={column sep=0.8em},
    },
]

\addplot[
    color=OliveGreen!55,
    dashed,
    const plot,
    mark=none,
    line width=0.9pt
] table[x=LZF_x,y=LZF_y] {\CRcdf};
\addlegendentry{LZF}

\addplot[
    color=blue!55,
    dashed,
    const plot,
    mark=none,
    line width=0.9pt
] table[x=BCC_x,y=BCC_y] {\CRcdf};
\addlegendentry{BCC}

\addplot[
    color=OliveGreen,
    solid,
    const plot,
    mark=none,
    line width=1.0pt
] table[x=FPA_x,y=FPA_y] {\CRcdf};
\addlegendentry{FPA}

\addplot[
    color=blue,
    solid,
    const plot,
    mark=none,
    line width=1.0pt
] table[x=CPA_x,y=CPA_y] {\CRcdf};
\addlegendentry{CPA}

\addplot[
    color=BrickRed,
    dashed,
    const plot,
    mark=none,
    line width=0.9pt
] table[x=FM_x,y=FM_y] {\CRcdf};
\addlegendentry{FM}

\nextgroupplot[
    xmode=log,
    xmin=0.1,
    xmax=100000,
    xlabel={Construction time (seconds)},
    yticklabels=\empty,
]

\addplot[
    color=OliveGreen!55,
    dashed,
    const plot,
    mark=none,
    line width=0.9pt
] table[x=LZF_x,y=LZF_y] {\Timecdf};

\addplot[
    color=blue!55,
    dashed,
    const plot,
    mark=none,
    line width=0.9pt
] table[x=BCC_x,y=BCC_y] {\Timecdf};

\addplot[
    color=OliveGreen,
    solid,
    const plot,
    mark=none,
    line width=1.0pt
] table[x=FPA_x,y=FPA_y] {\Timecdf};

\addplot[
    color=blue,
    solid,
    const plot,
    mark=none,
    line width=1.0pt
] table[x=CPA_x,y=CPA_y] {\Timecdf};

\addplot[
    color=BrickRed,
    dashed,
    const plot,
    mark=none,
    line width=0.9pt
] table[x=FM_x,y=FM_y] {\Timecdf};

\end{groupplot}

\node[rotate=90, anchor=south] at ($(group c1r1.west)+(-1.1cm,0)$)
{Graph instances (\%)};

\node[anchor=south]
at ($(group c1r1.north)!0.5!(group c2r1.north)+(0,0.2cm)$)
{\pgfplotslegendfromname{ccdflegend}};

\end{tikzpicture}

\caption{Cumulative distributions of compression ratio and construction time over all 48 graph instances.
The left plot shows the fraction attaining at least a given
compression ratio, and the right the fraction constructed within
a given time.}
\label{fig:comp_cr_time}
\end{figure}

\fref{fig:comp_cr_time} shows a clear separation in cover quality between the concise constructions and their baselines. FPA achieves the highest compression ratios overall, followed by CPA, while LZF and BCC remain substantially lower over most instances. The geometric-mean compression ratios are approximately \inline{731\times} for FPA and \inline{562\times} for CPA, compared with \inline{3.1\times} and \inline{4.8\times} for LZF and BCC, respectively. The separation is even more pronounced on the instance with highest FPA compression ratio: FPA and CPA achieve compression ratios of approximately \inline{21,762\times} and \inline{17,624\times}, respectively, while LZF and BCC achieve only \inline{4.87\times} and \inline{23.4\times} on the same graph.  The higher compression of FPA and CPA comes at increased construction time relative to the baselines, but both constructions exhibit empirical scaling consistent with \inline{\bigOtilde{dm}}, rather than the quadratic dependence on \inline{d} in their worst-case runtime bounds.

The FM compression ratios are the lowest among all methods compared in \fref{fig:comp_cr_time} on every graph instance. FM achieves a geometric mean compression ratio of only \inline{1.21\times}, and expands the representation on 32 of the 48 instances, as the biclique reduction is often insufficient to offset the initial duplication of each original edge. Its maximum compression ratio is \inline{3.85\times}, attained on an instance for which LZF achieves \inline{112\times} and its concise counterpart FPA achieves \inline{2559\times}. Although the FM representation is information-theoretically optimal up to a constant factor for graphs of any fixed density, this guarantee has no discernible practical effect on compression in our evaluation: FM is strictly dominated by the weakest clique-cover baseline on every instance, including the denser instances, where its theoretical size bound predicts the strongest compression.

\begin{table}[h]
\centering
\caption{A posteriori approximation-ratio bounds over the 48 graph instances.}
\begin{tabular}{@{}lcc@{\hspace{2.0em}}lcc@{}}
\toprule
Bound & LZF & BCC & Bound & FPA\inline{^+} & CPA\inline{^+} \\
\cmidrule(r){1-3}\cmidrule(l){4-6}
\(\le10\)   & 1/48  & 5/48
& \(\le1.001\) & 25/48 & 23/48 \\
\(\le100\)  & 13/48 & 13/48
& \(\le1.25\)  & 45/48 & 45/48 \\
\(\le1000\) & 38/48 & 46/48
& \(\le2\)     & 47/48 & 47/48 \\
\(\le5000\) & 48/48 & 48/48
& \(\le3.5\)   & 48/48 & 48/48 \\
\bottomrule
\end{tabular}
\label{tab:apx_ub}
\end{table}

We next assess the quality of the constructed covers relative to assignment-optimal covers. For each graph instance, we compute several empirical lower bounds on the total size of an assignment-optimal cover (see \hyref{Appendix}{subsec:cc_con_evals_deferred})  and take their maximum. Dividing the total size of a constructed cover by this value gives an a posteriori upper bound on its approximation ratio. \hyref{Table}{tab:apx_ub} shows these bounds for LZF, BCC, FPA\inline{^+}, and CPA\inline{^+}, where the \inline{^+} suffix denotes assignment-minimal post-processing. Applying the same post-processing to the baseline covers gives no meaningful reduction in total size. For CPA, the post-processing substantially improves the approximation-ratio bounds on some instances, whereas the improvements for FPA are modest: its largest bound decreases from \inline{3.63} to \inline{2.88}, while all remaining bounds are already below \inline{2} before post-processing. Thus, FPA is close to assignment-optimal across our graph instances, despite the fact that, for some \inline{\epsilon >0}, obtaining an \inline{n^\epsilon}-approximation to an assignment-optimal cover is NP-hard.

\hyref{Appendix}{subsec:cc_con_evals_deferred} provides additional details, including collection-wise compression ratios, the effect of density and average degree on compression, construction-time scaling, assignment-minimal post-processing and construction, and the empirical lower bounds used for the a posteriori approximation ratios.

\subsection{Performance of Cover-Parsimonious Algorithms}
\label{subsec:cp_algo_evals}

We summarize the evaluation results of five representative cover-parsimonious algorithms: the \inline{3/2}-approximation for maximum-cardinality matching (\aref{MCM}{fig:algo_matching2}), the \inline{\br{3/2+\epsilon}}-approximation for maximum vertex-weighted matching (\aref{MVM}{fig:algo_mvm}), \aref{BFS}{fig:algo_bfs}, \aref{DFS}{fig:algo_dfs}, and the \inline{2}-approximation for minimum weighted vertex cover (\aref{WVC}{fig:algo_wvc}). We compare these implementations against corresponding algorithms operating on compressed adjacency lists produced by WebGraph \cite{boldi2004webgraph} and Zuckerli \cite{versari2020zuckerli}, two state-of-the-art graph compression schemes (see \hyref{Appendix}{subsec:related_gc} for brief descriptions).

For the cover-parsimonious implementations, we use two clique-cover representations. \emph{Plain} stores each clique as a list of vertex identifiers. \emph{VByte} sorts the vertex identifiers within each clique, gap-encodes the resulting sequence, and applies variable-byte encoding to the gaps. The latter choice is intended to demonstrate that even simple compression schemes can further reduce memory usage while retaining the substantial running-time advantage of the clique-cover representation.

For the compressed-adjacency implementation of MVM, we use the randomized \inline{\br{3/2+\epsilon}}-approximation algorithm of Pettie and Sanders~\cite{pettie2004simpler}. Other MVM algorithms are either impractical at this scale or require operations incompatible with direct execution on compressed adjacency lists. For example, the algorithm of Al-Herz and Pothen~\cite{al20222} requires sorting adjacency lists according to vertex weights. Both the cover-parsimonious and compressed-adjacency MVM implementations use \inline{\epsilon=0.01}.

For each graph--algorithm pair, we report the median running time and peak memory over five runs. \emph{Peak memory} is measured as the peak resident set size (RSS) of the process during algorithm execution. \emph{Running-time speedup} and \emph{peak-memory ratio} are defined as the corresponding WebGraph or Zuckerli value divided by the clique-cover value.

\begin{figure}[h]
\centering
\input{tikz/evals_con/data_cp_ratio}

\begin{tikzpicture}

\begin{groupplot}[
    group style={
        group size=2 by 1,
        horizontal sep=1.0cm,
    },
    width=0.49\linewidth,
    height=0.34\linewidth,
    ymode=linear,
    ymin=0,
    ymax=100,
    ytick={0,20,40,60,80,100},
    tick align=outside,
    xmajorgrids=true,
    ymajorgrids=true,
    grid style={draw=gray!25},
]

\nextgroupplot[
    xmode=log,
    xmin=10,
    xmax=5000000,
    xlabel={Running-time speedup (log scale)},
    legend to name=algperflegend,
    legend columns=4,
    legend style={
        draw=none,
        /tikz/every even column/.append style={column sep=0.6em},
    },
]

\addplot[
    color=BrickRed,
    solid,
    const plot,
    mark=none,
    line width=1.1pt
] table[x=WGPlain_x,y=WGPlain_y] {\TimeRatioCDF};
\addlegendentry{WebGraph/plain}

\addplot[
    color=BrickRed,
    dash pattern=on 4pt off 2pt,
    const plot,
    mark=none,
    line width=1.1pt
] table[x=WGVByte_x,y=WGVByte_y] {\TimeRatioCDF};
\addlegendentry{WebGraph/VByte}

\addplot[
    color=RoyalBlue,
    solid,
    const plot,
    mark=none,
    line width=1.1pt
] table[x=ZKPlain_x,y=ZKPlain_y] {\TimeRatioCDF};
\addlegendentry{Zuckerli/plain}

\addplot[
    color=RoyalBlue,
    dash pattern=on 4pt off 2pt,
    const plot,
    mark=none,
    line width=1.1pt
] table[x=ZKVByte_x,y=ZKVByte_y] {\TimeRatioCDF};
\addlegendentry{Zuckerli/VByte}

\nextgroupplot[
    xmode=log,
    xmin=0.1,
    xmax=200,
    xlabel={Peak-memory ratio (log scale)},
    yticklabels=\empty,
]

\addplot[
    black!45,
    dotted,
    line width=0.8pt,
    forget plot
] coordinates {(1,0) (1,100)};

\addplot[
    color=BrickRed,
    solid,
    const plot,
    mark=none,
    line width=1.1pt
] table[x=WGPlain_x,y=WGPlain_y] {\MemoryRatioCDF};

\addplot[
    color=BrickRed,
    dash pattern=on 4pt off 2pt,
    const plot,
    mark=none,
    line width=1.1pt
] table[x=WGVByte_x,y=WGVByte_y] {\MemoryRatioCDF};

\addplot[
    color=RoyalBlue,
    solid,
    const plot,
    mark=none,
    line width=1.1pt
] table[x=ZKPlain_x,y=ZKPlain_y] {\MemoryRatioCDF};

\addplot[
    color=RoyalBlue,
    dash pattern=on 4pt off 2pt,
    const plot,
    mark=none,
    line width=1.1pt
] table[x=ZKVByte_x,y=ZKVByte_y] {\MemoryRatioCDF};

\end{groupplot}

\node[rotate=90, anchor=south]
at ($(group c1r1.west)+(-1.1cm,0)$)
{Graph--algorithm pairs (\%)};

\node[anchor=south]
at ($(group c1r1.north)!0.5!(group c2r1.north)+(0,0.2cm)$)
{\pgfplotslegendfromname{algperflegend}};

\end{tikzpicture}

\caption{Cumulative distributions of running-time speedup (left) and peak-memory ratio (right) over all graph--algorithm pairs. Each pair consists of one graph instance and one of the five evaluated algorithms.}
\label{fig:cp_ratio}
\end{figure}

\fref{fig:cp_ratio} summarizes the distributions of running-time speedup and peak-memory ratio over all 240 graph--algorithm pairs. For MVM on two graphs, WebGraph and Zuckerli did not complete within their time limits of 60 and 64 hours, respectively, although valid matchings were obtained from the partial runs. For these runs, the running time and peak memory observed at termination are used, making the corresponding speedups and peak-memory ratios conservative lower bounds. The cover-parsimonious implementations achieve large running-time speedups against both compressed-adjacency implementations: the median speedups are approximately \inline{3,700\times} and \inline{4,700\times} over WebGraph with Plain and VByte, respectively, and \inline{4,300\times} and \inline{4,900\times} over Zuckerli.

In addition to the large running-time speedups, the cover-parsimonious implementations also use less peak memory on most graph--algorithm pairs. With Plain, the peak-memory ratio exceeds one for approximately \inline{82\%} of the pairs against WebGraph and \inline{81\%} against Zuckerli, with median ratios of about \inline{2.7\times} and \inline{2.3\times}, respectively. VByte further strengthens this advantage: the ratio exceeds one for approximately \inline{91\%} of the pairs against both baselines, while the median ratios increase to about \inline{4.9\times} against WebGraph and \inline{4.5\times} against Zuckerli. Thus, the cover-parsimonious implementations achieve both substantially lower running time and lower peak memory in most graph--algorithm pairs.

\begin{figure}[h]
\centering
\input{tikz/evals_con/data_cp_scatter}

\begin{tikzpicture}

\begin{groupplot}[
    group style={
        group size=2 by 2,
        horizontal sep=1.0cm,
        vertical sep=1.45cm,
    },
    width=0.46\linewidth,
    height=0.34\linewidth,
    xmode=log,
    ymode=log,
    tick align=outside,
    xmajorgrids=true,
    ymajorgrids=true,
    grid style={draw=gray!25},
]

\nextgroupplot[
    title={WebGraph},
    xmin=1e-1,
    xmax=1e6,
    ymin=1e-4,
    ymax=1e3,
    legend to name=algscatterlegend,
    legend columns=5,
    legend style={
        draw=none,
        /tikz/every even column/.append style={column sep=0.7em},
    },
]

\addplot[
    black!45,
    dotted,
    line width=.8pt,
    domain=1e-1:1e3,
    samples=2,
    forget plot
] {x};

\addplot[
    only marks,
    mark=*,
    mark size=1.6pt,
    color=BrickRed,
] table[x=WGTime,y=CPTime] {\MCMScatter};
\addlegendentry{MCM}

\addplot[
    only marks,
    mark=pentagon*,
    mark size=1.7pt,
    color=BurntOrange,
] table[x=WGTime,y=CPTime] {\MVMScatter};
\addlegendentry{MVM}

\addplot[
    only marks,
    mark=pentagon,
    mark size=1.9pt,
    color=BurntOrange,
    forget plot,
] table[x=WGTime,y=CPTime] {\MVMCensoredScatter};

\addplot[
    only marks,
    mark=square*,
    mark size=1.6pt,
    color=RoyalBlue,
] table[x=WGTime,y=CPTime] {\BFSScatter};
\addlegendentry{BFS}

\addplot[
    only marks,
    mark=triangle*,
    mark size=1.8pt,
    color=OliveGreen,
] table[x=WGTime,y=CPTime] {\DFSScatter};
\addlegendentry{DFS}

\addplot[
    only marks,
    mark=diamond*,
    mark size=1.7pt,
    color=Purple,
] table[x=WGTime,y=CPTime] {\WVCScatter};
\addlegendentry{WVC}

\nextgroupplot[
    title={Zuckerli},
    xmin=1e-1,
    xmax=1e6,
    ymin=1e-4,
    ymax=1e3,
    yticklabels=\empty,
]

\addplot[
    black!45,
    dotted,
    line width=.8pt,
    domain=1e-1:1e3,
    samples=2,
    forget plot
] {x};

\addplot[
    only marks,
    mark=*,
    mark size=1.6pt,
    color=BrickRed,
] table[x=ZKTime,y=CPTime] {\MCMScatter};

\addplot[
    only marks,
    mark=pentagon*,
    mark size=1.7pt,
    color=BurntOrange,
] table[x=ZKTime,y=CPTime] {\MVMScatter};

\addplot[
    only marks,
    mark=pentagon,
    mark size=1.9pt,
    color=BurntOrange,
    forget plot,
] table[x=ZKTime,y=CPTime] {\MVMCensoredScatter};

\addplot[
    only marks,
    mark=square*,
    mark size=1.6pt,
    color=RoyalBlue,
] table[x=ZKTime,y=CPTime] {\BFSScatter};

\addplot[
    only marks,
    mark=triangle*,
    mark size=1.8pt,
    color=OliveGreen,
] table[x=ZKTime,y=CPTime] {\DFSScatter};

\addplot[
    only marks,
    mark=diamond*,
    mark size=1.7pt,
    color=Purple,
] table[x=ZKTime,y=CPTime] {\WVCScatter};

\nextgroupplot[
    xmin=5,
    xmax=1e4,
    ymin=5,
    ymax=1e3,
]

\addplot[
    black!45,
    dotted,
    line width=.8pt,
    domain=5:1e3,
    samples=2,
    forget plot
] {x};

\addplot[
    only marks,
    mark=*,
    mark size=1.6pt,
    color=BrickRed,
] table[x=WGMem,y=CPMem] {\MCMScatter};

\addplot[
    only marks,
    mark=pentagon*,
    mark size=1.7pt,
    color=BurntOrange,
] table[x=WGMem,y=CPMem] {\MVMScatter};

\addplot[
    only marks,
    mark=pentagon,
    mark size=1.9pt,
    color=BurntOrange,
    forget plot,
] table[x=WGMem,y=CPMem] {\MVMCensoredScatter};

\addplot[
    only marks,
    mark=square*,
    mark size=1.6pt,
    color=RoyalBlue,
] table[x=WGMem,y=CPMem] {\BFSScatter};

\addplot[
    only marks,
    mark=triangle*,
    mark size=1.8pt,
    color=OliveGreen,
] table[x=WGMem,y=CPMem] {\DFSScatter};

\addplot[
    only marks,
    mark=diamond*,
    mark size=1.7pt,
    color=Purple,
] table[x=WGMem,y=CPMem] {\WVCScatter};

\nextgroupplot[
    xmin=5,
    xmax=1e4,
    ymin=5,
    ymax=1e3,
    yticklabels=\empty,
]

\addplot[
    black!45,
    dotted,
    line width=.8pt,
    domain=5:1e3,
    samples=2,
    forget plot
] {x};

\addplot[
    only marks,
    mark=*,
    mark size=1.6pt,
    color=BrickRed,
] table[x=ZKMem,y=CPMem] {\MCMScatter};

\addplot[
    only marks,
    mark=pentagon*,
    mark size=1.7pt,
    color=BurntOrange,
] table[x=ZKMem,y=CPMem] {\MVMScatter};

\addplot[
    only marks,
    mark=pentagon,
    mark size=1.9pt,
    color=BurntOrange,
    forget plot,
] table[x=ZKMem,y=CPMem] {\MVMCensoredScatter};

\addplot[
    only marks,
    mark=square*,
    mark size=1.6pt,
    color=RoyalBlue,
] table[x=ZKMem,y=CPMem] {\BFSScatter};

\addplot[
    only marks,
    mark=triangle*,
    mark size=1.8pt,
    color=OliveGreen,
] table[x=ZKMem,y=CPMem] {\DFSScatter};

\addplot[
    only marks,
    mark=diamond*,
    mark size=1.7pt,
    color=Purple,
] table[x=ZKMem,y=CPMem] {\WVCScatter};

\end{groupplot}


\node[anchor=north]
at ($(group c1r1.south)!0.5!(group c2r1.south)+(0,-0.55cm)$)
{Compressed-adjacency running time (s)};

\node[rotate=90, anchor=south]
at ($(group c1r1.west)+(-1.05cm,0)$)
{Clique-cover running time (s)};

\node[anchor=north]
at ($(group c1r2.south)!0.5!(group c2r2.south)+(0,-0.55cm)$)
{Compressed-adjacency peak memory (MB)};

\node[rotate=90, anchor=south]
at ($(group c1r2.west)+(-1.05cm,0)$)
{Clique-cover peak memory (MB)};

\node[anchor=south]
at ($(group c1r1.north)!0.5!(group c2r1.north)+(0,0.55cm)$)
{\pgfplotslegendfromname{algscatterlegend}};

\end{tikzpicture}

\caption{Instance-wise running time and peak memory of the VByte cover-parsimonious implementations compared with WebGraph and Zuckerli. Each point represents one graph--algorithm pair. The top row shows running time, the bottom row peak memory, and the columns correspond to WebGraph and Zuckerli.}
\label{fig:cp_scatter}
\end{figure}

\fref{fig:cp_scatter} shows the instance-wise results for the VByte representation, separating the five algorithms. The running-time points lie well below the diagonal across all five algorithms and both compressed-adjacency baselines, confirming that the large aggregate speedups in \fref{fig:cp_ratio} are broadly consistent across algorithms and graph instances. The peak-memory advantage is also clear: most points lie below the diagonal. For MVM on the two instances where WebGraph and Zuckerli reached their time limits, open markers show the measurements at termination; the corresponding compressed-adjacency running times and peak-memory measurements are lower bounds on those of completed executions.

Additional details on WebGraph and Zuckerli construction and parameter selection, encoding-size and construction-time comparisons, and solution quality appear in \hyref{Appendix}{subsec:cp_algo_evals_deferred}.
\section{Conclusion}
\label{sec:conc}

Algorithms operating on compressed graph representations often trade running time for space. We show that clique-cover representations can avoid this tradeoff for graphs containing large cliques or overlapping collections of cliques. Our concise-cover constructions produce compact representations, while our cover-parsimonious algorithms operate directly on these representations in near-linear time in the cover size. Our experiments show substantial reductions in both running time and memory usage compared with algorithms operating on compressed adjacency lists.

This work initiates the systematic study of cover-parsimonious algorithms for fundamental graph problems. Further work could extend these algorithms to data-mining problems such as clustering, densest-subgraph computation, and core decomposition. More generally, clique covers may improve running-time bounds even when the resulting algorithms are not cover-parsimonious. Given an adjacency-list algorithm with running time \inline{\bigO{T(n,m)}}, where \inline{T} explicitly captures the dependence on the number of edges, when can the same solution guarantees be obtained directly from a clique cover \inline{\C} in \inline{\bigOtilde{T(n,\ts{\C})}} time?

A key step toward broader adoption of clique-cover representations and cover-parsimonious algorithms is to establish a public collection of graphs represented by concise clique covers, analogous to the SuiteSparse Matrix Collection~\cite{davis2011university}.
Such a collection would support graph analytics directly on compressed representations, without requiring repeated construction of clique covers.
Expanding it to graphs substantially larger than those evaluated in this work would require faster sequential and parallel cover-construction algorithms.

\section*{AI Usage Disclosure}
ChatGPT was used to assist with literature searches, language editing, exposition, and reviewing technical arguments.
OpenAI Codex was used for author-guided algorithm implementation, debugging, experimental data processing, figure generation, and artifact documentation.
The authors take full responsibility for the accuracy, integrity, and originality of this work.

\bibliographystyle{plainurl}
\typeout{}
\bibliography{7_ref}

@article{kou1978covering,
  title={Covering edges by cliques with regard to keyword conflicts and intersection graphs},
  author={Kou, Lawrence T. and Stockmeyer, Larry J. and Wong, Chak-Kuen},
  journal={Communications of the ACM},
  volume={21},
  number={2},
  pages={135--139},
  year={1978},
  publisher={ACM}
}

@inproceedings{orlin1977contentment,
  title={Contentment in graph theory: covering graphs with cliques},
  author={Orlin, James},
  booktitle={Indagationes Mathematicae (Proceedings)},
  volume={80},
  pages={406--424},
  year={1977},
  organization={Elsevier}
}

@article{ullah2022computing,
  title={Computing clique cover with structural parameterization},
  author={Ullah, Ahammed},
  journal={arXiv preprint arXiv:2208.12438},
  year={2022}
}

@article{lund1994hardness,
  title={On the hardness of approximating minimization problems},
  author={Lund, Carsten and Yannakakis, Mihalis},
  journal={Journal of the ACM (JACM)},
  volume={41},
  number={5},
  pages={960--981},
  year={1994},
  publisher={ACM New York, NY, USA}
}

@article{davoodi2016edge,
  title={Edge clique covering sum of graphs},
  author={Davoodi, Akbar and Javadi, Ramin and Omoomi, Behnaz},
  journal={Acta Mathematica Hungarica},
  volume={149},
  number={1},
  pages={82--91},
  year={2016},
  publisher={Springer}
}

@inproceedings{lovasz1968covering,
  title={On covering of graphs},
  author={Lov{\'a}sz, L{\'a}szl{\'o}},
  booktitle={Theory of Graphs (Proc. Colloq., Tihany, 1966)},
  pages={231--236},
  year={1968},
  organization={Academic Press New York}
}

@article{chiba1985arboricity,
  title={Arboricity and subgraph listing algorithms},
  author={Chiba, Norishige and Nishizeki, Takao},
  journal={SIAM Journal on computing},
  volume={14},
  number={1},
  pages={210--223},
  year={1985},
  publisher={SIAM}
}

@article{matula1983smallest,
  title={Smallest-last ordering and clustering and graph coloring algorithms},
  author={Matula, David W and Beck, Leland L},
  journal={Journal of the ACM (JACM)},
  volume={30},
  number={3},
  pages={417--427},
  year={1983},
  publisher={ACM New York, NY, USA}
}

@article{hopcroft1973n,
  title={An $n^{5/2}$ algorithm for maximum matchings in bipartite graphs},
  author={Hopcroft, John E and Karp, Richard M},
  journal={SIAM Journal on computing},
  volume={2},
  number={4},
  pages={225--231},
  year={1973},
  publisher={SIAM}
}

@article{pettie2004simpler,
  title={A simpler linear time 2/3- $\varepsilon$ approximation for maximum weight matching},
  author={Pettie, Seth and Sanders, Peter},
  journal={Information processing letters},
  volume={91},
  number={6},
  pages={271--276},
  year={2004},
  publisher={Elsevier}
}

@article{al20222,
  title={A 2/3-approximation algorithm for vertex-weighted matching},
  author={Al-Herz, Ahmed and Pothen, Alex},
  journal={Discrete Applied Mathematics},
  volume={308},
  pages={46--67},
  year={2022},
  publisher={Elsevier}
}

@article{barYehudaEven1981,
  author  = {Bar-Yehuda, Reuven and Even, Shimon},
  title   = {A Linear-Time Approximation Algorithm for the Weighted Vertex Cover Problem},
  journal = {Journal of Algorithms},
  volume  = {2},
  number  = {2},
  pages   = {198--203},
  year    = {1981}
}

@article{assadi2026simple,
  title = {Simple Sublinear Algorithms for {$(\Delta+1)$} Vertex Coloring via Asymmetric Palette Sparsification},
  author={Assadi, Sepehr and Yazdanyar, Helia},
  journal={TheoretiCS},
  volume={5},
  year={2026},
  publisher={Episciences. org}
}

@article{davis2011university,
  title={The {U}niversity of {F}lorida sparse matrix collection},
  author={Davis, Timothy A and Hu, Yifan},
  journal={ACM Transactions on Mathematical Software (TOMS)},
  volume={38},
  number={1},
  pages={1--25},
  year={2011},
  publisher={ACM New York, NY, USA}
}

@inproceedings{wang2024graphs,
  title={From graphs to hypergraphs: Hypergraph projection and its reconstruction},
  author={Wang, Yanbang and Kleinberg, Jon},
  booktitle={International Conference on Learning Representations},
  volume={2024},
  pages={52745--52775},
  year={2024}
}

@article{Benson-2018-simplicial,
 author = {Benson, Austin R. and Abebe, Rediet and Schaub, Michael T. and Jadbabaie, Ali and Kleinberg, Jon},
 title = {Simplicial closure and higher-order link prediction},
 year = {2018},
 doi = {10.1073/pnas.1800683115},
 publisher = {National Academy of Sciences},
 issn = {0027-8424},
 journal = {Proceedings of the National Academy of Sciences}
}

@inproceedings{Veldt-2020-local,
  title   = {Minimizing Localized Ratio Cut Objectives in Hypergraphs},
  author  = {Veldt, Nate and Benson, Austin R. and Kleinberg, Jon},
  publisher = {{ACM} Press},
  booktitle = {Proceedings of the 26th {ACM} {SIGKDD} International Conference on Knowledge Discovery and Data Mining},
  year    = {2020}
}

@inproceedings{ni2019justifying,
  title={Justifying Recommendations using Distantly-Labeled Reviews and Fine-Grained Aspects},
  author={Ni, Jianmo and Li, Jiacheng and McAuley, Julian},
  booktitle={Proceedings of the 2019 Conference on Empirical Methods in Natural Language Processing and the 9th International Joint Conference on Natural Language Processing (EMNLP-IJCNLP)},
  pages={188--197},
  year={2019}
}

@article{Fu-Yu-2019-sx-tags,
  author = {Xiang Fu and Shangdi Yu and Austin R. Benson},
  title = {Modeling and Analysis of Tagging Networks in Stack Exchange Communities},
  journal = {Journal of Complex Networks},
  year  = {2019},
}

@article{McCartney2014,
  author  = {McCartney, Gerry and Hacker, Thomas and Yang, Baijian},
  title   = {Empowering Faculty: A Campus Cyberinfrastructure Strategy for Research Communities},
  journal = {Educause Review},
  year    = {2014}
}

@article{FederMotwani1995,
  author  = {Tom{\'a}s Feder and Rajeev Motwani},
  title   = {Clique Partitions, Graph Compression and Speeding-Up Algorithms},
  journal = {Journal of Computer and System Sciences},
  volume  = {51},
  number  = {2},
  pages   = {261--272},
  year    = {1995},
  doi     = {10.1006/jcss.1995.1065}
}

@article{chavan2026speeding,
  title={Speeding-Up Graph Algorithms via Clique Partitioning},
  author={Chavan, Akshar and Rabinia, Sanaz and Grosu, Daniel and Brocanelli, Marco},
  journal={Networks},
  volume={88},
  number={1},
  pages={85--98},
  year={2026},
  publisher={Wiley Online Library}
}

@inproceedings{krapivin2026optimal,
  title={Optimal and Efficient Partite Decompositions of Hypergraphs},
  author={Krapivin, Andrew and Przybocki, Benjamin and Sanhueza-Matamala, Nicol{\'a}s and Subercaseaux, Bernardo},
  booktitle={Proceedings of the 58th Annual ACM Symposium on Theory of Computing},
  pages={2187--2198},
  year={2026}
}

@article{ennis2012assignment,
  title={Assignment-minimum clique coverings},
  author={Ennis, John M and Fayle, Charles M and Ennis, Daniel M},
  journal={Journal of Experimental Algorithmics (JEA)},
  volume={17},
  pages={1--1},
  year={2012},
  publisher={ACM New York, NY, USA}
}

@article{erdos1966representation,
  title={The representation of a graph by set intersections},
  author={Erd{\H{o}}s, Paul and Goodman, Adolph W and P{\'o}sa, Louis},
  journal={Canadian Journal of Mathematics},
  volume={18},
  pages={106--112},
  year={1966},
  publisher={Cambridge University Press}
}

@article{DBLP:journals/dam/Roberts85,
  author       = {Fred S. Roberts},
  title        = {Applications of edge coverings by cliques},
  journal      = {Discret. Appl. Math.},
  volume       = {10},
  number       = {1},
  pages        = {93--109},
  year         = {1985},
  url          = {https://doi.org/10.1016/0166-218X(85)90061-7},
  doi          = {10.1016/0166-218X(85)90061-7},
  bibsource    = {dblp computer science bibliography, https://dblp.org}
}

@inproceedings{pullman2006clique,
  title={Clique coverings of graphs—a survey},
  author={Pullman, Norman J},
  booktitle={Combinatorial Mathematics X: Proceedings of the Conference held in Adelaide, Australia, August 23--27, 1982},
  pages={72--85},
  year={2006},
  organization={Springer}
}

@article{piepho2004algorithm,
  title={An algorithm for a letter-based representation of all-pairwise comparisons},
  author={Piepho, Hans-Peter},
  journal={Journal of Computational and Graphical Statistics},
  volume={13},
  number={2},
  pages={456--466},
  year={2004},
  publisher={Taylor \& Francis}
}

@article{gramm2007algorithms,
  title={Algorithms for compact letter displays: Comparison and evaluation},
  author={Gramm, Jens and Guo, Jiong and H{\"u}ffner, Falk and Niedermeier, Rolf and Piepho, Hans-Peter and Schmid, Ramona},
  journal={Computational Statistics \& Data Analysis},
  volume={52},
  number={2},
  pages={725--736},
  year={2007},
  publisher={Elsevier}
}

@inproceedings{chang2001tree,
  title={On the tree-degree of graphs},
  author={Chang, Maw-Shang and M{\"u}ller, Haiko},
  booktitle={International Workshop on Graph-Theoretic Concepts in Computer Science},
  pages={44--54},
  year={2001},
  organization={Springer}
}

@article{hoover1992complexity,
  title={Complexity of graph covering problems for graphs of low degree},
  author={Hoover, Douglas N},
  journal={Journal of Combinatorial Mathematics and Combinatorial Computing},
  volume={11},
  pages={187--208},  
  year={1992}
}

@article{ma1989clique,
  title={Clique covering of chordal graphs},
  author={Ma, Shaohan and Wallis, Walter D and Wu, Julin},
  journal={Utilitas Mathematica},
  volume={36},
  pages={151--152},
  year={1989},
  publisher={UTIL MATH PUBL INC UNIV MANITOBA PO BOX 7 UNIV CENT, WINNIPEG MB R3T 2N2, CANADA}
}

@article{prisner1995clique,
  title={Clique covering and clique partition in generalizations of line graphs},
  author={Prisner, Erich},
  journal={Discrete applied mathematics},
  volume={56},
  number={1},
  pages={93--98},
  year={1995},
  publisher={Elsevier}
}

@inproceedings{hevia2023solving,
  title={Solving Edge Clique Cover Exactly via Synergistic Data Reduction},
  author={Hevia, Anthony and Kallus, Benjamin and McClintic, Summer and Reisner, Samantha and Strash, Darren and Wilson, Johnathan},
  booktitle={31st Annual European Symposium on Algorithms (ESA 2023)},
  pages={61:1--61:19},
  year={2023},
  organization={Schloss Dagstuhl--Leibniz-Zentrum f{\"u}r Informatik}
}

@article{ullah2021clique,
  title={Clique cover of graphs with bounded degeneracy},
  author={Ullah, Ahammed},
  journal={CoRR, abs/2108.09851},
  year={2021}
}

@article{markham2023neuro,
  title={Neuro-causal factor analysis},
  author={Markham, Alex and Liu, Mingyu and Aragam, Bryon and Solus, Liam},
  journal={arXiv preprint arXiv:2305.19802},
  year={2023}
}

@inproceedings{wen2025hyperplr,
  title={{HyperPLR}: Hypergraph Generation through Projection, Learning, and Reconstruction},
  author={Wen, Weihuang and Yu, Tianshu},
  booktitle={The Thirteenth International Conference on Learning Representations},
  year={2025}
}

@article{wen2023w2sat,
  title={{W2SAT}: Learning to generate {SAT} instances from weighted literal incidence graphs},
  author={Wen, Weihuang and Yu, Tianshu},
  journal={arXiv preprint arXiv:2302.00272},
  year={2023}
}

@article{zhang2023fixed,
  title={A fixed-parameter tractable algorithm for combinatorial filter reduction},
  author={Zhang, Yulin and Shell, Dylan A},
  journal={arXiv preprint arXiv:2309.06664},
  year={2023}
}

@article{kellerman1973determination,
  title={Determination of keyword conflict},
  author={Kellerman, Eduardo},
  journal={IBM Technical Disclosure Bulletin},
  volume={16},
  number={2},
  pages={544--546},
  year={1973}
}

@inproceedings{gramm2006data,
  title={Data reduction, exact, and heuristic algorithms for clique cover},
  author={Gramm, Jens and Guo, Jiong and H{\"u}ffner, Falk and Niedermeier, Rolf},
  booktitle={Proceedings of the Meeting on Algorithm Engineering \& Expermiments},
  pages={86--94},
  year={2006},
  organization={Society for Industrial and Applied Mathematics}
}

@article{helling2018constructing,
  title={Constructing an indeterminate string from its associated graph},
  author={Helling, Joel and Ryan, PJ and Smyth, WF and Soltys, Michael},
  journal={Theoretical Computer Science},
  volume={710},
  pages={88--96},
  year={2018},
  publisher={Elsevier}
}

@article{conte2020large,
  title={Large-scale clique cover of real-world networks},
  author={Conte, Alessio and Grossi, Roberto and Marino, Andrea},
  journal={Information and Computation},
  volume={270},
  pages={104464},
  year={2020},
  publisher={Elsevier}
}

@article{besta2018survey,
  title={Survey and taxonomy of lossless graph compression and space-efficient graph representations},
  author={Besta, Maciej and Hoefler, Torsten},
  journal={arXiv preprint arXiv:1806.01799},
  year={2018}
}

@inproceedings{boldi2004webgraph,
  title={The {WebGraph} framework {I}: compression techniques},
  author={Boldi, Paolo and Vigna, Sebastiano},
  booktitle={Proceedings of the 13th International Conference on World Wide Web},
  pages={595--602},
  year={2004}
}

@article{versari2020zuckerli,
  title={Zuckerli: A new compressed representation for graphs},
  author={Versari, Luca and Comsa, Iulia-Maria and Conte, Alessio and Grossi, Roberto},
  journal={IEEE Access},
  volume={8},
  pages={219233--219243},
  year={2020},
  publisher={IEEE}
}

@inproceedings{buehrer2008scalable,
  title={A scalable pattern mining approach to web graph compression with communities},
  author={Buehrer, Gregory and Chellapilla, Kumar},
  booktitle={Proceedings of the 2008 international conference on web search and data mining},
  pages={95--106},
  year={2008}
}

@inproceedings{karande2009speeding,
  title={Speeding up algorithms on compressed web graphs},
  author={Karande, Chinmay and Chellapilla, Kumar and Andersen, Reid},
  booktitle={Proceedings of the Second ACM International Conference on Web Search and Data Mining},
  pages={272--281},
  year={2009}
}

@article{xu2024laconic,
  author = {Qian Xu and Juan Yang and Feng Zhang and
            Zheng Chen and Jiawei Guan and Kang Chen and
            Ju Fan and Youren Shen and Ke Yang and
            Yu Zhang and Xiaoyong Du},
  title = {Improving Graph Compression for Efficient
           Resource-Constrained Graph Analytics},
  journal = {Proceedings of the VLDB Endowment},
  volume = {17},
  number = {9},
  pages = {2212--2226},
  year = {2024},
  doi = {10.14778/3665844.3665852}
}

@inproceedings{brisaboa2009k2tree,
  author    = {Nieves R. Brisaboa and Susana Ladra and
               Gonzalo Navarro},
  title     = {$k^2$-Trees for Compact Web Graph Representation},
  booktitle = {String Processing and Information Retrieval
               (SPIRE)},
  series    = {Lecture Notes in Computer Science},
  volume    = {5721},
  pages     = {18--30},
  publisher = {Springer},
  year      = {2009},
  doi       = {10.1007/978-3-642-03784-9_3}
}

@inproceedings{besta2018loggraph,
  author = {Maciej Besta and Dimitri Stanojevic and
            Tijana Zivic and Jagpreet Singh and
            Maurice Hoerold and Torsten Hoefler},
  title = {{Log(Graph)}: A Near-Optimal High-Performance
           Graph Representation},
  booktitle = {Proceedings of the 27th International Conference
               on Parallel Architectures and Compilation
               Techniques (PACT)},
  pages = {7:1--7:13},
  year = {2018},
  publisher = {ACM},
  doi = {10.1145/3243176.3243198}
}

\appendix
\section{Deferred Details: Concise Clique Covers}
\label{sec:concise_deferred}

This section provides proofs and supporting details deferred from \hyref{Section}{sec:concise}.

\subsection{Approximation Hardness of the Assignment-Optimal Objective}
\label{subsec:apx_aopt}

\cite{lund1994hardness} showed that the problem of coloring a graph with the minimum number of colors is inapproximable within a factor of \inline{n^\epsilon} for some \inline{\epsilon > 0}. This hardness transfers to the cardinality-optimal objective through an approximation-preserving gadget reduction due to \cite{kou1978covering}. A related gadget was used by \cite{ullah2022computing} to establish the \textsf{NP}-hardness of the assignment-optimal objective, but that reduction is not approximation-preserving. We therefore modify the gadget to obtain an approximation-preserving reduction for the assignment-optimal objective.

\begin{lemma}
\label{lemma:aopt_apx}
A polynomial-time \inline{\rho}-approximation algorithm for the assignment-optimal clique-cover objective gives a polynomial-time \inline{2\rho}-approximation algorithm for minimum vertex coloring.
\end{lemma}
\begin{proof}
Let \inline{H=\br{V_H, E_H}} be an input graph for coloring, and let \inline{G=\br{V_H,E}} be the complement graph of \inline{H}. Set \inline{n:=\card{V_H}}, \inline{m:=\card{E}}, and \inline{q:=n^2}. Construct a graph \inline{\widehat{G}} from \inline{G} by adding a set \inline{X=\brcu{x_1,\dots, x_q}} of pairwise adjacent vertices, each adjacent to every vertex of \inline{V_H}.

Let \inline{\chi\br{H}} denote the chromatic number of \inline{H}. We first derive an upper bound on the total size of an assignment-optimal clique cover of \inline{\widehat{G}}. Let \inline{S_1,\dots, S_{\chi\br{H}}} be the color classes of a minimum coloring of \inline{H}. Since \inline{G} is the complement of \inline{H}, each \inline{S_j} is a clique of \inline{G}, and \inline{\sum_{j=1}^{\chi\br{H}} \card{S_j} = n}. 

For each \inline{j \in \brsq{\chi\br{H}}}, the set \inline{S_j \cup X} is a clique of \inline{\widehat{G}}. These cliques cover every edge having at least one endpoint in \inline{X}. To cover the remaining edges of \inline{G}, add one 2-vertex clique for each edge of \inline{G}. This gives a clique cover of \inline{\widehat{G}} with total size \[\sum_{j=1}^{\chi\br{H}} \br{\card{S_j}+q} + 2m =q \cdot \chi\br{H}+n+ 2m.\] 
Therefore, if \inline{\C^*} is an assignment-optimal clique cover of \inline{\widehat{G}}, then \inline{\ts{\C^*} \leq q\cdot \chi\br{H} + n + 2m}.

Apply the assumed \inline{\rho}-approximation algorithm to \inline{\widehat{G}}, and let \inline{\widehat{\C}} be the clique cover it returns. Then 
\[\ts{\widehat{\C}} \leq \rho\ts{\C^{*}} \leq \rho \br{q\cdot \chi\br{H} + n + 2m}.\]

For each \inline{i \in [q]}, define \inline{\widehat{\C}^{\br{i}} := \brcu{C_\ell \in \widehat{\C} \mid x_i \in C_\ell, C_\ell \cap V_H \ne \emptyset}}. 
For every vertex \inline{v \in V_H}, the edge \inline{\brcu{x_i,v}} belongs to \inline{\widehat{G}}, so some clique in \inline{\widehat{\C}^{\br{i}}} contains \inline{v}. Thus, the family \inline{\brcu{C_\ell \cap V_H \mid C_\ell \in \widehat{\C}^{\br{i}}}} covers \inline{V_H}. Each such set is a clique of \inline{G}, and hence an independent set of \inline{H}. Assigning each vertex of \inline{V_H} to one such set containing it therefore gives a proper coloring of \inline{H} with at most \inline{\card{\widehat{\C}^{\br{i}}}} colors. 

We now bound \inline{\card{\widehat{\C}^{\br{i}}}} for some \inline{i \in [q]}. Counting the incidences between vertices of \inline{X} and cliques of \inline{\widehat{\C}} that intersect \inline{V_H}, we obtain
\[\sum_{i=1}^q \card{\widehat{\C}^{\br{i}}}  = \sum_{C_\ell \in \widehat{\C}: C_\ell \cap V_H \ne \emptyset} \card{C_\ell \cap X} \leq \ts{\widehat{\C}}.\]

Hence, by averaging, there exists \inline{r \in [q]} such that 
\[\card{\widehat{\C}^{\br{r}}} \leq \frac{\ts{\widehat{\C}}}{q} \leq \rho \cdot \chi\br{H} + \rho \cdot \frac{n+2m}{q} \leq 2\rho \cdot \chi\br{H},\]
where the last inequality follows from \inline{n+2m \leq n^2=q} and \inline{\chi\br{H} \geq 1}. Thus, the coloring obtained from \inline{\widehat{\C}^{\br{r}}} is a \inline{2\rho}-approximation for minimum vertex coloring.
\end{proof}

By \cite{lund1994hardness}, there exists a constant \inline{c>0} such that minimum vertex coloring admits no polynomial-time \inline{n^c}-approximation unless \inline{\mathsf{P}=\mathsf{NP}}. Set \inline{\epsilon=c/4}, and suppose that minimizing \inline{\ts{\C}} on an \inline{N}-vertex graph admitted a polynomial-time \inline{N^\epsilon}-approximation. Given an \inline{n}-vertex instance of minimum vertex coloring, the construction in the proof of \lemref{lemma:aopt_apx} produces a graph \inline{\widehat{G}} with \inline{N=n+n^2} vertices. By the lemma, the assumed algorithm would therefore give a polynomial-time coloring approximation with ratio \[2N^\epsilon=2\br{n+n^2}^\epsilon \leq 2^{1+\epsilon} n^{c/2} \leq n^c\] for all sufficiently large \inline{n}, contradicting the \inline{n^c}-inapproximability of minimum vertex coloring. Hence, unless \inline{\mathsf{P}=\mathsf{NP}}, no polynomial-time \inline{N^\epsilon}-approximation exists for minimizing \inline{\ts{\C}}. This completes the proof of \hyref{Theorem}{thm:aopt_apx}.

\subsection{Partial Order Details}
\label{subsec:landscape_dd}

\paragraph{Composition-minimal $\implies$ Inclusion-minimal.} If no pair of cliques $C_i, C_j \in \C$ with $i \ne j$ can be merged into a single clique, then neither clique is contained in the other. The converse does not hold: for example, the complete graph $K_n$ admits an inclusion-minimal clique cover where each clique covers a distinct edge, yet the entire collection can be composed into a single clique.

\begin{figure}[h]
    \centering
     \tikzset{main node/.style={text=blue, circle,fill=black!20,draw,minimum size=0.8cm,inner sep=0pt},
            }
 \begin{tikzpicture}
    \node[main node] (1) {$a$};
    \node[main node] (2) [above right = 1.0cm and 1.0cm of 1] {$b$};
    \node[main node] (3) [below = 2.0cm of 2] {$c$};
    \node[main node] (4) [right = 2.5cm of 1] {$d$};
    
    \node[main node] (5) [left = 2.0cm of 3] {$f$};
    
    \node[main node] (6) [left = 2.0cm of 2] {$e$};
    
    \node[main node] (7) [right = 2.5cm of 2] {$g$};
    
    \node[main node] (8) [right = 2.5cm of 3] {$h$};

    \path[draw,thick]
    (1) edge node {} (2)
    (1) edge node {} (3)
    (2) edge node {} (3)
    (2) edge node {} (4)
    (3) edge node {} (4)
    
    (5) edge node {} (1)
    (5) edge node {} (3)
    
    (6) edge node {} (1)
    (6) edge node {} (2)
    
    (7) edge node {} (4)
    (7) edge node {} (2)
    
    (8) edge node {} (4)
    (8) edge node {} (3)
    
    ;
\end{tikzpicture}
    \caption{An example graph used to illustrate multiple concepts.}
    \label{fig:example_graph}
\end{figure}

\paragraph{Support-minimal $\implies$ Inclusion-minimal.} If each clique in a cover contains at least one edge not covered by any other, then it cannot be fully contained in another clique. The converse does not hold: for example, the induced subgraph on vertices $\{a,b,c,d,e,f\}$ in \fref{fig:example_graph} admits an inclusion-minimal clique cover consisting of $C_1=\{a,b,c\}$, $C_2=\{a,b,e\}$, $C_3=\{a,c,f\}$, and $C_4=\{b,c,d\}$, but the cliques $\{C_2, C_3, C_4\}$ alone are sufficient to cover all edges.

\paragraph{Cardinality-optimal $\implies$ Support-minimal and Composition-minimal.} If a clique cover is not support-minimal then it contains at least one clique that does not cover any edge uniquely; such a clique can be removed to yield a clique cover of smaller cardinality. If it is not composition-minimal, then it contains a pair of distinct cliques that can be merged into a single clique, again yielding a cover of smaller cardinality.

\paragraph{Assignment-minimal $\implies$ Support-minimal.} If a clique cover is assignment-minimal, then no vertex can be removed from any clique without leaving some edge uncovered. This implies that each vertex-clique assignment supports at least one unique edge not covered by any other clique, making the cover support-minimal. The converse does not hold: the graph in \fref{fig:example_graph} admits a support-minimal clique cover consisting of $C_1=\{a,b,c\}$, $C_2=\{a,b,e\}$, $C_3=\{a,c,f\}$, $C_4=\{b,d,g\}$, and $C_5=\{c,d,h\}$, but the cover is not assignment-minimal, since vertex $a$ can be removed from $C_1$ without leaving any edge uncovered. 

\paragraph{Assignment-optimal $\not\Rightarrow$ Composition-minimal.} Consider the graph obtained by adding the edge $\{a,d\}$ to the graph in \fref{fig:example_graph}. The clique cover \inline{ \C=\{\{a,b,e\},\{a,c,f\},\{b,d,g\},\{c,d,h\}, \{a,d\},\{b,c\}\}} has total size $16$. Every clique cover requires at least three assignments to cover the two edges incident on each of $e,f,g,h$, and at least four additional assignments to cover $\{a,d\}$ and $\{b,c\}$, since neither edge can be covered by a clique containing any of $e,f,g,h$. Thus $\C$ is assignment-optimal. However, the cliques $\{a,d\}$ and $\{b,c\}$ can be merged into the clique $\{a,b,c,d\}$, so $\C$ is not composition-minimal.

\paragraph{}The inclusion-minimal clique cover of $K_n$ described above is also support-minimal, showing that support-minimal $\centernot\implies$ composition-minimal. The other inclusion-minimal clique cover shown above is composition-minimal, which establishes that composition-minimal $\centernot\implies$ support-minimal. The support-minimal clique cover that is not assignment-minimal is also cardinality-optimal, since at least four distinct cliques are needed to cover the edges incident on $e,f,g,h$, and an additional clique is needed to cover $\{b,c\}$. Thus, cardinality-optimal $\centernot\implies$ assignment-minimal. Assignment-optimal and cardinality-optimal clique covers are fundamentally different notions. \cite{davoodi2016edge} showed a family of graphs with a polynomial separation in total size between the two classes.

\subsection{Deferred Proofs}
\label{subsec:deferred_pfs_cons}

\LemmaCmExp*
\begin{proof}
Consider the Moon-Moser graph on \inline{n=3k} vertices with \inline{k} parts of size three, with all edges between distinct parts. Its \inline{3^k=3^{n/3}} maximal cliques are obtained by choosing one vertex from each part. This family covers all edges. The union of any two distinct cliques contains two vertices of the same part and hence is not a clique. Thus the family of all maximal cliques is a composition-minimal clique cover with \inline{3^{n/3}=2^{\bigTheta{n}}} cliques.
\end{proof}

\LemmaScSize*
\begin{proof}
Consider the graph family used in \lemref{lemma:admissible_space_lb} and make the vertex set \inline{U} a clique instead of an independent set (keeping all other adjacencies the same). Then the following family of cliques \inline{\C} is a composition-minimal cover of this graph:
\begin{itemize}
    \item $C_{i+(j-1)n} = \brcu{a_i, b_j} \cup U$ for $i \in \brsq{n}$, $j \in [n] $
    \item $C_{n^2+i+(j-1)n} = \brcu{c_i, d_j} \cup U$ for $i \in \brsq{n}$, $j \in [n] $
    \item $C_{2n^2+i} = \brcu{a_i, c_i}$ for $i \in [n]$
    \item $C_{2n^2+n+i} = \brcu{b_i, d_i}$ for $i \in [n]$
\end{itemize}
We have \inline{\card{\C} = \bigTheta{n^2}=\bigTheta{m}} and \inline{\ts{\C}= \bigTheta{nm}}. It is straightforward to verify that \inline{\C} is composition-minimal.
\end{proof}

\LemmaConciseSep*
\begin{proof}
For \inline{k \geq 3}, let \inline{G_k} be \inline{K_{2^k}} with a matching \inline{\brcu{\brcu{a_i, b_i} \mid i \in [k]}} removed. Then \inline{n=2^k}, \inline{\mu=k+1=\bigTheta{\log n}}, and \inline{m= \bigTheta{n^2}}.

Let \inline{\C_1} consist of all maximal cliques of \inline{G_k}. Every clique in \inline{\C_1} contains all vertices outside the matching endpoints and exactly one of \inline{a_i, b_i} for each \inline{i}. Hence \inline{\card{\C_1}=2^k=n}, and \inline{\card{C_\ell}=n-k=\bigTheta{n}} for every \inline{C_\ell \in \C_1}. Therefore, \inline{\ts{\C_1} = \bigTheta{n^2} =\bigTheta{m}}, and \inline{\C_1} is composition-minimal, since distinct cliques in \inline{\C_1} cannot be composed. Thus \inline{\C_1} is concise, but it is not \inline{\br{n\mu}}-concise, since \inline{n\mu=\bigTheta{n\log n} = \smallO{\ts{\C_1}}}.

Let \inline{U} denote the vertices outside the matching endpoints, and define \[A:=U \cup \brcu{a_i: i \in [k]}, \quad B:=U \cup \brcu{b_i: i \in [k]},\] and, for each \inline{i \in [k]}, \[A_i := U \cup \brcu{b_i} \cup \brcu{a_j\mid j\ne i}, \quad B_i := U \cup \brcu{a_i} \cup \brcu{b_j\mid j\ne i}.\]
The family \inline{\C_2 =\brcu{A,B} \cup \brcu{A_i,B_i\mid i \in [k]}} consists of distinct maximal cliques and covers every edge. Hence \inline{\C_2} is composition-minimal and contains \inline{2k+2} cliques, each of size \inline{n-k}. Therefore, \inline{\ts{\C_2} = \bigTheta{nk} = \bigTheta{n\mu}}. Thus \inline{\C_2} is \inline{\br{n\mu}}-concise, but it is not \inline{\br{n\sqrt{\mu}}}-concise, since \inline{n\sqrt{\mu}=\bigTheta{n\sqrt{\log n}} = \smallO{\ts{\C_2}}}.
\end{proof}

\LemmaLac*
\begin{proof}
We show that for every pair of distinct cliques $C_i, C_j \in \C$, there exist vertices $u \in C_i$ and $v \in C_j$, with $u\ne v$, such that $\brcu{u,v} \not\in E$. Suppose, for contradiction, $\C$ contains a pair of distinct cliques $C_i, C_j$ such that for all pairs of vertices $u \in C_i$, $v \in C_j$, with $u\ne v$, we have $\brcu{u,v} \in E$. Without loss of generality, assume that in the LA construction $C_i$ was created before $C_j$. Let $\brcu{u,v}$ be the edge on which \inline{C_j} is initialized as \inline{\brcu{u,v}} by Step~2 of the construction. Since $C_i \cup \brcu{u,v}$ induces a clique in $G$, Step~1 of the construction was applicable at that point. The fact that it was not applied contradicts the assumption that $\C$ is a LA-concise cover.

Fix a vertex $v \in V$. For each clique \inline{C_\ell \in \C} with \inline{v \in C_\ell}, look at the time step when \(v\) first becomes a member of \inline{C_\ell}. At that time step, an edge \inline{\brcu{u,v}} is processed by the construction: either it creates \inline{C_\ell = \brcu{u,v}} (Step~2), or extends an existing clique \inline{C_\ell} by adding \(v\) (Step~1). In both cases, we can assign the edge \inline{\brcu{u,v}} as the witness for the pair \inline{\br{v,C_\ell}}.  Since each edge either extends or creates at most one clique, this gives an injective map of the cliques containing \(v\) to the edges incident on \(v\); that is, \inline{\card{\brcu{C_\ell \in \C \mid v \in C_\ell}} \leq \card{N(v)}}. Summing over all vertices, we have
\[\ts{\C} = \sum_{C_\ell \in \C} \card{C_\ell} = \sum_{v \in V} \card{\brcu{C_\ell \in \C \mid v \in C_\ell}} \leq \sum_{v \in V} \card{N(v)} = \bigO{m}.\]
\end{proof}

\begin{lemma}
\label{lemma:admissible_space_lb}
There exists a family of graphs \inline{G=(V,E)} with $n$ vertices and $m$ edges that admits a LA-concise clique cover \inline{\C} for which the admissible sets satisfy \inline{\ts{\A} = \bigTheta{nm}}.
\end{lemma}
\begin{proof}
Let $G_1$ and $G_2$ be complete bipartite graphs $K_{n,n}$, where $G_1$ has vertex partition $A=\brcu{a_i}_{i \in [n]}$ and $B=\brcu{b_i}_{i \in [n]}$, and $G_2$ has vertex partition $C=\brcu{c_i}_{i \in [n]}$ and $D=\brcu{d_i}_{i \in [n]}$. 

Add a perfect matching between $A$ and $C$ and between $B$ and $D$, by inserting the edges $\brcu{a_i, c_i}$ and $\brcu{b_i, d_i}$ for all $i \in \brsq{n}$.

Next, introduce an independent set $U=\brcu{u_i}_{i \in \brsq{n}}$, and connect each $u_i$ to all vertices in $A \cup B \cup C \cup D$. The resulting graph \inline{G=(V,E)} has $\bigTheta{n}$ vertices and $m=\bigTheta{n^2}$ edges.

Now consider the following LA-concise cover \inline{\C} of \inline{G}:
\begin{itemize}
    \item $C_{i+(j-1)n} = \brcu{a_i, c_i, u_j}$ for $i \in \brsq{n}$, $j \in [n] $
    \item $C_{n^2+i+(j-1)n} = \brcu{b_i, d_i, u_j}$ for $i \in \brsq{n}$, $j \in [n] $
    \item $C_{2n^2+i+(j-1)n} = \brcu{a_i, b_j}$ for $i \in \brsq{n}$, $j \in [n] $
    \item $C_{3n^2+i+(j-1)n} = \brcu{c_i, d_j}$ for $i \in \brsq{n}$, $j \in [n] $
\end{itemize}
Thus \inline{\card{\C} = \bigTheta{n^2} = \bigTheta{m}}. A suitable edge order suffices to realize \inline{\C} via the LA construction.

The first two groups of cliques have no admissible vertices outside their own members: the set $U$ is independent, and the graph does not contain any edge between $a_i$ and $d_k$, nor between $b_i$ and $c_k$.

In contrast, each clique in the last two groups, $\brcu{C_{2n^2+1}, \dots, C_{4n^2}}$, can admit any $u_k \in U$ individually.
Hence for each $u_k \in U$, all cliques in $\brcu{C_{2n^2+1}, \dots, C_{4n^2}}$ belong to its \emph{admissible set} $A_{u_k}$, so $\card{A_{u_k}} = \bigTheta{n^2} = \bigTheta{m}$.
Considering all $n$ vertices in $U$, we have \[\ts{\A}  = \sum_{v\in V\setminus U} \card{A_v}+ \sum_{k=1}^n\card{A_{u_k}} = \bigTheta{m} + \bigTheta{nm} =\bigTheta{nm}.\]
\end{proof}

\LemmaFpaSz*
\begin{proof}
We first verify conciseness. Consider the nontrivial first-fit classes \inline{F_k} in increasing order of \inline{k}. The initialization by \inline{F_k} can be realized by processing the edges of \inline{G[F_k]} using the LA rule. By the first-fit property, no edge of \inline{G[F_k]} can augment an earlier clique \inline{F_\ell}. Hence the first processed edge of \inline{G[F_k]} creates a new clique, and the remaining uncovered edges of \inline{G[F_k]} augment this clique to \inline{F_k}. Thus the initialization at Step~3 is consistent with the LA rule. The updates performed by \aref{\textsc{AugmentCover}}{fig:sub_fpa} are also consistent with the LA rule. An edge deferred from a closed phase remains inadmissible to every clique of that phase and can simply be processed in a later phase, since the LA construction is oblivious to edge ordering. Therefore, by \lemref{lemma:lac}, \inline{\C} is a concise cover satisfying \inline{\ts{\C} = \bigO{m}}.

It remains to show that \inline{\card{\C} = \bigO{\mu}}. Let \inline{q} be the number of nontrivial first-fit classes. Besides these \inline{q} initial cliques, the construction creates at most one clique for each pair \inline{\br{v,\ell}} with \inline{\ell \in T_v}: if an edge in the group associated with \inline{\br{v,\ell}} creates a clique, then every remaining edge in that group is admissible to the newly created clique. Therefore
\[\card{\C} \leq q + \sum_{v \in V} \card{T_v}.\]

We bound the second term by charging its pairs to distinct non-edges. Fix \inline{v \in F_k} and \inline{\ell \in T_v}. Since \inline{\ell <k}, the first-fit property guarantees a vertex \inline{u \in F_\ell} such that \inline{\brcu{u,v}} is a non-edge of \inline{G}. Charge the pair \inline{\br{v,\ell}} to one such non-edge. For a fixed \inline{v}, different values of \inline{\ell} give distinct charged non-edges because their other endpoints lie in different classes \inline{F_\ell}. For different vertices \inline{v}, the charged non-edges are distinct because their endpoints in the later first-fit class are different. Hence the pairs \inline{\br{v,\ell}} inject into the \inline{\mu -1} non-edges of \inline{G}, and 
\[\sum_{v \in V}\card{T_v} \leq \mu -1.\]

Any vertex \inline{v \in F_p} has a non-neighbor in each earlier class \inline{F_\ell}, giving \inline{p-1} distinct non-edges. Thus \inline{p \leq \mu}, and since \inline{q \leq p}, we obtain 
\inline{\card{\C} = \bigO{\mu}}.
The claim follows by combining the respective bounds.
\end{proof}

\LemmaEdgeSum*
\begin{proof}
Let $\pi$ be a degeneracy ordering of $V$, and for each vertex \inline{v \in V}, define 
\[N_{\pi}^{<}(v) :=\{u \in N(v) \mid \pi(u) < \pi(v)\}.\]

By definition of degeneracy, $\card{N_{\pi}^{<}(v)} \leq d$ for every $v \in V$. Rewriting the sum by the later endpoint of each edge gives
\[
\sum_{\brcu{u,v} \in E} \min \brcu{ f(u), f(v)} = \sum_{v \in V } \sum_{u \in N_{\pi}^{<}(v)} \min \brcu{ f(u), f(v)} \leq \sum_{v \in V } \card{N_{\pi}^{<}(v)} \cdot f(v) \leq d \sum_{v \in V} f(v),        
\]
where the last inequality uses \inline{f(v) \geq 0}.
\end{proof}

\LemmaFpaTight*
\begin{proof}
For \inline{k \geq 2}, let \inline{G_k} be \inline{K_{2k}} with the perfect matching \inline{\brcu{\brcu{a_i, b_i} \mid i \in [k]}} removed. Then \inline{n=2k}, \inline{\mu=k+1=\bigTheta{n}}, and \inline{m=\bigTheta{n^2}}. The complement of \inline{G_k} is a matching, so any first-fit coloring produces two classes containing opposite endpoints of every missing edge. WLOG assume that the two classes are \inline{F_1 = \brcu{a_1, \dots, a_k}} and \inline{F_2=\brcu{b_1, \dots, b_k}}. Thus \inline{T_{b_i} = \brcu{1}}, and the edges associated with \inline{\br{b_i,1}} are \inline{\brcu{b_i,a_j}} for \inline{i \ne j}.

Consider the pairs \inline{\br{b_i,1}} in the order they are first processed during an open phase. We show inductively that each such pair creates \inline{C_i = \brcu{b_i} \cup \br{F_1 \setminus \brcu{a_i}}}. For its first processed edge \inline{\brcu{b_i, a_j}}, neither \inline{F_1} nor \inline{F_2} is admissible, since \inline{a_i \in F_1} and \inline{b_j \in F_2}. By the induction hypothesis, every clique \inline{C_\ell} previously created in the same phase, for \inline{\ell \ne i}, is also inadmissible since \inline{a_i \in C_\ell}. Therefore, the first edge creates a new clique \inline{C_i}. This clique is admissible for every remaining edge \inline{\brcu{b_i,a_j}} with \inline{i \ne j}, whereas the preceding argument excludes every other active label. Hence all remaining edges augment \inline{C_i}, resulting in \inline{C_i = \brcu{b_i} \cup \br{F_1 \setminus \brcu{a_i}}} and \inline{\card{C_i} = k}. 

If a pair is encountered while a phase is closed, it is deferred until a subsequent open phase, so the same argument applies. Consequently, \inline{\ts{\C} \geq k^2 = \bigTheta{n^2} = \bigTheta{n\mu}}. The matching upper bound from \lemref{lemma:fpa_sz} completes the proof.
\end{proof}

For disjoint vertex sets \inline{X, Y \subseteq V}, let \inline{\etaf{X}{Y}} denote the number of non-edges of \inline{G} with one endpoint in \inline{X} and the other in \inline{Y}.

\begin{restatable}{lemma}{LemmaCcov}
\label{lemma:ccov}
For any two disjoint cliques \inline{X} and \inline{Y} of \inline{G}, the family \inline{\ccov{X}{Y}} covers every edge between \inline{X} and \inline{Y}, uses at most \inline{\etaf{X}{Y}+1} cliques, and has total size at most \[ \sum_{C_\ell \in \ccov{X}{Y}} \card{C_\ell} \leq \card{X} + \card{Y}\br{\etaf{X}{Y}+1}.\]    
\end{restatable}
\begin{proof}
For each equivalence class \inline{R}, both \inline{R} and \inline{Z_R} are cliques, and every vertex of \inline{R} is adjacent to every vertex of \inline{Z_R}. Hence, \inline{R \cup Z_R} is a clique of \inline{G}.

Consider an edge \inline{\br{u,v} \in E} with \inline{u \in X} and \inline{v \in Y}, and let \inline{R} be the equivalence class containing \inline{u}. Since \inline{u} is adjacent to \inline{v}, we have \inline{v \in Z_R}. Thus \inline{u} and \inline{v} both belong to \inline{R \cup Z_R}.

To bound the number of cliques in \inline{\ccov{X}{Y}}, observe that there is at most one equivalence class \inline{R} with \inline{Z_R = Y}. For every other equivalence class \inline{R} with \inline{Z_R \ne \emptyset}, choose a representative \inline{a_R \in R} and a vertex \inline{b_R \in Y \setminus Z_R}. These choices give distinct non-edges \inline{\brcu{a_R, b_R}} between \inline{X} and \inline{Y}. Therefore, the number of cliques in \inline{\ccov{X}{Y}} is at most \inline{\etaf{X}{Y}+1}.

The equivalence classes are disjoint subsets of \inline{X}, and each clique in \inline{\ccov{X}{Y}} contains at most \inline{\card{Y}} vertices from \inline{Y}. Hence the total size of \inline{\ccov{X}{Y}} is at most
\inline{\card{X} + \card{Y}\br{\etaf{X}{Y}+1}}.
\end{proof}

\LemmaBccSz*
\begin{proof}
Each \inline{P_j} is contained in a color class of \inline{\overline{G}}, so each \inline{P_j} is a clique of \inline{G}. The nontrivial parts among \inline{P_1, \dots, P_r} cover all edges whose endpoints belong to the same part. By \lemref{lemma:ccov}, \inline{\ccov{P_i}{P_j}} covers all edges between \inline{P_i} and \inline{P_j} for every \inline{1 \leq i <j \leq r}. Hence \inline{\C} is a clique cover of \inline{G}.

Since \inline{\overline{G}} has at most \inline{\mu -1} edges, it has degeneracy \inline{\bigO{\sqrt{\mu}}}. Thus, it can be colored using \inline{t=\bigO{\sqrt{\mu}}} colors. Therefore, the total number of parts is \[r = \sum_{\ell=1}^t \left\lceil \frac{\card{F_\ell}}{b} \right\rceil \leq \sum_{\ell=1}^t \br{1+ \frac{\card{F_\ell}}{b}} = t + \frac{n} {b} =\bigO{\sqrt{\mu}}.\] 

The nontrivial cliques among \inline{P_1,\dots, P_r} have total assignment cost at most \inline{\sum_{j=1}^r\card{P_j} \leq n}. By \lemref{lemma:ccov}, the families \inline{\ccov{P_i}{P_j}}, for \inline{i<j}, contribute at most
\[\sum_{i<j} \br{\card{P_i} + \card{P_j} \br{\etaf{P_i}{P_j} +1}} \leq \sum_{i<j} \br{2b + b\etaf{P_i}{P_j}} = \bigO{br^2}+ b\sum_{i<j}\etaf{P_i}{P_j} = \bigO{n\sqrt{\mu}},\]
since \inline{\sum_{i<j}\etaf{P_i}{P_j} \leq \mu -1}. Therefore \inline{\ts{\C} = \bigO{n\sqrt{\mu}}}.

For every non-trivial initial part \inline{T}, \inline{\card{T} \leq 2\binom{\card{T}}{2}}. For every clique \inline{R\cup Z_R} produced by a cross-cover, \inline{\card{R}+ \card{Z_R} \leq 2\card{R}\cdot \card{Z_R}}, and all pairs in \inline{R \times Z_R} are edges. The internal edge sets of the initial parts and the edge sets \inline{R \times Z_R} are pairwise disjoint across different parts, pairs of parts, and equivalence classes within each cross-cover. Thus \inline{\ts{\C} \leq 2m}. Since each clique in \inline{\C} corresponds to one of these nonempty edge sets, we have \inline{\card{\C} \leq m}. We also have
\inline{\card{\C} \leq r+\sum_{i<j}\br{\etaf{P_i}{P_j}+1} = \bigO{\mu + r^2} = \bigO{\mu}}.
Combining these bounds proves the lemma.
\end{proof}

\LemmaCpa*
\begin{proof}
Each \inline{B_i} is processed atomically, so introducing \inline{C_\ell := B_i} or adding \inline{B_i} to an existing \inline{C_\ell} adds at most \inline{\card{B_i}} assignments to the current cover \inline{\C}. Hence \inline{\ts{\C} \leq \ts{\B} = \bigO{\min\brcu{m, n\sqrt{\mu}}}}. Suppose for contradiction that two distinct cliques \inline{C_i, C_j \in \C} are mergeable, where \inline{C_j} was created later from \inline{B_k}. Let \inline{C_i^\prime \subseteq C_i} be the state of \inline{C_i} when \inline{B_k} was considered against label \inline{i}. Since \inline{C_i^\prime \cup B_k \subseteq C_i \cup C_j}, \inline{B_k} was admissible to \inline{C_i^\prime}. Hence \inline{B_k} would have been added to \inline{C_i}, contradicting the creation of \inline{C_j}. Thus \inline{\C} is concise.

We now account for the time and space for constructing \inline{\B} and \inline{\T} directly from \inline{G}. A degeneracy coloring of \inline{\overline{G}} can be computed in \inline{\bigO{m}} time and space by reversing the order produced by repeatedly peeling a maximum-degree vertex from the current induced subgraph of \inline{G}, as in \cite{matula1983smallest}, and implementing greedy coloring through adjacency-list scans that count previously colored neighbors by color class. 

After splitting the color classes, store with each vertex the part index containing it. Then iterate over the edges and for every crossing edge \inline{\brcu{u,v}}, where \inline{u \in P_i, v \in P_j}, and \inline{i<j}, store the record \inline{\br{i,j,u,v}}. Sorting these records lexicographically groups the sorted list \inline{N\br{u} \cap P_j}, for each relevant \inline{\br{i,j,u}}. From the resulting lists, group equal lists under each \inline{\brcu{i,j}}. Each group is exactly one equivalence class of \inline{P_i} with respect to \inline{P_j}, since \inline{u \sim_{P_j}v \iff N\br{u}\cap P_j = N\br{v} \cap P_j}.

The total number of \inline{\br{i,j,u,v}} records is \inline{\bigO{m}}, since every crossing edge appears in exactly one list. Using standard deterministic string-sorting algorithms, sorting and grouping take \inline{\bigOtilde{m}} time, hence all cross-cover cliques can be constructed in \inline{\bigOtilde{m}} time and \inline{\bigO{m}} space. Since \inline{\ts{\B} = \bigO{m}}, constructing \inline{\B} and the spanning trees \inline{\T} take \inline{\bigOtilde{m}} time and \inline{\bigO{m}} space.

For the phased-admissibility scheme, we reuse the analysis of  \aref{\textsc{FirstFit-PA}}{fig:algo_fpa} in \lemref{lemma:fpa_cx}. For a label \inline{\ell}, introduced from \inline{B_i}, choose any edge \inline{\brcu{u_\ell, v_\ell} \in E\br{T_i}} as its pivot. The pivot edges are pairwise distinct since the trees \inline{T_i} are pairwise edge-disjoint, and \inline{D_\ell \subseteq N[u_\ell] \cap N[v_\ell]}. Therefore, the pivot-edge bound used in  \aref{\textsc{FirstFit-PA}}{fig:algo_fpa} holds. An admissibility test now intersects \inline{A_u \cap A_v} over the edges \inline{\br{u,v}} of \inline{T_i}. In any phase
\[\sum_{i} \sum_{\brcu{u,v} \in E\br{T_i}} \min\brcu{\card{A_u}, \card{A_v}} \leq \sum_{\brcu{u,v} \in E} \min\brcu{\card{A_u}, \card{A_v}},\]
since the trees are pairwise edge-disjoint. Applying \lemref{lemma:edgesum} therefore gives the same bound for admissibility tests as in \aref{\textsc{FirstFit-PA}}{fig:algo_fpa}. 

For dual creation, charging the intersections to the edges of \inline{T_i} gives \inline{\bigOtilde{d\min\brcu{n\card{\C}, m}}} time. For dual updates, the tree-edge intersections take \inline{\bigOtilde{d\min\brcu{n\card{\B},m}}} total time by edge-disjointness and \inline{\card{E\br{T_i}} = \bigO{d}}, while the scans of the current duals take \inline{\bigOtilde{d \min \brcu{n\card{\C}, dm}}} total time.

Using the $M_v$-sets as in \aref{\textsc{FirstFit-PA}}{fig:algo_fpa}, the coverage tests in each phase take
\[
    \bigOtilde{\sum_{B_i\in\B}|B_i|^2}
    \subseteq \bigOtilde{b\,\ts{\B}}
    \subseteq \bigOtilde{m},
\]
where the last bound follows from \lemref{lemma:bcc_size} and the choice of $b$, noting that $\mu=\Theta(n^2)$ if $m\leq\mu$, and $m=\Theta(n^2)$ otherwise.
The remaining time and space analysis of \aref{\textsc{FirstFit-PA}}{fig:algo_fpa} holds, since \lemref{lemma:bcc_size} and atomic processing of each \inline{B_i} give \inline{\card{\C} \leq \card{\B} = \bigO{\min\brcu{m,\mu}}}.
\end{proof}

\subsection{Assignment-Minimal Covers}
\label{subsec:cc_am}

\hyref{Section}{sec:concise} explores subclasses of composition-minimal covers, including covers of total size \inline{\bigO{\min\brcu{m,n\sqrt{\mu}}}}, but these constructions do not guarantee assignment-minimality. Assignment-minimality can trivially be achieved by taking every edge as a distinct clique in the cover. However, none of the prior clique-cover constructions is assignment-minimal by construction. The following construction guarantees assignment-minimality by construction and in practice produces nontrivial compression in total size.

\begin{figure}[h]
    \centering
    \begin{parbox}{4.0in}{    
        \begin{mdframed}[linewidth=0.5pt, roundcorner=7pt, backgroundcolor=gray!5, frametitle={\underline{\textsc{DP-FirstFit}$\br{G=\br{V,E}}$}}]          
        \begin{enumerate}                                
            \item Initialize \inline{\C \gets \emptyset}. 
            \item Let $\pi=\langle v_1, \cdots, v_n \rangle$ be a degeneracy ordering of $V$, and define \inline{N_{\pi}^{<}\br{v_i} := N\br{v_i} \cap \brcu{v_1, \dots, v_{i-1}}}.
            \item For \inline{i=\card{V},\dots, 2} do    
            \begin{enumerate}         
                \item Set \inline{X \gets \brcu{u \in N_{\pi}^{<}\br{v_i} \mid \brcu{u,v_i} \text{ is uncovered}}}.
                \item If \inline{X = \emptyset}, then continue to the next iteration.
                \item Compute a first-fit greedy coloring \inline{\F} of \inline{\overline{G}[X]}.
                \item Set \inline{\C \gets \C \cup \brcu{F_\ell \cup \brcu{v_i} \mid F_\ell \in \F}}.
            \end{enumerate}
            \item Return \inline{\C}.
        \end{enumerate}
        \end{mdframed}        
    }
    \end{parbox}    
    \caption{Assignment-minimal covers via coloring of degeneracy-pivoted subgraphs.}
    \label{fig:algo_dpf}
\end{figure}

Algorithm~\aref{\textsc{DP-FirstFit}}{fig:algo_dpf} processes vertices in reverse degeneracy order. At each pivot \inline{v_i}, it considers earlier neighbors whose edges to \inline{v_i} remain uncovered and partitions them into cliques by first-fit coloring of the complement of their induced subgraph. Each resulting clique is then augmented with \inline{v_i}.

\begin{lemma}
\label{lemma:dpf}
Let \inline{G} be a graph with \inline{m} edges and degeneracy \inline{d}. Algorithm~\aref{\textsc{DP-FirstFit}}{fig:algo_dpf} returns an assignment-minimal cover \inline{\C} of \inline{G} in \inline{\bigOtilde{dm}} time and \inline{\bigO{m}} space.
\end{lemma}
\begin{proof}
Consider an edge \inline{\brcu{v_j, v_i}} with \inline{j<i}. When \inline{v_i} is processed, the edge is either covered or \inline{v_j \in X}. In the latter case, \inline{v_j} belongs to some \inline{F_\ell \in \F}, so \inline{F_\ell \cup \brcu{v_i}} covers \inline{\brcu{v_j,v_i}}. Hence the returned \inline{\C} covers every edge.

Now consider a clique \inline{C_\ell =F_\ell \cup \brcu{v_i}}, created at pivot \inline{v_i}. For every \inline{u \in F_\ell}, the edge \inline{\brcu{u,v_i}} is uncovered when \inline{C_\ell} is created. No clique created afterward can contain \inline{v_i}, because every subsequent pivot precedes \inline{v_i} in \inline{\pi}, and its cliques contain only that pivot and vertices preceding it. Thus \inline{\brcu{u,v_i}} remains uniquely covered by \inline{C_\ell}. Hence no assignments can be removed from any clique in \inline{\C} without leaving some edge uncovered, so \inline{\C} is assignment-minimal.

Let \inline{X_i} denote the set \inline{X} at pivot \inline{v_i}. Since \inline{X_i \subseteq N_{\pi}^{<}\br{v_i}}, we have \inline{\card{X_i} \leq d}. Every vertex in \inline{X_i} corresponds to a distinct uncovered pivot edge, so \inline{\sum_{i}\card{X_i} \leq m}. 
First-fit coloring of \inline{\overline{G}[X_i]} can be implemented in \inline{\bigOtilde{\card{X_i}^2}} time. Therefore, \inline{\sum_{i} \card{X_i}^2 \leq d\sum_{i} \card{X_i} \leq dm}, giving \inline{\bigOtilde{dm}} total time. The remaining steps are dominated by this cost. The coloring step can use at most \inline{\bigO{\card{X_i}^2} \subseteq \bigO{d^2} \subseteq \bigO{m}} space. The claim follows since the remaining steps use \inline{\bigO{m}} space.
\end{proof}

A second route to obtain nontrivial assignment-minimal covers is to start with a concise cover and greedily remove vertex-clique assignments while maintaining a valid cover. This post-processing step appeared previously in an application of clique covers \cite{piepho2004algorithm}. Prior work claimed an \inline{\bigO{n^3}} runtime for this step without accounting for the total size of the input cover, which does not hold for arbitrary covers \cite{gramm2007algorithms}. We instead give a size-sensitive analysis in terms of total size and degeneracy.

\begin{figure}[h]
    \centering
    \begin{parbox}{4.6in}{    
        \begin{mdframed}[linewidth=0.5pt, roundcorner=7pt, backgroundcolor=gray!5, frametitle={\underline{\textsc{AM-PostProcess}$\br{V, \C}$}}]          
        \begin{enumerate}                                
            \item Let \inline{L_v :=\brcu{\ell \mid v \in C_\ell, C_\ell \in \C}}, for each \inline{v \in V}. 
            \item For each \inline{v \in V} do    
            \begin{enumerate}         
                \item Construct \inline{r_{v}\br{u}: = \card{\brcu{\ell \in L_v \mid u \in C_\ell}}}.
                \item For each \inline{\ell \in L_v} do
                \begin{enumerate}
                    \item If \inline{r_{v}\br{u} =1} for some \inline{u \in C_\ell \setminus \brcu{v}}, then continue.
                    \item Set \inline{C_\ell \gets C_\ell \setminus \brcu{v}} and \inline{r_v\br{u} \gets r_v\br{u} -1}, for each \inline{u \in C_\ell}.
                \end{enumerate}
            \end{enumerate}
            \item Return \inline{\C \gets \brcu{ C_\ell \in \C \mid \card{C_\ell} \geq 2}}.
        \end{enumerate}
        \end{mdframed}        
    }
    \end{parbox}    
    \caption{Assignment-minimal covers via post-processing.}
    \label{fig:algo_am_pp}
\end{figure}

Algorithm~\aref{\textsc{AM-PostProcess}}{fig:algo_am_pp} greedily removes redundant vertex-clique assignments. For a fixed vertex \inline{v}, \inline{r_v\br{u}} denotes the current number of cliques containing both \inline{u} and \inline{v}. Thus \inline{v} can be removed from \inline{C_\ell} exactly when \inline{r_v\br{u} \geq 2} for every \inline{u \in C_\ell \setminus \brcu{v}}.

\begin{lemma}
\label{lemma:am_pp} 
Given a clique cover \inline{\C} of a graph \inline{G} with degeneracy \inline{d}, Algorithm~\aref{\textsc{AM-PostProcess}}{fig:algo_am_pp} returns an assignment-minimal cover \inline{\C} of \inline{G} in \inline{\bigOtilde{d\ts{\C}}} time and \inline{\bigO{\ts{\C}}} space.
\end{lemma}
\begin{proof}
An assignment \inline{\br{v,\ell}} is removed only when \inline{r_v\br{u} \geq 2} for each \inline{u \in C_\ell \setminus \brcu{v}}, thus each edge \inline{\brcu{u,v}} remains covered. If an assignment \inline{\br{v,\ell}} is kept, then necessarily there exists \inline{u \in C_\ell \setminus \brcu{v}} with \inline{r_v\br{u}=1}, so  \inline{u,v} remains uniquely covered  by \inline{C_\ell}. Hence the returned cover \inline{\C} is assignment-minimal.

The incidence lists \inline{L_v} can be constructed in \inline{\bigOtilde{\ts{\C}}} time. The remaining steps take time at most \inline{\sum_{v \in V}\sum_{\ell \in L_v} \bigO{\card{C_\ell} + \log \card{C_\ell}}\subseteq \bigOtilde{d\sum_{v \in V}\card{L_v}} = \bigOtilde{d\ts{\C}}}. The incidence lists and the counters use \inline{\bigO{n+\ts{\C}} \subseteq \bigO{\ts{\C}}} space.
\end{proof}
\section{Deferred Details: Cover-Parsimonious Algorithms}
\label{sec:cp_algo_deferred}
This section provides additional cover-parsimonious algorithms and proofs deferred from \hyref{Section}{sec:cp_algo}.

\subsection{Weighted Vertex Cover}

\begin{definition}[Minimum Weighted Vertex Cover (WVC)]
\label{def:wvc}
Let \inline{G=\br{V,E}} be a graph with a positive vertex-weight function \inline{w:V\rightarrow \mathbb{R}_{>0}}, and write
\inline{w_v:=w\br{v}} for each \inline{v\in V}. For a vertex cover \inline{S}, its weight is \inline{w\br{S}:=\sum_{v\in S} w\br{v}}. A minimum weighted vertex cover in \inline{G} is a vertex cover of minimum weight among all vertex covers in \inline{G}.
\end{definition}

We use the \inline{2}-approximation of Bar-Yehuda and Even~\cite{barYehudaEven1981}, presented through the primal-dual formulation. \fref{fig:wvc_lp} shows the standard LP relaxation, which assigns a variable \inline{x_v} to each vertex, while its dual assigns a variable \inline{y_e} to each edge.

\begin{figure}[h]
    \centering
    \begin{minipage}{0.49\textwidth}  
        \begin{mdframed}[
            linewidth=0.5pt,
            roundcorner=7pt,
            backgroundcolor=gray!5,
            frametitle={\underline{Primal LP}},
            frametitlebelowskip=-4pt
        ]
            \begin{equation*}
                \begin{array}{lll}
                \text{minimize}
                    & \displaystyle\sum\limits_{v \in V} w_v x_v \\[3pt]
                \text{subject to}
                    & \displaystyle\sum\limits_{v \in e} x_v \geq 1,
                    & \text{for all } e \in E \\[5pt]
                    & x_v \geq 0,
                    & \text{for all } v \in V
                \end{array}
            \end{equation*}
        \end{mdframed}
    \end{minipage}
    \hspace{\fill}
    \begin{minipage}{0.49\textwidth}
        \begin{mdframed}[
            linewidth=0.5pt,
            roundcorner=7pt,
            backgroundcolor=gray!5,
            frametitle={\underline{Dual LP}},
            frametitlebelowskip=-4pt
        ]
            \begin{equation*}
                \begin{array}{lll}
                \text{maximize}
                    & \displaystyle\sum\limits_{e \in E} y_e \\[3pt]
                \text{subject to}
                    & \displaystyle\sum\limits_{e \ni v} y_e \leq w_v,
                    & \text{for all } v \in V \\[5pt]
                    & y_e \geq 0,
                    & \text{for all } e \in E
                \end{array}
            \end{equation*}
        \end{mdframed}
    \end{minipage}
    \caption{The LP relaxation of the WVC problem and its dual.}
    \label{fig:wvc_lp}
\end{figure}

For a feasible dual solution \inline{y}, define the \emph{residual weight} of a vertex \inline{v} by
\inline{r_v = w_v-\sum_{e \ni v} y_e}. Thus dual feasibility requires \inline{r_v \geq 0}, and a vertex is \emph{tight} when \inline{r_v=0}.

\fref{fig:algo_wvc} describes an implementation of the primal-dual algorithm using a clique cover. Initially, \inline{r_v=w_v} for every vertex. While scanning a clique \inline{C_\ell}, the algorithm maintains at most one previously scanned vertex \inline{x} with positive residual weight. When another vertex \inline{v} with \inline{r_v>0} is encountered, it sets \inline{\delta=\min\brcu{r_x,r_v}} and subtracts \inline{\delta} from both residual weights. This corresponds to increasing the dual variable of the edge \inline{\brcu{x,v}} by \inline{\delta}. The edge exists because \inline{x,v \in C_\ell}. At least one endpoint becomes tight and is added to the cover. If the other endpoint still has positive residual weight, it is retained as \inline{x}. Consequently, after a clique is scanned, at most one of its vertices has positive residual weight.

\begin{figure}[h]
    \centering
    \begin{parbox}{4.5in}
    {
        \begin{mdframed}[
            linewidth=0.5pt,
            roundcorner=7pt,
            backgroundcolor=gray!5,
            frametitle={\underline{\textsc{Approximate-WVC}$\br{V,\C,w}$}}
        ]
        \begin{enumerate}
            \item Initialize \inline{S \gets \emptyset} and set
            \inline{r_v \gets w_v}, for all \inline{v \in V}.

            \item For each \inline{C_\ell \in \C} do
            \begin{enumerate}
                \item Set \inline{x \gets \bot}.
                \item For each \inline{v \in C_\ell} with \inline{r_v>0} do
                \begin{enumerate}
                    \item If \inline{x=\bot}, set \inline{x \gets v} and continue.
                    \item Set \inline{\delta \gets \min\brcu{r_x,r_v}},
                    \inline{r_x \gets r_x-\delta}, and
                    \inline{r_v \gets r_v-\delta}.
                    \item Add to \inline{S} every vertex in
                    \inline{\brcu{x,v}} whose residual weight becomes zero.
                    \item Set \inline{x} to the vertex in \inline{\brcu{x,v}}
                    with positive residual weight, if one exists; otherwise,
                    set \inline{x \gets \bot}.
                \end{enumerate}
            \end{enumerate}

            \item Return \inline{S}.
        \end{enumerate}
        \end{mdframed}
        }
    \end{parbox}

    \caption{A \inline{2}-approximation for minimum weighted vertex cover using a clique cover.}
    \label{fig:algo_wvc}
\end{figure}

\begin{lemma}
\label{lemma:wvc}
Let \inline{\C} be a clique cover of \inline{G}. Algorithm~\aref{\textsc{Approximate-WVC}}{fig:algo_wvc}
computes a \inline{2}-approximation for minimum weighted vertex cover in
\inline{\bigO{\ts{\C}}} time and space.
\end{lemma}
\begin{proof}
Consider the scan of a clique \inline{C_\ell}. After each processed vertex, at most one processed vertex has positive residual weight.
This holds initially, and each update preserves it because \inline{\delta=\min\brcu{r_x,r_v}} makes at least one of \inline{r_x} and
\inline{r_v} zero.
Hence, after \inline{C_\ell} is scanned, at most one of its vertices has positive residual weight. Every other vertex has zero residual weight and therefore belongs to \inline{S}. 
Thus every edge of \inline{C_\ell} has at least one endpoint in \inline{S}. Since \inline{\C} covers every edge of \inline{G}, \inline{S} is a vertex cover.

For the approximation guarantee, interpret each update on
\inline{e=\brcu{x,v}} as increasing \inline{y_e} by
\inline{\delta=\min\brcu{r_x,r_v}}. Neither residual weight becomes negative, and all other residual weights are unchanged, so the resulting dual solution remains feasible. 
Let \inline{S^{*}} be a minimum weight vertex cover. By weak duality,
\inline{\sum_{e \in E} y_e \leq \sum_{v \in S^{*}} w_v}.
Every vertex \inline{v \in S} is tight and therefore satisfies \inline{w_v=\sum_{e \ni v} y_e}. Thus
\[
    \sum_{v \in S} w_v
    =
    \sum_{v \in S}\sum_{e \ni v} y_e
    =
    \sum_{e \in E} \card{e \cap S} y_e
    \leq
    2\sum_{e \in E} y_e \leq 2 \sum_{v \in S^{*}} w_v.
\]

Each clique is scanned once, and each vertex occurrence requires constant work.
Hence the total running time is \inline{\bigO{\ts{\C}}}. The residual weights use \inline{\bigO{\card{V}} \subseteq \bigO{\ts{\C}}} space.
\end{proof}

\subsection{Maximal Independent Set}

\begin{definition}[Maximal Independent Set (MIS)]
Let \inline{G=\br{V,E}} be a graph. A set \inline{S \subseteq V} is an \emph{independent set} if no edge has both endpoints in \inline{S}. It is \emph{maximal} if no strict superset of \inline{S} is independent. Equivalently, every vertex in \inline{V\setminus S} has a neighbor in \inline{S}.
\end{definition}

A standard greedy algorithm constructs an MIS by processing the vertices in any fixed order. Whenever an eligible vertex \inline{v} is encountered, it adds \inline{v} to \inline{S} and marks all neighbors of \inline{v} as ineligible.

\fref{fig:algo_mis} describes an implementation of this algorithm using a clique cover. The indicator \inline{I_v} records whether \inline{v} is ineligible, with \inline{I_v=0} initially. When an eligible vertex \inline{v} is added to \inline{S}, every vertex in each clique containing \inline{v} is marked ineligible. Since every neighbor of \inline{v} shares at least one such clique with \inline{v}, this marks all of its neighbors.

\begin{figure}[h]
    \centering
    \begin{parbox}{4.0in}{    
        \begin{mdframed}[linewidth=0.5pt, roundcorner=7pt, backgroundcolor=gray!5,
                         frametitle={\underline{\textsc{MIS}$\br{V, \C}$}}]  
        \begin{enumerate}            
            \item Initialize \inline{S \gets \emptyset}.
            \item Set \inline{L_v \gets \brcu{\ell \mid v \in C_\ell, C_\ell \in \C}} and \inline{I_v \gets 0}, for all \inline{v \in V}.
            \item For each \inline{v \in V} with \inline{I_v = 0} do
            \begin{enumerate}
                \item Set \inline{S \gets S \cup \brcu{v}}.
                \item For each \inline{\ell \in L_v}, set \inline{I_u \gets 1} for all \inline{u \in C_\ell}.
            \end{enumerate}
            \item Return \inline{S}.
        \end{enumerate}
        \end{mdframed}    
    }
    \end{parbox}    
    \caption{Maximal independent set using a clique cover.}
    \label{fig:algo_mis}
\end{figure}

\begin{lemma}
\label{lemma:mis}
Let \inline{\C} be a clique cover of \inline{G}. Algorithm~\aref{\textsc{MIS}}{fig:algo_mis} computes a maximal independent set of \inline{G} in \inline{\bigO{\ts{\C}}} time and space.
\end{lemma}
\begin{proof}
Whenever a vertex \inline{v} is added to \inline{S}, every vertex sharing a clique with \inline{v}, and hence every neighbor of \inline{v}, is marked ineligible. Thus no vertex subsequently added to \inline{S} is adjacent to \inline{v}, so \inline{S} is independent.

Conversely, every vertex not added to \inline{S} was marked ineligible when some previously selected vertex \inline{v} was added to \inline{S}. Such a vertex shares a clique with \inline{v} and is therefore adjacent to a vertex in \inline{S}. Hence \inline{S} is maximal.

Each incidence list \inline{L_v} is scanned at most once. Moreover, each clique \inline{C_\ell} is scanned at most once: after its first scan, all vertices of \inline{C_\ell} are ineligible, so no later selected vertex belongs to \inline{C_\ell}. Therefore the total running time is \inline{\bigO{\ts{\C}}}. The incidence lists and indicators use \inline{\bigO{\ts{\C}}} space.
\end{proof}

\subsection{Vertex Coloring}

\begin{definition}[Vertex Coloring]
\label{def:coloring}
Let \inline{G=\br{V,E}} be a graph. A \emph{proper vertex coloring}
is a function \inline{\chi:V\rightarrow\brcu{1,\dots,K}} such that
\inline{\chi\br{u}\neq\chi\br{v}} for every edge
\inline{\brcu{u,v}\in E}. Such a coloring is called a
\inline{K}-coloring.
\end{definition}

Let \inline{\Delta} be the maximum degree of \inline{G}, assumed to be given, and set \inline{K:=\Delta+1}. We adapt the randomized greedy coloring procedure described in Appendix~A of Assadi and Yazdanyar~\cite{assadi2026simple} to clique-cover representations. Their algorithm processes vertices in uniformly random order and, for each vertex, samples colors independently and uniformly from \inline{\brcu{1,\dots,K}} until finding an available color.
Each sampled color is tested for conflicts with previously colored vertices using adjacency queries.

\fref{fig:algo_coloring} shows how these tests can be performed directly on a clique cover \inline{\C}. For each clique \inline{C_\ell\in\C}, the algorithm maintains a set \inline{S_\ell} containing the colors assigned to its
already colored vertices. A color \inline{a} is available at a vertex \inline{v} if and only if \inline{a\notin S_\ell} for every \inline{\ell\in L_v}, where \inline{L_v} is the set of cliques containing \inline{v}. A previously colored neighbor of \inline{v} has color \inline{a} if and only if some clique containing \inline{v} contains a vertex of that color. Thus each color trial requires at most \inline{\card{L_v}} membership tests. Once a color is accepted, it is inserted into every color set corresponding to a clique containing \inline{v}.

\begin{figure}[h]
    \centering
    \begin{minipage}{5.0in}
        \begin{mdframed}[
            linewidth=0.5pt,
            roundcorner=7pt,
            backgroundcolor=gray!5,
            frametitle={\underline{\textsc{Coloring}$\br{V,\C,\Delta}$}}
        ]
        \begin{enumerate}
            \item Set \inline{K\gets\Delta+1}. Initialize
            \inline{S_\ell\gets\emptyset} for every
            \inline{C_\ell\in\C}, and construct the incidence
            lists \inline{L_v=\brcu{\ell\mid v\in C_\ell}}
            for all \inline{v\in V}.

            \item Choose a uniformly random permutation
            \inline{\pi} of \inline{V}.

            \item For each vertex \inline{v} in the order
            \inline{\pi} do
            \begin{enumerate}
                \item Repeatedly draw a color \inline{a}
                uniformly from \inline{\brcu{1,\dots,K}}
                until \inline{a\notin S_\ell} for every
                \inline{\ell\in L_v}.

                \item Set \inline{\chi\br{v}\gets a} and
                insert \inline{a} into \inline{S_\ell}
                for every \inline{\ell\in L_v}.
            \end{enumerate}

            \item Return \inline{\chi}.
        \end{enumerate}
        \end{mdframed}
    \end{minipage}
    \caption{\inline{\br{\Delta+1}}-coloring
    using a clique cover.}
    \label{fig:algo_coloring}
\end{figure}

\begin{lemma}
\label{lemma:coloring}
Let \inline{\C} be a clique cover of a graph \inline{G} whose maximum degree \inline{\Delta} is given.
Algorithm~\aref{\textsc{Coloring}}{fig:algo_coloring} computes a proper \inline{\br{\Delta+1}}-coloring in \inline{\bigO{\ts{\C}\log^2 n}} expected time using \inline{\bigO{\ts{\C}}} space.
\end{lemma}

\begin{proof}
At every step, \inline{S_\ell} contains exactly the colors assigned to the already colored vertices of
\inline{C_\ell}. Since every edge belongs to some cover clique, the membership tests reject precisely those colors used by previously colored neighbors. Consequently, no two adjacent vertices receive the same color. Moreover, at most \inline{\Delta} colors are forbidden at any vertex, so an available color always exists and the algorithm terminates with probability one.

We bound the expected number of color trials. For each vertex \inline{v}, let \inline{T_v} be its number of trials, let \inline{d_v:=\card{N\br{v}}}, and let \inline{J_v} be the number of neighbors preceding \inline{v} in \inline{\pi}. Conditioned on \inline{J_v=j} and any color assignments to preceding vertices, at most
\inline{j} colors are forbidden. Each independent trial therefore succeeds with probability at least
\inline{\br{K-j}/K}, giving
\[
    \mathbb{E}[T_v\mid J_v=j] \leq \frac{K}{K-j}.
\]

Since \inline{\pi} is uniformly random, the relative position of \inline{v} among itself and its \inline{d_v} neighbors is uniform. Thus \inline{J_v} is uniform on \inline{\brcu{0,\dots,d_v}}, and
\[
    \mathbb{E}[T_v] \leq \frac{1}{d_v+1} \sum_{j=0}^{d_v}\frac{K}{K-j} \leq H_K,
\]
where \inline{H_K:=\sum_{i=1}^{K}1/i}. The last inequality follows because \inline{d_v\leq K-1} and
\inline{K/\br{K-j}} increases with \inline{j}; hence the average of the first \inline{d_v+1} terms is at most the average of all \inline{K} terms.

Implement each color set \inline{S_\ell} as a balanced binary search tree. Since \inline{\card{S_\ell}\leq K},
each membership test and insertion takes \inline{\bigO{\log\br{K+1}}} time.
A trial at \inline{v} requires at most \inline{\card{L_v}} membership tests, so the expected total number of tests is at most
\[
    \sum_{v\in V}\card{L_v}\, \mathbb{E}[T_v] \leq H_K\,\ts{\C}.
\]
Each vertex inserts its accepted color into every incident clique, giving exactly \inline{\ts{\C}} insertions. The incidence lists and random permutation can be constructed in \inline{\bigO{n+\ts{\C}}} time. Since
\inline{H_K=\bigO{\log K}} and \inline{K\leq n}, the expected total running time is
\inline{\bigO{\ts{\C}\log^2 n}}.

Each \inline{S_\ell} stores at most \inline{\card{C_\ell}} colors, so all color sets together use \inline{\bigO{\ts{\C}}} space.
The incidence lists also use \inline{\bigO{\ts{\C}}} space, while the permutation requires \inline{\bigO{n} \subseteq \bigO{\ts{\C}}} space, giving \inline{\bigO{\ts{\C}}} total space.
\end{proof}

\subsection{Deferred Proofs}
\label{subsec:deferred_pfs_apps}

\LemmaMVMPathPacking*
\begin{proof}
Let \inline{H=\br{V,\M \oplus \M^*}}. Every nontrivial component of \inline{H} is an alternating path or cycle. On an alternating cycle, \inline{\M} and \inline{\M^*} match the same vertices, so it produces no gain. 

Consider an alternating path component \inline{P}. List the edges of \inline{\M^* \setminus \M} that lie on \inline{E\br{P}}, in their order along \inline{P}, as \inline{e_0, e_1, \dots, e_{r-1}}. For each \inline{i \in \brcu{0, \dots, k}}, delete the edges \inline{e_j} with \inline{j\equiv i \pmod{k+1}}. The resulting components are vertex-disjoint alternating paths, each containing at most \inline{k} edges outside \inline{\M}. 

Across the \inline{k+1} choices of \inline{i}, every edge of \inline{\M^*} on \inline{P} is retained \inline{k} times and every edge of \inline{\M} on \inline{P} is retained \inline{k+1} times. Therefore, the sum of the changes in matching weight obtained by flipping the resulting path components over all \inline{k+1} choices is \[k\cdot w\br{\M^* \cap E\br{P} } - \br{k+1}w\br{\M \cap E\br{P}}.\]
By averaging, for at least one choice of \inline{i}, the total change is at least
\begin{equation}
\label{eqn:gain_avg}
\frac{k}{k+1}\cdot w\br{\M^* \cap E\br{P} } - w\br{\M \cap E\br{P}}.    
\end{equation}

From \eqref{eqn:gain_avg}, keep only the resulting path components whose flips have positive gain. Consider one such path \inline{Q}. If both endpoints of \inline{Q} are matched in \inline{\M}, then flipping \inline{Q} makes both endpoints unmatched while keeping every internal vertex matched as before, contradicting the positive gain of \inline{Q}. If both endpoints of \inline{Q} are unmatched in \inline{\M}, then \inline{Q} is an \inline{\M}-augmenting path. Since it contains at most \inline{k} edges outside \inline{\M}, its length is at most \inline{2k-1}, contradicting the assumption. 

Consequently, \inline{Q} has exactly one endpoint \inline{x} unmatched in \inline{\M} and one endpoint \inline{z} matched in \inline{\M}. Flipping \inline{Q} changes the matching weight by \inline{w\br{x} - w\br{z}} > 0. Thus \inline{Q} is an \inline{\M}-increasing path of length at most \inline{2k}.

Applying this argument to every alternating path component of \inline{H}, gives a vertex-disjoint family \inline{\A} of \inline{\M}-increasing paths of length at most \inline{2k}. By \eqref{eqn:gain_avg}, its total gain satisfies \[g\br{\A} \geq \frac{k}{k+1}w\br{\M^*} - w\br{\M}.\]
\end{proof}
\section{Deferred Details: Empirical Evaluation}
\label{sec:evals_deferred}

This section provides additional details extending the evaluation results of \hyref{Section}{sec:evals}. The complete evaluation artifact for reproducing these results is available at \url{https://github.com/ahammed-ullah/ccc4pa}.

\subsection{Graph Instance Details}
\label{subsec:graphs_deferred}

\begin{table}[!hbtp]
    \centering
    \caption{Graph instances in each collection, ordered by increasing average degree.}
    \begin{tabular}{@{}llc@{}}
        \toprule
        Collection & Graph instances & Average degree \\
        \midrule

        \multirow{2}{*}{SuiteSparse}
        & coPapersCiteseer, dielFilterV3real, audikw\_1
        & \multirow{2}{*}{\(74\)--\(161\)} \\
        & hollywood-2009, RM07R, HV15R
        & \\

        \addlinespace

        \multirow{2}{*}{Affiliation-Sparse}
        & askubuntu-th, mathoverflow-ans, stackoverflow-th
        & \multirow{2}{*}{\(348\)--\(2.7\mathrm{k}\)} \\
        & stackoverflow-ans, math-sx-th, amazon-reviews
        & \\

        \addlinespace

        AffiliationSX-LD
        & wordpress, tex, unix, serverfault, askubuntu, math
        & \(5.1\mathrm{k}\)--\(48.9\mathrm{k}\) \\

        AffiliationSX-MD
        & gamedev, scifi, physics, english, stats, superuser
        & \(4.8\mathrm{k}\)--\(24.0\mathrm{k}\) \\

        AffiliationSX-HD
        & rpg, codegolf, codereview, gis, dba, drupal
        & \(5.7\mathrm{k}\)--\(16.5\mathrm{k}\) \\

        \addlinespace

        RI-Power-Law
        & 250M, 500M, 1B, 2B, 4B, 8B
        & \(272\)--\(8.0\mathrm{k}\) \\

        RI-Uniform
        & 250M, 500M, 1B, 2B, 4B, 8B
        & \(273\)--\(8.7\mathrm{k}\) \\

        RI-Lognormal
        & 250M, 500M, 1B, 2B, 4B, 8B
        & \(278\)--\(8.8\mathrm{k}\) \\

        \bottomrule
    \end{tabular}
    \label{tab:graph_details}
\end{table}

The 30 real instances in \hyref{Table}{tab:graph_details} are collected from the SuiteSparse Matrix Collection~\cite{davis2011university} and \cite{Benson-2018-simplicial, Fu-Yu-2019-sx-tags, ni2019justifying, Veldt-2020-local}. Instance-wise statistics and preprocessing details for obtaining canonical graphs are included in the evaluation artifact.

The \emph{Affiliation-Sparse} collection contains graphs derived from product-review, question-answer, and discussion-thread data: products are adjacent when reviewed by a common user, questions are adjacent when answered by a common user, and discussion threads are adjacent when they share a participant. In the remaining 18 affiliation graphs, two Stack Exchange questions are adjacent when they share a tag. We order these 18 graphs by edge density and divide them into three six-instance collections, suffixed LD, MD, and HD.

The random intersection graphs are the synthetic analogues of the affiliation graphs. Given positive integers \inline{n_0} and \inline{r}, let \inline{\card{V_0}= n_0}. For each \inline{i \in [r]}, independently sample a size \inline{K_i}, choose a uniformly random \inline{K_i}-subset \inline{S_i \subseteq V_0}, and make \inline{S_i} a clique. The resulting edge set is \inline{\bigcup_{i=1}^r \binom{S_i}{2}}, and we obtain the canonical graph by removing isolated vertices. All 18 instances use \inline{n_0=2M}. The distributions of \inline{K_i} and the parameter ranges for the three collections are summarized in \hyref{Table}{tab:ri_params}.

\begin{table}[!hbtp]
    \centering
    \caption{Distributions of \inline{K_i} and parameter ranges used to generate the synthetic graphs.}
    \begin{tabular}{@{}lccc@{}}
    \toprule
        Collection & Distribution of \inline{K_i} & \inline{r} & Distribution parameters \\
    \midrule
        \emph{RI-Uniform}
        & \inline{K_i=c}
        & 1.56k--48.5k
        & \inline{c=102}--3.2k \\

        \emph{RI-Lognormal}
        & \inline{K_i=\left\lfloor e^{X_i} + 1/2 \right\rfloor}, \inline{X_i\sim\N\br{a,b^2}}
        & 2k--63.6k
        & \inline{a=4.25}--7.7, \inline{b=0.5} \\

        \emph{RI-Power-Law}
        & \inline{\Pr[K_i=s]\propto s^{-\alpha}},
          \inline{s_{\min}\leq s\leq s_{\max}}
        & 324k--537k
        & \inline{\alpha=s_{\min}=2}, \inline{s_{\max}=610}--31.9k \\
    \bottomrule
    \end{tabular}
    \label{tab:ri_params}
\end{table}

The graph instances in \hyref{Table}{tab:graph_details} emphasize graphs that are globally sparse but locally dense, for which clique-cover representations are most likely to provide substantial compression. Affiliation-derived graphs and random intersection graphs naturally exhibit this structure. Such structure also arises in many real-world graphs formed through shared groups, interactions, or other affiliation mechanisms. Adversarial graphs such as triangle-free graphs admit no compression, and hence are not informative for studying the potential of clique-cover representations.

\subsection{Construction Quality and Time Details}
\label{subsec:cc_con_evals_deferred}

\paragraph{Role of conciseness.} The large gap in cover quality between the baselines and FPA/CPA in \fref{fig:comp_cr_time} shows that the high compression of FPA and CPA is not explained by the theoretical size bounds inherited from LZF and BCC alone. A remaining possibility is that the baseline covers nevertheless provide useful structure for obtaining such large compression ratios. To test this, we implemented the \aref{LA construction}{def:la_con} directly using the phased-admissibility scheme, denoted \aref{Basic Phased-Admissibility (BPA)}{para:bpa}. 
By \lemref{lemma:lac}, BPA produces concise covers but does not retain the \inline{\mu}-based size bounds of FPA and CPA.
Nevertheless, its geometric-mean compression ratio is \inline{735\times}, compared with \inline{731\times} for FPA and \inline{562\times} for CPA, and BPA produces smaller covers than CPA on all 48 instances. These results suggest that the baseline structures are not responsible for the large compression ratios of FPA and CPA. The \inline{\mu}-based guarantees may still be important for extremely dense graphs, but for most real instances conciseness itself may suffice to obtain highly compressed covers.

\begin{figure}[!hbtp]
\centering
\begin{tikzpicture}

\begin{groupplot}[
    group style={
        group size=3 by 2,
        horizontal sep=0.8cm,
        vertical sep=2.0cm,
    },
    width=0.33\linewidth,
    height=0.25\linewidth,
    xmin=0.6,
    xmax=6.4,
    ymode=log,
    ymin=0.7,
    ymax=30000,
    ytick={1,10,100,1000,10000},
    xtick={1,2,3,4,5,6},
    tick align=outside,
    ymajorgrids=true,
    grid style={draw=gray!25},
    xticklabel style={
        font=\tiny,
        rotate=30,
        anchor=east,
    },
    title style={
        at={(0.5,-0.48)},
        anchor=north,
        font=\small,
    },
]


\nextgroupplot[
    title={Affiliation-Sparse},
    ylabel={Compression ratio},
    xticklabels={
        askubuntu-th,
        mathoverflow-ans,
        stackoverflow-th,
        stackoverflow-ans,
        math-sx-th,
        amazon-reviews
    },
]

\addplot[only marks, mark=o,         mark size=2.0pt, draw=OliveGreen, fill=white] table[x=xLZF,y=LZF] {\AffSparseCR};
\addplot[only marks, mark=triangle,  mark size=2.2pt, draw=blue,       fill=white] table[x=xBCC,y=BCC] {\AffSparseCR};
\addplot[only marks, mark=*,         mark size=2.0pt, draw=OliveGreen, fill=OliveGreen] table[x=xFPA,y=FPA] {\AffSparseCR};
\addplot[only marks, mark=triangle*, mark size=2.2pt, draw=blue,       fill=blue] table[x=xCPA,y=CPA] {\AffSparseCR};
\addplot[only marks, mark=x,         mark size=2.3pt, draw=BrickRed] table[x=xFM,y=FM] {\AffSparseCR};

\nextgroupplot[
    title={AffiliationSX-LD},
    yticklabels=\empty,
    xticklabels={wordpress,tex,unix,serverfault,askubuntu,math},
    legend style={
        at={(0.5,1.05)},
        anchor=south,
        legend columns=5,
        draw=none,
    },
]

\addplot[only marks, mark=o,         mark size=2.0pt, draw=OliveGreen, fill=white] table[x=xLZF,y=LZF] {\SXLDCR};
\addlegendentry{LZF}
\addplot[only marks, mark=triangle,  mark size=2.2pt, draw=blue,       fill=white] table[x=xBCC,y=BCC] {\SXLDCR};
\addlegendentry{BCC}
\addplot[only marks, mark=*,         mark size=2.0pt, draw=OliveGreen, fill=OliveGreen] table[x=xFPA,y=FPA] {\SXLDCR};
\addlegendentry{FPA}
\addplot[only marks, mark=triangle*, mark size=2.2pt, draw=blue,       fill=blue] table[x=xCPA,y=CPA] {\SXLDCR};
\addlegendentry{CPA}
\addplot[only marks, mark=x,         mark size=2.3pt, draw=BrickRed] table[x=xFM,y=FM] {\SXLDCR};
\addlegendentry{FM}

\nextgroupplot[
    title={AffiliationSX-MD},
    yticklabels=\empty,
    xticklabels={gamedev,scifi,physics,english,stats,superuser},
]

\addplot[only marks, mark=o,         mark size=2.0pt, draw=OliveGreen, fill=white] table[x=xLZF,y=LZF] {\SXMDCR};
\addplot[only marks, mark=triangle,  mark size=2.2pt, draw=blue,       fill=white] table[x=xBCC,y=BCC] {\SXMDCR};
\addplot[only marks, mark=*,         mark size=2.0pt, draw=OliveGreen, fill=OliveGreen] table[x=xFPA,y=FPA] {\SXMDCR};
\addplot[only marks, mark=triangle*, mark size=2.2pt, draw=blue,       fill=blue] table[x=xCPA,y=CPA] {\SXMDCR};
\addplot[only marks, mark=x,         mark size=2.3pt, draw=BrickRed] table[x=xFM,y=FM] {\SXMDCR};


\nextgroupplot[
    title={RI-Power-Law},
    ylabel={Compression ratio},
    xticklabels={250M,500M,1B,2B,4B,8B},
]

\addplot[only marks, mark=o,         mark size=2.0pt, draw=OliveGreen, fill=white] table[x=xLZF,y=LZF] {\RIPLCR};
\addplot[only marks, mark=triangle,  mark size=2.2pt, draw=blue,       fill=white] table[x=xBCC,y=BCC] {\RIPLCR};
\addplot[only marks, mark=*,         mark size=2.0pt, draw=OliveGreen, fill=OliveGreen] table[x=xFPA,y=FPA] {\RIPLCR};
\addplot[only marks, mark=triangle*, mark size=2.2pt, draw=blue,       fill=blue] table[x=xCPA,y=CPA] {\RIPLCR};
\addplot[only marks, mark=x,         mark size=2.3pt, draw=BrickRed] table[x=xFM,y=FM] {\RIPLCR};

\nextgroupplot[
    title={RI-Uniform},
    yticklabels=\empty,
    xticklabels={250M,500M,1B,2B,4B,8B},
]

\addplot[only marks, mark=o,         mark size=2.0pt, draw=OliveGreen, fill=white] table[x=xLZF,y=LZF] {\RIUCR};
\addplot[only marks, mark=triangle,  mark size=2.2pt, draw=blue,       fill=white] table[x=xBCC,y=BCC] {\RIUCR};
\addplot[only marks, mark=*,         mark size=2.0pt, draw=OliveGreen, fill=OliveGreen] table[x=xFPA,y=FPA] {\RIUCR};
\addplot[only marks, mark=triangle*, mark size=2.2pt, draw=blue,       fill=blue] table[x=xCPA,y=CPA] {\RIUCR};
\addplot[only marks, mark=x,         mark size=2.3pt, draw=BrickRed] table[x=xFM,y=FM] {\RIUCR};

\nextgroupplot[
    title={RI-Lognormal},
    yticklabels=\empty,
    xticklabels={250M,500M,1B,2B,4B,8B},
]

\addplot[only marks, mark=o,         mark size=2.0pt, draw=OliveGreen, fill=white] table[x=xLZF,y=LZF] {\RILNCR};
\addplot[only marks, mark=triangle,  mark size=2.2pt, draw=blue,       fill=white] table[x=xBCC,y=BCC] {\RILNCR};
\addplot[only marks, mark=*,         mark size=2.0pt, draw=OliveGreen, fill=OliveGreen] table[x=xFPA,y=FPA] {\RILNCR};
\addplot[only marks, mark=triangle*, mark size=2.2pt, draw=blue,       fill=blue] table[x=xCPA,y=CPA] {\RILNCR};
\addplot[only marks, mark=x,         mark size=2.3pt, draw=BrickRed] table[x=xFM,y=FM] {\RILNCR};

\end{groupplot}

\end{tikzpicture}
\caption{Compression ratios on representative real and synthetic graph collections. Within each panel, graph instances are ordered by increasing average degree.}
\label{fig:cwise}
\end{figure}

\paragraph{Predictors of compression.} \fref{fig:cwise} and the top plot of \fref{fig:comp_am} suggest that the compression ratios of the concise constructions are strongly associated with average degree. A general upper bound on compression ratio is \inline{2m/\ts{\C} \leq d}, while \inline{2m/n \leq 2d}. The relationship \inline{\frac{2m}{\ts{\C}} = \frac{2m/n}{\ts{\C}/n}} indicates that average degree can explain variability in compression ratios if the average number of vertex-clique assignments \inline{\ts{\C}/n} is relatively stable as average degree varies. 
Edge density \inline{\br{2m/n}/\br{n-1}} is also a strong predictor modulo graph order.
Consistent with these observations, average degree is a better predictor of the observed compression ratios than either density or degeneracy. For FPA\inline{^+}, the log-scale correlations of compression ratios with average degree, degeneracy, and density are approximately 0.97, 0.87, and 0.82, respectively. The same ordering holds for FPA and CPA, before and after assignment-minimal post-processing. The synthetic collections in \fref{fig:cwise} show the relation most clearly, with FPA compression increasing nearly in step with average degree, while the top panel of \fref{fig:comp_am} shows that the same overall trend persists across all 48 instances after assignment-minimal post-processing.

\begin{figure}[!hbtp]
\centering
\begin{tikzpicture}

\begin{groupplot}[
    group style={
        group size=1 by 2,
        vertical sep=1.25cm,
    },
    width=0.9\linewidth,
    height=0.3\linewidth,
    xmin=0,
    xmax=49,
    tick align=outside,
    xmajorgrids=true,
    ymajorgrids=true,
    grid style={draw=gray!25},
]

\nextgroupplot[
    ymode=log,
    ymin=3,
    ymax=40000,
    ylabel={Compression ratio},
    xlabel={Graph instances ordered by average degree},
    xtick={1,8,16,24,32,40,48},
    xticklabels={74,273,1.1k,2.7k,7.0k,9.1k,48.9k},
    xticklabel style={font=\footnotesize},
    legend style={
        at={(0.5,1.05)},
        anchor=south,
        legend columns=4,
        draw=none,
    },
]

\addplot[
    only marks,
    color=Sepia,
    mark=pentagon*,
    mark size=2.2pt,
]
table[x=idx,y=BPA] {\AvgDegCRData};
\addlegendentry{BPA\inline{^+}}

\addplot[
    only marks,
    color=OliveGreen,
    mark=*,
    mark size=2.2pt,
]
table[x=idx,y=FPA] {\AvgDegCRData};
\addlegendentry{FPA\inline{^+}}

\addplot[
    only marks,
    color=blue,
    mark=square*,
    mark size=2.4pt,
]
table[x=idx,y=CPA] {\AvgDegCRData};
\addlegendentry{CPA\inline{^+}}

\addplot[
    only marks,
    color=orange,
    mark=triangle*,
    mark size=2.2pt,
]
table[x=idx,y=DPF] {\AvgDegCRData};
\addlegendentry{DPF}

\nextgroupplot[
    ymode=log,
    ymin=0.5,
    ymax=40000,
    ylabel={Construction time (s)},
    xlabel={Graph instances ordered by \(\lvert E\rvert\)},
    xtick={1,8,16,24,32,40,48},
    xticklabels={16M,56M,250M,581M,1.17B,4.01B,24.3B},
    xticklabel style={font=\footnotesize},
]

\addplot[
    only marks,
    color=Sepia,
    mark=pentagon*,
    mark size=2.2pt,
]
table[x=idx,y=BPA] {\EdgeTimeData};

\addplot[
    only marks,
    color=OliveGreen,
    mark=*,
    mark size=2.2pt,
]
table[x=idx,y=FPA] {\EdgeTimeData};

\addplot[
    only marks,
    color=blue,
    mark=square*,
    mark size=2.4pt,
]
table[x=idx,y=CPA] {\EdgeTimeData};

\addplot[
    only marks,
    color=orange,
    mark=triangle*,
    mark size=2.2pt,
]
table[x=idx,y=DPF] {\EdgeTimeData};

\end{groupplot}

\end{tikzpicture}
\caption{Construction quality and time over all 48 graph instances, ordered by average degree (top) and number of edges (bottom). }
\label{fig:comp_am}
\end{figure}

\paragraph{Assignment-minimal covers.} Unlike conciseness, which strengthens composition-minimality, assignment-minimality alone is not sufficient to guarantee nontrivial compression in total size. Nevertheless, assignment-minimal covers with nontrivial compression can be obtained either directly by construction or by post-processing suitably compressed covers. We formalize these two routes in \hyref{Appendix}{subsec:cc_am} by describing DPF, which produces assignment-minimal covers directly, and a post-processing step that makes any given cover assignment-minimal. 
We apply the post-processing step to the covers produced by BPA, FPA, and CPA, and denote the resulting constructions BPA\inline{^+}, FPA\inline{^+}, and CPA\inline{^+}, respectively. 

\fref{fig:comp_am} compares the two routes to assignment-minimal covers. The top plot shows a clear advantage for post-processing concise covers. Notably, applying the same post-processing step to LZF and BCC gives no noticeable reduction in total size, despite their much larger covers. The bottom plot of \fref{fig:comp_am} shows that the construction times grow approximately linearly with the number of edges, with log-log slopes close to one for all four methods. As in \fref{fig:comp_cr_time}, the post-processed FPA and CPA show empirical scaling consistent with \inline{\bigOtilde{dm}}, rather than the quadratic dependence on \inline{d} in their worst case bounds.

\paragraph{Empirical lower bounds.} To compute the approximation-ratio bounds in \hyref{Table}{tab:apx_ub}, we computed several lower bounds on assignment-optimal covers for each graph instance. We describe two representative bounds here. Let \inline{\C^*} be an assignment-optimal cover, and let \inline{\Lcal=\brcu{L_v}_{v \in V}} denote its vertex-clique incidence lists. For a vertex-based bound, let \inline{I_v} be an independent set in \inline{G[N\br{v}]}. Every clique containing \inline{v} contains at most one vertex of \inline{I_v}, and hence \inline{\card{L_v} \geq \card{I_v}}. Constructing a greedy maximal independent set for every \inline{v} therefore gives \inline{\ts{\C^*} = \sum_{v \in V}\card{L_v} \geq \sum_{v \in V} \card{I_v}}. 

For an edge-based bound, let \inline{\phi} be a proper coloring of \inline{G}, and let \inline{r\br{v}:=\card{\brcu{\phi\br{u} \mid u \in N\br{v}}}} be the number of colors appearing in \inline{N\br{v}}. If a clique of order \inline{q} contains the edge \inline{\brcu{u,v}}, then \inline{q-1 \leq \min\brcu{r\br{u}, r\br{v}}}. Thus, assigning edge \inline{\brcu{u,v}} weight \inline{1/\min\brcu{r\br{u}, r\br{v}}}, the edges incident to \inline{v} covered by any one clique have total weight at most one. Consequently, \inline{\card{L_v} \geq \left\lceil \sum_{u \in N\br{v}} \frac{1}{\min\brcu{r\br{u}, r\br{v}}} \right\rceil}.
Combining the two vertex-wise bounds gives \inline{\ts{\C*} \geq \sum_{v \in V} \max\brcu{\card{I_v}, \left\lceil \sum_{u \in N\br{v}} \frac{1}{\min\brcu{r\br{u}, r\br{v}}} \right\rceil}}.

\subsection{Performance Details of Cover-Parsimonious Algorithms}
\label{subsec:cp_algo_evals_deferred}

\paragraph{Representations and parameter selection.} All clique covers used to evaluate the cover-parsimonious algorithms were generated by BPA+ and stored in VByte format. To store a clique \inline{\langle v_1,v_2,\dots,v_k\rangle} with \inline{v_1<v_2<\dots<v_k}, we first form the gap sequence \inline{\langle v_1,v_2-v_1,v_3-v_2,\dots,v_k-v_{k-1}\rangle}. VByte then encodes each integer using one or more bytes. Each byte holds seven data bits, and its highest bit indicates whether another byte follows.

We built WebGraph and Zuckerli encodings for all 48 graphs. We used the Rust WebGraph encoder to construct WebGraph files and implemented a C++ reader for the algorithm experiments. We implemented Zuckerli’s list-access scheme in C++, correcting defects found in its public code. To limit input reading and memory use, both encoders generated sorted neighborhoods from VByte clique covers.

We aimed for small WebGraph and Zuckerli encodings that could be constructed for all 48 graphs. An initial WebGraph search began with a reference window of 512, but most graph jobs did not finish within their time limits. Preliminary runs also favored BFS+LLP ordering and unbounded reference chains for compression. We therefore used a shorter search with unbounded chains and windows no larger than 128. We compared the original and BFS+LLP orders at windows 16 and 32, then tested windows 64 and 128 only with the order that gave the smaller encoding at window 32. We retained the smallest completed WebGraph encoding for each graph. All 48 selected encodings used window 128, and 44 used BFS+LLP. Zuckerli reused the selected WebGraph order and used window 128, unbounded reference chains, Huffman coding, and two rounds of reference selection. The artifact records the remaining settings.

\begin{figure}[h]
\centering
\pgfplotstableread{
graph TCover TWG TZK SCover SWG SZK
coPapersCiteseer 1.255478379 23.502383 65.513349298039 1.798224 4.209538 4.397736
RM07R 2.213063161 29.634116 102.412595131085 2.631145 3.716773 2.983795
audikw_1 7.084088005 75.322037 199.620856971946 8.781414 14.970320 13.548405
dielFilterV3real 4.446099833 68.132831 204.106801006012 3.890710 13.929894 13.143376
hollywood-2009 22.689758867 146.937405 1076.421605860000 16.076819 70.108407 65.678747
HV15R 33.003141110 257.710481 858.052512457012 36.924140 40.533288 34.823929
ri-ln-250m 34.950151832 666.906700 1847.583519758075 12.872246 466.229280 385.219796
ri-ln-500m 78.378429849 1464.846975 4229.999888039078 12.111099 963.034372 798.524164
ri-ln-1b 201.441849512 4316.804188 8487.402776294039 11.178347 1732.740941 1452.396743
ri-ln-2b 424.727955430 5620.526925 17033.632070293999 10.824355 3129.244090 2666.967429
ri-ln-4b 979.994566438 11529.994539 33960.492107174941 10.497386 5303.292045 4613.407993
ri-ln-8b 1756.161513743 16451.086615 51490.634534085053 9.115814 6610.796404 5837.492640
ri-pl-250m 29.882248356 412.057715 1572.357848971034 17.891768 200.294596 169.065329
ri-pl-500m 58.904773546 703.306981 2046.672273474978 16.869604 291.286393 245.922699
ri-pl-1b 136.569584073 1290.601007 3773.147258918965 16.141414 443.573766 378.011787
ri-pl-2b 276.484936506 2244.027589 6501.444875968038 15.399572 614.965361 528.914904
ri-pl-4b 560.964755132 3570.816925 10923.055831920006 14.591593 797.396330 690.137043
ri-pl-8b 1577.823438494 6108.446129 18680.434854961932 13.964531 1046.020902 909.357887
ri-u-250m 34.792895353 745.793726 2010.236672565923 12.283997 592.408792 488.042745
ri-u-500m 67.972209811 1460.585144 3970.515654029907 10.944202 1054.416355 873.288209
ri-u-1b 139.278106831 2744.654440 7869.458045335021 10.100436 1873.593741 1566.564454
ri-u-2b 375.444566699 5349.442436 15624.180828719051 9.789481 3281.833446 2781.765556
ri-u-4b 717.125547199 10596.491078 31562.621137585025 9.514944 5564.309880 4827.707362
ri-u-8b 1768.837903618 19953.019753 61816.682563964045 9.028095 8710.663430 7684.663028
scifi 10.434977596 64.765661 154.345294093946 0.121558 1.436012 1.299643
english 53.244989731 307.314196 882.473769228905 0.264776 14.200999 12.834855
tex 77.096380054 460.609141 2302.194200547994 0.516363 33.297453 30.303427
unix 103.184643311 542.452532 1660.333924592007 0.605797 43.644376 39.571461
superuser 1401.901339633 3517.641973 10740.597947858972 1.657011 234.459120 214.038192
mathoverflow-answers 1.440463733 17.825473 51.896963263978 0.263678 3.489000 3.233115
askubuntu-threads 2.066881265 21.543734 58.838276969967 0.839435 3.562404 3.382997
math-sx-threads 89.227953109 634.892648 1873.058721165988 3.887201 114.558788 100.895965
amazon-reviews 4157.396193850 9050.010286 27602.840449434938 252.057746 4366.804541 4007.668598
stackoverflow-threads 1691.329976043 7247.175354 20291.071090869955 81.357215 1820.368456 1585.951720
stackoverflow-answers 3122.483251206 10504.443444 27861.173929791083 74.063333 1824.404857 1609.461907
gamedev 13.011163003 81.350599 228.696889184997 0.161284 5.765441 5.238052
codereview 26.039682326 166.085881 497.478181876009 0.203641 9.106374 8.261519
wordpress 32.039806117 235.365244 754.879396808916 0.284681 19.573641 17.744792
gis 93.365192201 428.738361 1230.126945095020 0.433423 24.141092 21.844444
stats 98.200544925 566.705564 1790.549275918049 0.512258 54.437644 49.726990
serverfault 325.760519058 1355.698365 4184.658932029968 1.115841 117.901444 107.156428
codegolf 2.905678749 19.230744 48.276510720956 0.030269 0.294871 0.255781
rpg 10.275941047 59.036283 170.848594250041 0.100281 2.571965 2.319130
dba 70.681262858 316.316122 868.446668653050 0.242228 9.568397 8.572193
physics 77.463047090 485.879598 1538.121173348045 0.491921 43.674309 40.033082
drupal 103.469635760 490.361703 1432.747384638991 0.180094 11.288479 9.911936
askubuntu 411.816060470 1986.966112 5456.099849910010 1.201624 129.851037 118.255561
math 8956.290023533 17513.195024 51711.569129385985 2.981190 767.084862 700.478659
}\ConstructionTimeSizeScatter

\begin{tikzpicture}

\begin{groupplot}[
    group style={
        group size=2 by 2,
        horizontal sep=1.0cm,
        vertical sep=1.35cm,
    },
    width=0.46\linewidth,
    height=0.31\linewidth,
    xmode=log,
    ymode=log,
    tick align=outside,
    xmajorgrids=true,
    ymajorgrids=true,
    grid style={draw=gray!25},
]

\nextgroupplot[
    title={WebGraph},
    xmin=1e1,
    xmax=1e5,
    ymin=1e0,
    ymax=1e4,
    xtick={1e1,1e2,1e3,1e4,1e5},
    ytick={1e0,1e1,1e2,1e3,1e4},
]

\addplot[
    black!55,
    dashed,
    line width=0.8pt,
    domain=1e1:1e4,
    samples=2,
    forget plot
] {x};

\addplot[
    only marks,
    mark=*,
    mark size=1.7pt,
    color=OliveGreen,
] table[
    x=TWG,
    y=TCover
] {\ConstructionTimeSizeScatter};

\nextgroupplot[
    title={Zuckerli},
    xmin=1e1,
    xmax=1e5,
    ymin=1e0,
    ymax=1e4,
    xtick={1e1,1e2,1e3,1e4,1e5},
    yticklabels=\empty,
]

\addplot[
    black!55,
    dashed,
    line width=0.8pt,
    domain=1e1:1e4,
    samples=2,
    forget plot
] {x};

\addplot[
    only marks,
    mark=*,
    mark size=1.7pt,
    color=RoyalBlue,
] table[
    x=TZK,
    y=TCover
] {\ConstructionTimeSizeScatter};

\nextgroupplot[
    xmin=2e-2,
    xmax=1e4,
    ymin=2e-2,
    ymax=5e2,
    xtick={1e-2,1e-1,1e0,1e1,1e2,1e3,1e4},
    ytick={1e-2,1e-1,1e0,1e1,1e2},
]

\addplot[
    black!55,
    dashed,
    line width=0.8pt,
    domain=2e-2:5e2,
    samples=2,
    forget plot
] {x};

\addplot[
    only marks,
    mark=*,
    mark size=1.7pt,
    color=OliveGreen,
] table[
    x=SWG,
    y=SCover
] {\ConstructionTimeSizeScatter};

\nextgroupplot[
    xmin=2e-2,
    xmax=1e4,
    ymin=2e-2,
    ymax=5e2,
    xtick={1e-2,1e-1,1e0,1e1,1e2,1e3,1e4},
    yticklabels=\empty,
]

\addplot[
    black!55,
    dashed,
    line width=0.8pt,
    domain=2e-2:5e2,
    samples=2,
    forget plot
] {x};

\addplot[
    only marks,
    mark=*,
    mark size=1.7pt,
    color=RoyalBlue,
] table[
    x=SZK,
    y=SCover
] {\ConstructionTimeSizeScatter};

\end{groupplot}


\node[anchor=north]
at ($(group c1r1.south)!0.5!(group c2r1.south)+(0,-0.55cm)$)
{Compressed-adjacency construction time (s)};

\node[anchor=north]
at ($(group c1r2.south)!0.5!(group c2r2.south)+(0,-0.55cm)$)
{Compressed-adjacency encoding size (MB)};

\node[rotate=90, anchor=south]
at ($(group c1r1.west)+(-1.05cm,0)$)
{Cover construction time (s)};

\node[rotate=90, anchor=south]
at ($(group c1r2.west)+(-1.05cm,0)$)
{Cover encoding size (MB)};

\end{tikzpicture}
\caption{Construction time and encoding size on 48 graphs. Each point compares the VByte clique-cover representation with the WebGraph or Zuckerli encoding on one graph instance. The top row shows construction time, the bottom row shows encoding size, and the dashed diagonal marks equal values.}
\label{fig:con_scatter}
\end{figure}

\paragraph{Encoding size and construction time.} \fref{fig:con_scatter} compares the final VByte clique-cover representations with the final WebGraph and Zuckerli encodings. The VByte clique covers were smaller than WebGraph on all 48 graphs and smaller than Zuckerli on 47. The WebGraph-to-cover size ratio ranged from 1.10 to 965, with a median of 50.9 and a geometric mean of 41.5. The Zuckerli-to-cover ratio ranged from 0.94 to 851, with a median of 43.9 and a geometric mean of 36.7.

The clique-cover construction times include VByte encoding time. Clique covers were faster to construct on all 48 graphs. WebGraph took 2.0--21.5 times as long, with a median ratio of 7.2 and a geometric mean of 8.0. Zuckerli took 5.8--58.4 times as long, with a median of 24.8 and a geometric mean of 23.9. The WebGraph and Zuckerli measurements correspond to one encoding run after the vertex order and settings were chosen.

The final WebGraph search recorded six successful candidates per graph. Their combined encoding and index times were 6.3--103.3 times the clique-cover time, with a median ratio of 28.6 and a geometric mean of 31.5. This subtotal excludes preliminary and interrupted trials. Zuckerli reused the selected WebGraph order and used one fixed setting. Since clique-cover construction is faster on every graph, the running-time advantage of the cover-parsimonious implementations persists when representation-construction time is included.

\paragraph{Solution quality.} The cover-parsimonious and standard (compressed-adjacency) implementations returned similar solution values. For MVM and WVC, each vertex received an integer weight sampled uniformly from \inline{[1,n^2]}, with the same weights used across all implementations and runs. For MCM, their matching sizes differed by at most 2.2\% across the 48 graphs. The standard algorithm’s matchings were 0.34\% larger in geometric mean. 
For WVC, the standard algorithm returned a heavier cover on every graph, by 1.7\% in geometric mean and at most 10.0\%.

For MVM, every matching from the standard algorithm differed in weight from the cover-parsimonious result by less than 0.03\%, including valid partial matchings from the two unfinished graphs. Every reported MVM matching had weight at least 99.7\% of the total vertex weight, and hence every result attained at least 99.7\% of optimal weight.
\section{Additional Related Work}
\label{sec:related}

\subsection{Related Work on Clique Covers}
\label{subsec:related_cc}
We briefly review prior work on clique covers, spanning combinatorial formulations, applications, exact and parameterized algorithms, and polynomial-time constructions.

The problem of covering the edges of a graph with cliques is equivalent to several formulations arising in different settings, including intersection representations of graphs \cite{erdos1966representation} and the keyword conflict problem \cite{kou1978covering}. Many variants of this covering problem have been studied, all sharing the edge-covering aspect but differing in their optimization objectives. These include minimizing the number of cliques, minimizing the total number of vertex-clique assignments, and several weighted variants; see \cite{ullah2022computing} and the references therein.

Among these variants, the two most relevant to our work are minimizing the number of cliques and minimizing the total number of vertex-clique assignments, which we refer to as the cardinality-optimal and assignment-optimal objectives, respectively. The cardinality-optimal objective is NP-hard \cite{kou1978covering, orlin1977contentment}, even for planar graphs \cite{chang2001tree} and graphs of maximum degree six \cite{hoover1992complexity}. It is polynomial-time solvable for some restricted graph classes, including graphs of maximum degree five \cite{hoover1992complexity}, chordal graphs \cite{ma1989clique}, line graphs \cite{orlin1977contentment}, and certain generalizations of line graphs \cite{prisner1995clique}. Moreover, it is inapproximable within a factor of $n^\epsilon$, for some $\epsilon > 0$ unless $P = NP$ \cite{lund1994hardness}.

In contrast, the assignment-optimal objective has received much less attention \cite{ennis2012assignment, davoodi2016edge, ullah2021clique, ullah2022computing}. This objective is NP-hard \cite{ullah2022computing}. \hyref{Appendix}{subsec:apx_aopt} shows that it is inapproximable within a factor of $n^\epsilon$, for some $\epsilon > 0$ unless $P = NP$.

Earlier work on clique covers and their applications was surveyed in \cite{pullman2006clique, DBLP:journals/dam/Roberts85}. Exact and parameterized algorithms have since been studied for several clique-cover variants \cite{gramm2006data, hevia2023solving, ullah2022computing}. Some of these variants are also motivated by domain-specific applications; for recent examples, see \cite{helling2018constructing, markham2023neuro, ullah2022computing, wen2023w2sat, wen2025hyperplr, zhang2023fixed} and the references therein.

A preliminary exposition of the local-augmentation construction and the associated admissibility sets used in \hyref{Section}{subsec:concise_lzf} appeared in \cite{ullah2021clique, ullah2022computing}. Those works established that the resulting covers are composition-minimal and have total size \inline{\bigO{m}}, making them concise in the terminology of this paper.

\begin{table}[h]
    \centering    
    \caption{Reinterpretation of prior clique-cover construction algorithms through the minimality notions of this paper. The bounds under \inline{\ts{\C}} are tight. Only the bottom two algorithms guarantee composition-minimality, but their total-size bounds violate the \inline{\bigO{m}} requirement for conciseness.}
    \begin{tabular}{clllc}
    \toprule
    \textbf{Minimality} & \textbf{Construction} &
    \boldmath{\inline{\ts{\C}}} & \textbf{Runtime} & \textbf{Target Minimality?}\\
    \midrule
    
    \multirow{2}{*}{\textbf{Inclusion}}
    & \cite{kellerman1973determination} & \(\bigO{nm}\) & \(\bigO{nm^2}\) & \multirow{2}{*}{\textbf{\ding{55}}} \\
    & \cite{conte2020large} & \(\bigO{m}\) & \(\bigO{dm}\) &\\
    
    \midrule
    
    \multirow{2}{*}{\textbf{Support}}
    & \cite{kou1978covering} & \(\bigO{nm}\) & \(\bigO{nm^2}\)  & \multirow{2}{*}{\textbf{\ding{55}}}\\
    & \cite{gramm2006data} & \(\bigO{m}\) & \(\bigO{nm}\) &\\
    
    \midrule
    
    \multirow{2}{*}{\textbf{Composition}}
    & \cite{piepho2004algorithm} & \boldmath{\(3^{\bigO{n}}\)} & \(3^{\bigO{n}}\)  & \multirow{2}{*}{\textbf{\ding{51}}}\\
    & \cite{helling2018constructing} & \boldmath{\(\bigO{nm}\)} & \(\bigO{n^4}\) & \\
    \bottomrule
    \end{tabular}
    \label{tab:prior_algos}
\end{table}

We now reinterpret several prior clique-cover construction algorithms with correct running-time analyses through the minimality notions introduced in this paper. \hyref{Table}{tab:prior_algos} summarizes this reinterpretation with six prior construction algorithms most relevant to our setting, that is, they are intended to run fast, without guarantees of optimality. Five of them run in polynomial time, while one has exponential worst-case complexity.

The first construction is due to \cite{kellerman1973determination}, which produces inclusion-minimal covers. This algorithm processes vertices sequentially. For each current vertex, it first tries to insert the vertex into previously created cliques whenever possible, and then creates new cliques to cover the remaining uncovered incident edges. \cite{kou1978covering} added a post-processing step that removes every clique whose edges are all covered by the remaining cliques, thereby making the output support-minimal. 

\cite{gramm2006data} analyzed the construction of \cite{kellerman1973determination} together with the post-processing of \cite{kou1978covering} and improved the runtime. They also improved the size bound, but the resulting covers remain support-minimal. \cite{conte2020large} studied several construction algorithms under a common framework that repeatedly selects an uncovered edge and grows a clique around it. Their constructions ensure that no clique is contained in another, so the resulting covers are inclusion-minimal. \hyref{Table}{tab:prior_algos} includes the clean variant, which has a valid \inline{\bigO{dm}} running-time and \inline{\bigO{m}} total-size bound.

\cite{piepho2004algorithm} proposed a construction that starts with the single vertex set \inline{V} and processes the non-edges sequentially. For each non-edge \inline{\brcu{u,v}}, every current set containing both \inline{u} and \inline{v} is replaced by two copies, one with \inline{u} removed and the other with \inline{v} removed, followed by the deletion of absorbed sets. This produces composition-minimal covers. On the Moon--Moser family, the construction generates all maximal cliques, resulting in covers with \inline{3^{n/3}} cliques.

\cite{helling2018constructing} constructs a vertex-label assignment in which adjacent vertices share a label and non-adjacent vertices do not. Since each label induces a clique, obtaining clique covers from the label sets is straightforward. Their algorithm works by exploring maximal cliques and assigning the same label to all vertices in such a clique. Consequently, the resulting covers are composition-minimal but their size is \inline{\bigTheta{nm}} in the worst case.

All size bounds in \hyref{Table}{tab:prior_algos} are tight, meaning that there are explicit families where they are realized. In particular, the tightness of the \inline{\bigO{nm}} bounds follows from the graph family used in the proof of \lemref{lemma:sc_size}. Thus, none of the constructions listed in \hyref{Table}{tab:prior_algos} guarantees conciseness. The semi-clean variant of \cite{conte2020large} can be shown to produce concise covers. However, the running-time analysis in their Lemma~5 is incorrect. Their released implementation requires \inline{\bigO{m^2}} time even on graphs of constant degeneracy. Our tests with several more efficient implementations found that this approach is impractical on our graph instances.

\subsection{Related Work on Graph Compression}
\label{subsec:related_gc}

Prior work on lossless graph compression includes methods designed
to speed up graph algorithms, general-purpose representations supporting 
efficient graph queries, and methods exploiting the structure of
particular graph classes. The first two directions are most relevant
to our work. \hyref{Section}{subsec:cc_con_evals} compares the compression achieved by clique covers and biclique partitions, while \hyref{Section}{subsec:cp_algo_evals} compares cover-parsimonious algorithms with standard algorithms operating on compressed adjacency lists.

Feder and Motwani describe a graph-restructuring method that partitions the edges of a bipartite graph into bicliques and uses this partition to construct a path-preserving representation \cite{FederMotwani1995}. We refer to this as the FM method. The method also applies to general graphs through a bipartite transformation. For an undirected graph \inline{G=\br{V,E}}, each vertex \inline{u \in V} is replaced by two copies \inline{u_L} and \inline{u_R}, and each edge \inline{\brcu{u,v}} is represented by two directed edges \inline{\br{u_L, v_R}} and \inline{\br{v_L, u_R}}. In addition, an edge \inline{\br{u_R, u_L}} is added for each vertex \inline{u \in V}. Thus, an undirected graph with \inline{n} vertices and \inline{m} edges is transformed into a bipartite graph with \inline{2n} vertices and \inline{2m+n} edges.

From the \inline{2m} edges corresponding to the original edges, the FM method repeatedly extracts a biclique with sides \inline{U} and \inline{W}, removes the \inline{\card{U}\card{W}} edges between them, and introduces an intermediate vertex adjacent to all vertices in \inline{U \cup W}. The biclique is therefore represented with \inline{\card{U}+ \card{W}} edges instead of \inline{\card{U}\card{W}}. Repeating this process gives a tripartite representation that preserves the path information needed by the algorithms considered in \cite{FederMotwani1995}. 
For a bipartite graph with \inline{n} vertices on each side and \inline{m \geq 2n^{2-\delta}} edges, where \inline{\delta \in (0,1)} is fixed, let \inline{k\br{n,m,\delta}= \left\lfloor \delta \log n / \log \br{2n^2/m}\right\rfloor}. 
The FM method computes an edge-disjoint biclique partition whose total size \inline{\sum_{i} \br{\card{U_i}+ \card{W_i}}} is \inline{\bigO{m/k}} in \inline{\bigO{mn^\delta \log^2n}} time. Thus, after replacing each biclique by an intermediate star, the restructured graph has \inline{\bigO{m/k} = \bigO{\frac{m\log \br{2n^2/m}}{\delta \log n}}} edges. This gives a reduction by a factor of \inline{\bigTheta{\delta\log n}}  when \inline{m =\bigTheta{n^2}}. Then they adapt several path-based graph algorithms to exploit the representation, obtaining improved running times in sufficiently dense graphs.

Recent work proposes a variant of the FM method and claims to improve the construction time from \inline{\bigO{mn^\delta \log^2n}} to \inline{\bigO{mn^\delta}} while retaining the same \inline{\bigO{m/k}} total-size bound \cite{chavan2026speeding}. However, the analysis relies on an unproven assumption that each selected \inline{k}-vertex set has at least \inline{n^{1-\delta}} common neighbors, which does not hold for the proposed selection procedure in general. Subsequent work gives a different \inline{\bigO{m}}-time construction with the asymptotically optimal total-size bound \inline{\br{1/2+ \smallO{1}}H\br{p}n^2/\log n} for fixed density \inline{p} \cite{krapivin2026optimal}. However, for any fixed \inline{p}, the FM method already achieves the same asymptotic order up to a constant factor, so the improvement is primarily in the optimal constant and construction time rather than the asymptotic size reduction.

Similar FM-style biclique-to-star compression methods have been
proposed by Buehrer and Chellapilla, who use frequent itemset
mining to identify biclique-like patterns~\cite{buehrer2008scalable}.
Unlike FM restructuring, their method does not provide a provable
global bound on the compressed size. Karande, Chellapilla, and
Andersen build on this compression scheme and show that selected
web-graph algorithms can be executed directly on the compressed
graph~\cite{karande2009speeding}.

Before designing biclique analogues of our cover-parsimonious
algorithms, we investigated whether biclique representations could approach the compression achieved by our concise covers. The results in \hyref{Section}{subsec:cc_con_evals} show that FM is not competitive on our graph collection. Its partitions are several hundred times larger than our concise covers in geometric mean. Moreover, FM expands the representation on most instances, so scanning its biclique partition requires visiting more vertex entries than scanning the original adjacency lists. The compression results reported for other biclique-based methods gave us little reason to expect them to approach the compression achieved by our concise covers. We therefore did not pursue algorithms operating directly on biclique-based representations.

For the running-time comparison in \hyref{Section}{subsec:cp_algo_evals}, we therefore sought general-purpose representations that offer strong compression
while supporting access to individual adjacency lists. General-purpose graph compression schemes include WebGraph~\cite{boldi2004webgraph}, Zuckerli~\cite{versari2020zuckerli}, $k^2$-trees~\cite{brisaboa2009k2tree}, Log(Graph)~\cite{besta2018loggraph}, and Laconic~\cite{xu2024laconic}, among others~\cite{besta2018survey}. The BFS and DFS experiments of Versari et al.~\cite{versari2020zuckerli} illustrate the tradeoff between running time and compression. Both traversals ran faster on uncompressed adjacency lists, which incur no decoding cost, than on WebGraph or Zuckerli.
However, a running-time advantage obtained by using
substantially more space does not establish an advantage
at comparable compression levels.
We selected WebGraph and Zuckerli because of their state-of-the-art compression performance while supporting individual adjacency-list access.

WebGraph is a framework for queryable graph compression that retains the adjacency-list view of a graph \cite{boldi2004webgraph}. It compresses sorted adjacency lists by exploiting locality of vertex IDs and similarity among nearby lists. In particular, a node may refer to a previous adjacency list within a fixed window, and the current list is represented by specifying which portions of the reference list should be copied and which should be skipped, together with additional neighbors not obtained from the reference list. The copy pattern is encoded as alternating runs of copied and skipped entries. Among the additional neighbors, sufficiently long runs of consecutive values are grouped into intervals, while the remaining residual neighbors are gap-encoded using compact universal integer codes. To access a chosen adjacency list, WebGraph uses an index of its starting position in the compressed bitstream. Lazy iterators can then enumerate neighbors by combining copied entries, intervals, and residuals as needed, without fully materializing the referenced lists.

Zuckerli builds on the WebGraph compression framework while improving several components of its encoding scheme \cite{versari2020zuckerli}.
It retains WebGraph’s reference lists and alternating copy and skip blocks, but encodes the resulting integers using a hybrid representation and context-dependent entropy coding. Small integers are represented directly as symbols, while larger integers are represented by a symbol together with additional bits. The version supporting access to individual adjacency lists uses Huffman coding for these symbols. It replaces WebGraph’s explicit intervals with runs of zero residual gaps. A gap is zero when the next residual is the next vertex ID not already supplied by copied blocks. When many zero gaps occur consecutively, Zuckerli stores their count instead of encoding each zero separately. With unbounded reference chains, both WebGraph and Zuckerli choose a reference independently for each list, but Zuckerli refines the estimated encoding costs over multiple rounds. When chain length is bounded, Zuckerli additionally uses a procedure based on dynamic programming to improve reference choices while respecting the bound.

\end{document}